\newif\ifarxiv
\arxivtrue
\documentclass[a4paper,fleqn]{cas-sc}

\ifarxiv
\ExplSyntaxOn
\cs_gset:Npn \__first_footerline:
  { \group_begin: \small \sffamily \__short_authors: \group_end: }
\ExplSyntaxOff 
\fi

\ExplSyntaxOn
\keys_set:nn { stm / mktitle } { nologo }
\ExplSyntaxOff

\usepackage{hyperref}
\hypersetup{%
backref=true, 
pagebackref=true, 
hypertexnames=true,
colorlinks=true,citecolor=green!35!black,linkcolor=red!60!black%
} 
\usepackage{graphicx,dsfont,mathtools}
\usepackage{placeins}

\usepackage{amsmath, amsthm}

\usepackage{cleveref}

\newtheorem{observation}{Observation}
\newtheorem{proposition}{Proposition}

\newtheorem{clm}{Claim}
\newtheorem{fact}{Fact}
\newtheorem{definition}{Definition}
\newtheorem{theorem}{Theorem}
\newtheorem{remark}{Remark}

\crefname{table}{Table}{Tables}
\crefname{figure}{Figure}{Figures}
\crefname{theorem}{Theorem}{Theorems}
\crefname{corollary}{Corollary}{Corollaries}
\crefname{observation}{Observation}{Observations}
\crefname{lemma}{Lemma}{Lemmas}
\crefname{example}{Example}{Examples}
\crefname{reduction}{Reduction}{Reductions}
\crefname{construction}{Construction}{Constructions}
\crefname{subsection}{Section}{Sections}
\crefname{section}{Section}{Sections}
\crefname{claim}{Claim}{Claims}
\crefname{clm}{Claim}{Claims}
\crefname{proposition}{Proposition}{Propositions}
\crefname{definition}{Definition}{Definitions}
\crefname{remark}{Remark}{Remarks}
\crefname{fact}{Fact}{Facts}

\usepackage{thm-restate}

\newcommand{\appendixtitle}{Supplementary Material for the Paper ``Multidimensional Manhattan Preferences''}

\usepackage{multicol}

\newcommand{\close}{$\ast$}

\usepackage{xspace} 

\usepackage{comment}

\usepackage{xifthen}
\usepackage{subcaption}

\usepackage[numbers]{natbib}

\usepackage{paralist}

\usepackage{tikz}
\usepackage{pgfplots}

\usetikzlibrary{decorations,arrows,calc,topaths,backgrounds,shapes,positioning,fit,decorations.pathreplacing,patterns,patterns.meta,intersections,decorations.pathmorphing,matrix,pgfplots.fillbetween,fillbetween}
\usepackage[textsize=tiny,textwidth=1.5cm,linecolor=green!70!black, backgroundcolor=green!10, bordercolor=black]{todonotes}
\usepackage{tkz-euclide}

\definecolor{colorV}{rgb}{0.01,0.6,0.1}
\definecolor{colorU}{rgb}{0.75,0.48,0.07}
\definecolor{colorW}{RGB}{16,95,231}

\definecolor{colorA}{RGB}{53,53,233}
\definecolor{colorB}{RGB}{255,20,7}
\definecolor{colorC}{HTML}{035900}
\definecolor{colorD}{RGB}{113,0,85}
\definecolor{colorE}{RGB}{233,134,20}
\definecolor{colorF}{RGB}{33,172,197}
\definecolor{alterC}{RGB}{35,143,35}

\newcommand{\wred}[1]{{\color{red!40!black}#1}}
\newcommand{\db}[1]{{\color{blue!40!black}#1}}

\tikzstyle{voter} = [draw=blue!70!black, circle, fill=blue!70!black, inner sep=.8pt]
\tikzstyle{vA} = [draw=colorA, circle, fill=colorA, inner sep=.8pt]
\tikzstyle{vB} = [draw=colorB, circle, fill=colorB, inner sep=.8pt]
\tikzstyle{vC} = [draw=colorC, circle, fill=colorC, inner sep=.8pt]
\tikzstyle{vD} = [draw=colorD, circle, fill=colorD, inner sep=.8pt]
\tikzstyle{vE} = [draw=colorE, circle, fill=colorE, inner sep=.8pt]
\tikzstyle{vF} = [draw=colorF, circle, fill=colorF, inner sep=.8pt]

\tikzstyle{voterU} = [draw=colorU, circle, fill=colorU, inner sep=.8pt]
\tikzstyle{voterV} = [draw=colorV, circle, fill=colorV, inner sep=.8pt]
\tikzstyle{voterW} = [draw=colorW, circle, fill=colorW, inner sep=.8pt]

\tikzstyle{alter} = [draw=alterC, rectangle, fill=alterC, inner sep=.8pt]

\tikzstyle{reg} = [text=gray]
\tikzstyle{lines} = [draw=gray!50]
\tikzstyle{linesdark} = [draw=black!50]
\tikzstyle{lineA} = [thick, blue!60]
\tikzstyle{lineB} = [dashed, red!70!black]
\tikzstyle{bistyle} = [thick,green!80!black]

\tikzstyle{BBstyle} = [draw, darkgreen,opacity=0.9]
\tikzstyle{bisectorstyle} = [draw, darkgreen]

\tikzstyle{redpattern} = [pattern={Lines[
      distance=2mm,
      angle=-45,
      line width=0.25mm]}, pattern color = red, opacity = 0.5]

\tikzstyle{redpatterntilt} = [pattern={Lines[
      distance=2mm,
      angle=45,
      line width=0.25mm]}, pattern color = red, opacity = 0.5]

\tikzstyle{hide} = [draw=none, inner sep =0pt]

\makeatletter
\newcommand{\gettikzxy}[3]{%
  \tikz@scan@one@point\pgfutil@firstofone#1\relax
  \edef#2{\the\pgf@x}%
  \edef#3{\the\pgf@y}%
}
\tikzset{
    hlines/.style={%
      label cells,
      initialise hlines,
      append after command={%
        \pgfextra{\pgfmathtruncatemacro{\hline@cols}{\pgfmatrixcurrentcolumn - 1}}%
        \ifx\hline@rows\pgfutil@empty
        \else
        \foreach \hline@row in \hline@rows {
          \pgfextra{\edef\hline@temp{
          (\tikzlastnode-cell-\hline@row-1.north west) edge[\csname hline@row@\hline@row\endcsname] (\tikzlastnode-cell-\hline@row-\hline@cols.north east)}}
          \hline@temp
        }
        \fi
      }
    },
    initialise hlines/.code={
      \global\let\hline@rows=\pgfutil@empty
      \let\hline=\hline@inmatrix
    }
}

\newcommand\hline@inmatrix[1][]{%
    \ifx\hline@rows\pgfutil@empty
     \xdef\hline@rows{\the\pgfmatrixcurrentrow}%
    \else
     \xdef\hline@rows{\hline@rows,\the\pgfmatrixcurrentrow}%
    \fi
    \expandafter\xdef\csname hline@row@\the\pgfmatrixcurrentrow\endcsname{#1}%
}

\makeatother

\def \sx {.32}
\def \sy {.28}
\def \ofs {40}

\newcommand{\drawgridA}{
  \def \sx {.25}
  \def \sy {.28}
  \def \ofs {4}
    \foreach \x / \i in {
      0/1,3/2,6/3,9/4,12/5%
    }{
      \foreach \y / \j in {
        0/1,1.5/2,3/3,4.5/4,6/5%
    } {
        \node at (\x*\sx,\y*\sy) (\i\j) {};
      }
    }
}

\newcommand{\drawgridSP}{
  \def \sx {.25}
  \def \sy {.28}
  \def \ofs {4}
    \foreach \x / \i in {
      0/1,3/2,6/3,9/4,12/5,15/6%
    }{
      \foreach \y / \j in {
        0/1,1.5/2,3/3,4.5/4,6/5,7.5/6%
    } {
        \node at (\x*\sx,\y*\sy) (\i\j) {};
      }
    }
}

\newcommand{\drawgridB}{
    \foreach \x / \i in {
      -1/1,2/2,4/3,7/4,9/5%
    }{
      \foreach \y / \j in {
        0/1,1.5/2,4/3,5.5/4,8.5/5%
      } {
        \node at (\x*\sx,\y*\sy) (\i\j) {};
      }
    }
  }
  
  \newcommand{\drawgridU}{
    \foreach \x / \i in {
      -1/1,1/2,3.8/3,5.5/4,7/5,10/6,11/7%
    }{
      \foreach \y / \j in {
        -.5/1,1/2,3/3,5/4,7/5,9/6,11/7%
      } {
        \node[text=red,scale=0.7] at (\x*\sx,\y*\sy) (\i\j) {}; %
      }
    }
  }
  \newcommand{\drawregU}{
    
  \begin{pgfonlayer}{background}
    \foreach \s / \t in {12/72,13/73,14/74,15/75,16/76,21/27,31/37,41/47,51/57,61/67} {
      \path[draw,lines] (\s) edge (\t);
    }
  \end{pgfonlayer}
}

\newcommand{\drawregNN}{
    
  \begin{pgfonlayer}{background}
    \foreach \s / \t in {12/52,13/53,14/54,21/25,31/35,41/45} {
      \path[draw,lines] (\s) edge (\t);
    }
  \end{pgfonlayer}
}

\newcommand{\addN}{
    \foreach \r / \d in {5/4, 4/3,3/2,2/1} {
      \foreach \i / \j in {1/2,2/3,3/4,4/5} {
        \gettikzxy{(\i\r)}{\px}{\py}
        \gettikzxy{(\j\r)}{\qx}{\qy}
        \gettikzxy{(\i\d)}{\rx}{\ry}
        \node[reg, below right = -4pt and -4pt of \i\r] {
          \pgfmathparse{\i+4*(-\r+5)}%
          \footnotesize \pgfmathprintnumber{\pgfmathresult}};
      }
    }
}

\newcommand{\drawreg}{
    
  \begin{pgfonlayer}{background}
    \foreach \s / \t in {12/52,13/53,14/54,21/25,31/35,41/45} {
      \path[draw,lines] (\s) edge (\t);
    }

  \end{pgfonlayer}
}

\newcommand{\drawcircle}[3]{
  \gettikzxy{(#1)}{\xx}{\xy}
  \gettikzxy{(#2)}{\vx}{\vy}
  
  \pgfmathsetlengthmacro\disRX{abs(\vx-\xx)+abs(\vy-\xy)}
  \begin{scope}[xshift=\vx,yshift=\vy]
    \draw[#3] (-\disRX,0) -- (0,-\disRX) -- (\disRX,0) -- (0,\disRX) -- (-\disRX,0);
  \end{scope}
}

\newcommand{\drawcirclefill}[3]{
  \gettikzxy{(#1)}{\xx}{\xy}
  \gettikzxy{(#2)}{\vx}{\vy}
  
  \pgfmathsetlengthmacro\disRX{abs(\vx-\xx)+abs(\vy-\xy)}
  \begin{scope}[xshift=\vx,yshift=\vy]
    \filldraw[#3] (-\disRX,0) -- (0,-\disRX) -- (\disRX,0) -- (0,\disRX) -- (-\disRX,0);
  \end{scope}
}

\newcommand{\drawcircleD}[7]{
  \gettikzxy{(#1)}{\xx}{\xy}
  \gettikzxy{(#2)}{\vx}{\vy}
  \gettikzxy{(#4)}{\ux}{\uy}
  \gettikzxy{(#5)}{\wx}{\wy}
  \gettikzxy{(#6)}{\ax}{\ay}
  \gettikzxy{(#7)}{\bx}{\by}
  
  \pgfmathsetlengthmacro\disRX{abs(\vx-\xx)+abs(\vy-\xy)} %
  \pgfmathsetlengthmacro\lx{min(\ax,\disRX)}%
  \pgfmathsetlengthmacro\rx{\disRX+\wy}%

  \begin{scope}[xshift=\vx,yshift=\vy]
    \ifthenelse{\lengthtest{\ux > -\disRX}}{}{
      \ifthenelse{\lengthtest{\wy > -\disRX}}{%
        \draw[#3] (-\disRX, 0) -- (-\disRX-\wy, \wy);}{
        \draw[#3] (-\disRX, 0) -- (0, -\disRX);}}
    \ifthenelse{\lengthtest{\rx < \ax}}{
      \ifthenelse{\lengthtest{\wy > -\disRX}}{\draw[#3] (\disRX+\wy,\wy) -- (\lx,0);}
      {\draw[#3] (0,-\disRX) -- (\lx,0);}
    }
    {}
    \ifthenelse{\lengthtest{\disRX < \ax}}
    {\ifthenelse{\lengthtest{\by > \disRX}}
      {\draw[#3] (\disRX, 0) -- (0,\disRX);}
      {\draw[#3] (\disRX, 0) -- (\disRX-\by,\by);}
    }
    {\draw[#3] (\ax, \disRX-\ax) -- (0,\disRX);}
    \ifthenelse{\lengthtest{\by > \disRX}}
    {%
      \ifthenelse{\lengthtest{\ux > -\disRX}}
      {\draw[#3] (0,\disRX) -- (\ux,\ux+\disRX);}
      {\draw[#3] (0,\disRX) -- (-\disRX,0);}%
    }
    {}
  \end{scope}
}

\newcommand{\exPosi}{
  \foreach \x / \y / \n / \nn / \p / \dx / \dy / \t  / \c in
  {6/0/a1/{1}/{above right}/0/0/alter/black,
    0/6/a2/{2}/{above right}/0/0/alter/black, %
    0/0/a3/{3}/{above right}/0/0/alter/black, %
    6/5/v1/{v_1}/{below right}/0.1/0.1/vA/colorA, %
    6/1/v2/{v_2}/{below right}/0/0/vB/colorB, %
    5/6/v3/{v_3}/{below left}/0/0/vC/colorC, %
    1/6/v4/{v_4}/{below right}/0/0/vD/colorD, %
    2/1/v5/{v_5}/{above right}/0/0/vE/colorE, %
    1/2/v6/{v_6}/{above right}/0/-0.1/vF/colorF%
  } {
    \node[\t] at (\x,\y) (\n) {};
    \node[draw=none,\p = \dx and \dy of \n,fill=white, inner sep=.5pt, text=\c] {$\nn$};
    \node[\t] at (\x,\y) (\n) {};
  }    
}

\newcommand{\exPosiTmod}{
  \foreach \x / \y / \n / \nn / \p / \dx / \dy / \t  / \c in
  {%
    -2/2/a1/{1}/{above left}/-0.1/0.05/alter/black, %
    -3/5/a2/{2}/{above left}/0/0/alter/black, %
    1/3/a3/{3}/{above right}/0.1/-0.1/alter/black, 
    3/3/a4/{4}/{above right}/-0.1/0/alter/black, %
    1/7/a5/{5}/{above right}/0/0/alter/black,
    -5/0/v1/{v_1}/{below left}/0.1/0/vA/colorA, %
    5/-0/v2/{v_2}/{below right}/0/0/vB/colorB%
  } {
    \node[\t] at (\x,\y) (\n) {};
    \node[draw=none,circle,\p = \dx and \dy of \n,fill=white, inner sep=.1pt, text=\c] {$\nn$};
    \node[\t] at (\x,\y) (\n) {};
  }
}

\usepackage{cleveref}
\newtheorem{example}{Example}
\newtheorem{lemma}{Lemma}
\crefname{table}{Table}{Tables}
\crefname{figure}{Figure}{Figures}
\crefname{theorem}{Theorem}{Theorems}
\crefname{corollary}{Corollary}{Corollaries}
\crefname{observation}{Observation}{Observations}
\crefname{lemma}{Lemma}{Lemmas}
\crefname{example}{Example}{Examples}
\crefname{reduction}{Reduction}{Reductions}
\crefname{construction}{Construction}{Constructions}
\crefname{subsection}{Section}{Sections}
\crefname{section}{Section}{Sections}
\crefname{claim}{Claim}{Claims}
\crefname{clm}{Claim}{Claims}
\crefname{proposition}{Proposition}{Propositions}
\crefname{appendix}{Appendix}{Appendices}


\crefname{definition}{Definition}{Definitions}

\newcommand{\rz}{\ensuremath{\mathds{R}}}
\newcommand{\dspace}{\ensuremath{\mathds{R}^{d}}}

\newcommand{\ppp}{{\cal P}}

\newcommand{\vvv}{{\cal V}}

\newcommand{\aaa}{{\cal A}}
\newcommand{\rrr}{{\cal R}}
\newcommand{\sss}{{\cal S}}
\newcommand{\RR}{\mathds{R}}

\newcommand{\pa}{\ensuremath{\vect{a}}}
\newcommand{\pb}{\ensuremath{\vect{b}}}
\newcommand{\pc}{{\ensuremath{\color{blue}\vect{c}}}}
\newcommand{\pd}{\ensuremath{\vect{d}}}

\newcommand{\px}{\ensuremath{\vect{x}}}
\newcommand{\py}{\ensuremath{\vect{y}}}
\newcommand{\pz}{\ensuremath{\vect{z}}}
\newcommand{\ps}{\ensuremath{\vect{s}}}
\newcommand{\pt}{\ensuremath{\vect{t}}}
\newcommand{\pr}{\ensuremath{\vect{r}}}
\newcommand{\pp}{\ensuremath{\vect{p}}}
\newcommand{\pq}{\ensuremath{\vect{q}}}
\newcommand{\pu}{\ensuremath{\vect{u}}}
\newcommand{\pv}{{\ensuremath{\color{red!50!black}\vect{v}}}}
\newcommand{\pw}{\ensuremath{\vect{w}}}

\newcommand{\dNE}{\ensuremath{\mathsf{NE}}}
\newcommand{\dSE}{\ensuremath{\mathsf{SE}}}
\newcommand{\dNW}{\ensuremath{\mathsf{NW}}}
\newcommand{\dSW}{\ensuremath{\mathsf{SW}}}

\newcommand{\myemph}[1]{{\color{green!30!black}\emph{#1}}}

\definecolor{darkgreen}{HTML}{035900}
\definecolor{darkblue}{rgb}{0,0,0.4}
\definecolor{winered}{rgb}{0.6,0.1,0.1}
\definecolor{lightblue}{rgb}{0.527,0.805,0.977}

\newcommand{\worstinconsistentconfig}{all-triples worst-diverse configuration\xspace}
\newcommand{\Worstinconsistentconfig}{All-triples worst-diverse configuration\xspace}

\newcommand{\dde}[1][]{{\ifthenelse{\equal{#1}{}}{$d$}{$#1$}-dimensional Euclidean}\xspace}
\newcommand{\dder}[1][]{{\ifthenelse{\equal{#1}{}}{$d$}{$#1$}-dimensional Euclidean representation}\xspace}

\newcommand{\dEuclid}[1][]{{\ifthenelse{\equal{#1}{}}{$d$}{$#1$}-{\color{blue!40!black}Euclidean}}\xspace}
\newcommand{\dManhattan}[1][]{{\ifthenelse{\equal{#1}{}}{$d$}{$#1$}-{\color{red!40!black}Manhattan}}\xspace}
\newcommand{\dMax}[1][]{{\ifthenelse{\equal{#1}{}}{$d$}{$#1$}-{\color{green!40!black}Max}}\xspace}

\newcommand{\dEuclidness}[1][]{{\ifthenelse{\equal{#1}{}}{$d$}{$#1$}-{\color{blue!40!black}Euclideanness}}\xspace}
\newcommand{\dManhattanness}[1][]{{\ifthenelse{\equal{#1}{}}{$d$}{$#1$}-{\color{red!40!black}Manhattanness}}\xspace}
\newcommand{\dMaxness}[1][]{{\ifthenelse{\equal{#1}{}}{$d$}{$#1$}-{\color{green!40!black}Maxness}}\xspace}

\newcommand{\dSP}[1][]{{\ifthenelse{\equal{#1}{}}{$d$}{$#1$}-{\color{orange!40!black}dimensional single-peaked}}\xspace}
\newcommand{\SP}[1][]{{{\color{orange!40!black}single-peaked}}\xspace}
\newcommand{\SC}[1][]{{{\color{orange!40!black}single-crossing}}\xspace}
\newcommand{\dSPh}[1][]{{\ifthenelse{\equal{#1}{}}{$d$}{$#1$}-{\color{orange!40!black}dimensional single-peaked}}\xspace}
\newcommand{\SPness}[1][]{{{\color{orange!40!black}single-peakedness}}\xspace}
\newcommand{\cSPness}[1][]{{{\color{orange!40!black}Single-peakedness}}\xspace}
\newcommand{\SCness}[1][]{{{\color{orange!40!black}single-crossingness}}\xspace}
\newcommand{\dSPhness}[1][]{{\ifthenelse{\equal{#1}{}}{$d$}{$#1$}-{\color{orange!40!black}dimensional single-peakedness}}\xspace}

\newcommand{\ax}{\triangleright}
\newcommand{\axb}{\overrightarrow{\ax}}

\newcommand{\dis}[1]{\ensuremath{\|#1 \|}}
\newcommand{\Edis}[1]{\ensuremath{\|#1 \|_2}}
\newcommand{\Mdis}[1]{\ensuremath{\|#1 \|_1}}
\newcommand{\Maxdis}[1]{\ensuremath{\|#1 \|_{\infty}}}

\newcommand{\pref}{\ensuremath{\succ}}
\newcommand{\prefeq}{\ensuremath{\succeq}}
\newcommand{\ETR}{\ensuremath{\exists\mathds{R}}}

\newcommand{\rank}{\ensuremath{\mathsf{rk}}}

\DeclareMathOperator*{\argmax}{arg\,max}
\DeclareMathOperator*{\argmin}{arg\,min}

\newcommand{\Mhalfspace}{\ensuremath{\mathsf{H}}}
\newcommand{\vect}[1]{\ensuremath{\boldsymbol{#1}}}

\newcommand{\BB}{\ensuremath{\mathsf{BB}}}

\newcommand{\bet}{\textsf{BE}}
\newcommand{\ext}{\textsf{EX}}

\newcounter{myprofilecounter}
\newcommand{\mypar}[1]{\smallskip
  \noindent\textbf{#1}}
\newcounter{vtwocounter}
\newcounter{vthreecounter}
\newcounter{vfourcounter}
\newcounter{vfivecounter}
\newcounter{vsixcounter}
\newcounter{vsevencounter}
\newcounter{veightcounter}
\newcounter{betcounter}
\newcounter{extcounter}
\newcommand{\basicfigfivepoints}{
	\foreach \x / \y / \n in {0/0/r,3/1/w} {
	\node[voter,color=red] at (\x , \y) (\n) {};
	\node[above left = -2 pt and -2 pt of \n] {$\n$};
}

}

\usepackage{etoolbox} %

\newcommand{\appsymb}{$\star$}

\newcommand{\toappendix}[1]{%
  \gappto{\appendixtext}{
    {#1}
   }
}

\newcommand{\appendixproofwithstatement}[3]{%
  \gappto{\appendixtext}{
    \subsection{Proof of \cref{#1}}\label{proof:#1}
    #2 #3
  }
}

\newcommand{\appendixsection}[1]{%
  \gappto{\appendixtext}{
    \section{Additional Material for Section~\ref{#1}}
    \label{appsec:#1}
  }
}

\newcommand{\clmmManhattan}{%
  For each voter~$v_i\in \vvv$ and each alternative~$j\in \aaa$, we have
    \begin{align*}
      \Mdis{E(v_i)-E(j)} = \begin{cases}
        \Mdis{E(v_i)} + 2(m-E(v_i)[j]), & \text{ if } j \neq m,\\
        \Mdis{E(v_i)}, & \text{ otherwise. }
        \end{cases}
    \end{align*}
  }

\begin{document}

\title[mode=title]{Multidimensional Manhattan Preferences}
\shorttitle{Multidimensional Manhattan Preferences}

\author[l1]{Jiehua Chen}
\ead{jiehua.chen@tuwien.ac.at}
\author[l1]{Martin N{\"o}llenburg}
\ead{martin.noellenburg@tuwien.ac.at}
\author[l1]{Sofia Simola}
\ead{sofia.simola@tuwien.ac.at}
\author[l1]{Ana{\"i}s Villedieu}
\ead{anais.villedieu@tuwien.ac.at}
\author[l1]{Markus Wallinger}
\ead{markus.wallinger@tuwien.ac.at}

\shortauthors{Chen et~al.}

\affiliation[l1]{organization={TU Wien},%
            city={Vienna},
            country={Austria}}

\maketitle

\begin{abstract}
  A preference profile (i.e., a collection of linear preference orders of the voters over a set of alternatives) with $m$~alternatives and $n$~voters is \dManhattan (resp.\ \dEuclid) 
  if both the alternatives and the voters can be placed into a $d$-dimensional space 
  such that between each pair of alternatives, every voter prefers the one 
  which has a shorter Manhattan (resp.\ Euclidean) distance to the voter.

  We study how \dManhattan preference profiles depend on the values~$m$ and $n$.
  First, we provide explicit constructions to show that each preference profile with $m$ alternatives and $n$ voters is \dManhattan{} whenever $d \ge \min(n, m-1)$.
  We further extend this positive result for other $p$-norms with $p \in \mathds{R}_{\geq 1} \cup \{\infty\}$.

  Second, for $d=2$, we develop \emph{forbidden substructures}---preference patterns among small sets of voters that constrain any \dManhattan[2] embedding---and use them to
  show that the smallest non-\dManhattan[2] preference profile has either $3$ voters and $6$ alternatives, or $4$ voters and $5$ alternatives, or $5$ voters and $4$ alternatives.
  This is more complex than the case with \dEuclid{} preferences~(see (Bogomolnaia and Laslier, 2007) and (Bulteau and Chen, 2022)).
 
 We also show that \dManhattan{} preferences imply $2^{d-1}$-dimensional single-peakedness, while \dManhattanness[2] is incomparable with single-peakedness and single-crossingness.
  
\end{abstract}

\section{Introduction}\label{sec:intro}
Modeling voters' linear preferences (aka.\ rankings) over a set of alternatives as geometric distances is an approach popular in many research fields such as economics~\citep{Hotelling1929,Downs1957,Eguia2011},
political and social sciences~\citep{Sto1963,Poo1989,Ene1990,BoLa2006},
and psychology~\citep{Coombs1964,BorGroMai2018}.
The idea is to consider the alternatives and voters as points in a $d$-dimensional space such that 
\noindent\begin{align*}%
  \text{for each two alternatives, each voter prefers the one that is \emph{closer} to her.}\tag{\close}\label{eq:closeness}
\end{align*}%
If the proximity is measured via the Euclidean distance, then \myemph{preference profiles} (i.e., a collection of distinct linear preference orders specifying voters' preferences) obeying \eqref{eq:closeness} are called \emph{\dEuclid}.
While the \dEuclid{} model seems to be canonical, in real life the shortest path between two points may be Manhattan rather than Euclidean.
For instance, in urban geography, the alternatives (e.g., a shop or a supermarket) and the voters (e.g., individuals) are often located on grid-like streets.
That is, the distance between an alternative and a voter is more likely to be measured according to the Manhattan distance (aka.\ Taxicab distance or $1$-norm-distance), i.e., the sum of the absolute differences of the coordinates of the alternative and the voter.
Similarly to the Euclidean preference notion, we call a preference profile \emph{\dManhattan} if there exists an embedding of the voters and the alternatives which satisfies condition~\eqref{eq:closeness} under the Manhattan distance. %
Indeed, Manhattan preferences have been studied for a wide range of applications such as facility location~\citep{LarSad1983,SuiBoutilier2015adt}, group decision making~\citep{SSL2007}, and voting and committee elections~\citep{EckKla2010}. Many voting advice applications, such as the German Wahl-O-Mat~\citep{wahlomat} and Finnish Ylen vaalikone~\citep{finylen} use Manhattan distances to measure the distance between a voter and alternative, indicating that such distances may be perceived as more natural in human decision making.

Despite their practical relevance, Manhattan preferences have attracted far less attention than their close relative Euclidean preferences. Bogomolnaia and Laslier~\cite{BoLa2006} studied how restrictive the assumption of Euclidean preferences is.
They showed that for every~$n$, $m$, and $d$,
  every preference profile with $m$ alternatives and $n$ voters, and with possibly indifferent preferences, is \dEuclid{} if and only if $d\ge \min(n,m-1)$.\footnote{In fact, their proof for showing that every profile with $n$ voters is \dEuclid[n] does not work for indifferent preferences; see \cref{ex:BoLaBug}.
    However, it is fairly straightforward to fix it.
    We provide such a fix in \cref{prop:BoLaFix} since it also works for preferences under other $p$-norms.}
For $d=1$, their smallest non-\dEuclid[1] preference profile with strict preferences consists of either $3$ voters and $3$ alternatives or $2$ voters and $4$~alternatives, which is tight according to Chen and Grottke~\cite{ChenGrottke2021}. 
For $d=2$, their smallest non-\dEuclid[2] profile consists of either $4$ voters and $4$ alternatives or $3$ voters and $8$~alternatives, which is also tight by \citet{BulChe2022}.
To the best of our knowledge, 
no analogous characterization of \dManhattan{} preferences exists.

\citet{Bennett_Hays_1960},\cite{Hays_Bennett_1961} study maximally \dEuclid{} profiles.
They show that a \dEuclid[2] preference profile with four alternatives can contain up to $18$ distinct preference orders and offer a general result for $d$ dimensions and $m$ alternatives.
Recently, 
Escoffier et al.~\cite{EST2021Euclidlp}\footnote{The work of Escoffier et al.\ and ours were carried out independently and concurrently. A conference version of our paper appeared at LATIN 2022, and our preprint~\cite{ourArxiv} appeared on arXiv in January 2022; their preprint appeared on arXiv in February 2022.} show that a
\dManhattan[2] preference profile for four alternatives can contain up to $19$ distinct preference orders.

From the computational point of view, it is known that for $d=1$,
deciding whether a given preference profile is Euclidean (and hence Manhattan) can be done in polynomial time~\citep{DoiFal1994,Knoblauch2010,ElkFal2014}.
For each fixed~$d \ge 2$, however, 
testing Euclidean preferences is complete for the complexity class \myemph{existential theory of the reals~$\ETR$},
while it is straightforward to see that the problem for the Manhattan case is contained in NP~\citep{Peters2017}; note that NP{\,}$\subseteq\!\ETR$. See \cite{Schaefer2010ETR} for more information on the complexity class~$\ETR$.
Nothing about the complexity lower bound is known for Manhattan preferences.

\paragraph{Our contribution.}
In this paper, we study how to find a \dManhattan{} embedding for a given preference profile and what is the smallest dimension for such an embedding.

First, we prove that, similarly to the Euclidean case, every preference profile with $m$ alternatives and $n$ voters is \dManhattan\ if $d\ge \min(m-1,n)$ (\cref{thm:d=n->dManhattan,thm:d=m-1->Manhattan}).
This extends for other $p$-norms as well for $p \in \mathds{R}_{\geq 1} \cup \{\infty\}$; see appendix.

Our main technical contribution lies in developing \emph{forbidden substructures} for \dManhattan[2] embeddings.
Specifically, we introduce the \bet-configuration (\cref{def:3voters-forbidden-profiles-B}) and the \ext-configuration (\cref{def:3voters-forbidden-profiles-E}), which describe preference patterns among three voters that restrict how voters can be placed relative to one another in any \dManhattan[2] embedding.
We prove that all \dManhattan[2] embeddings must respect these constraints (\cref{lem:bet-property,lem:ext-property}), and we additionally identify the \worstinconsistentconfig (\cref{def:threeworst}), which interacts with the geometry of bounding boxes in two dimensions to preclude \dMax[2] embeddability.
These forbidden substructures are the key tool in our proofs of the non-embeddability results below, and we believe they are of independent interest:
they constitute the first forbidden subprofile characterization results for \dManhattan[2] preferences, and they may serve as building blocks for future complexity-theoretic results, e.g., for constructing gadgets in NP-hardness reductions.

Using these forbidden substructures, we determine tight bounds on the smallest non-\dManhattan[2] profile.
We show that an arbitrary preference profile with $n$ voters and $m$ alternatives is \dManhattan[2] if and only if
either $m\le 3$ (\cref{thm:d=m-1->Manhattan,thm:no-n5-m4}),
or $n\le 2$ (\cref{thm:d=n->dManhattan}), 
or $n \le 3$ and $m\le 5$ (\cref{thm:no-n3-m6} and \cref{prop:n3-m5+n4-m4}),
or $n\le 4$ and $m\le 4$ (\cref{thm:no-n4-m5} and \cref{prop:n3-m5+n4-m4}).
Note that this is considerably different from the Euclidean case:
There exists a non-\dEuclid[2] preference profile with $n=4$ and $m=4$,
while every preference profile with $n\le 3$ and $m\le 7$ is \dEuclid[2]. 
The ``if'' part is verified computationally.
See \cref{fig:overview} for a summary for~$d=2$.

We also study the relationship between \dManhattan[2] preferences and \SP and/or \SC preferences. 
\cSPness\ and \SCness\ are well-studied restricted preference structures, see \cref{def:SP,def:SC} from \cref{sec:other_domains}.
Our finding is that \dManhattan[2] preferences and the other two preference structures are in general incomparable.

\begin{figure}
  \centering
  \begin{tikzpicture}[black, draw=black!70, xscale=0.25, yscale=-0.25]
    \node[text width=1.4cm] at (-0.5,0.5){No\\ instances}; 
    \foreach \nc in {3,...,8} \node[] at (\nc*2,1.5) {$\nc$};
    \node at (4.5,1.5) {$\le$};
    \foreach \nv in {2,...,5} \node[] at (2,\nv*2+0.5) {$\nv$};
    \node at (0.9,4+0.5) {$\le$};
    \node[] at (11,-0.5){\#alternatives $(m)$};
    \node[] at (2,12.2) {$\vdots$}; 
    \node[] at (18,1.5) {$\cdots$};
    \node[rotate=90] at (-1.7,8) {\#voters $(n)$};
    \draw (3,-1.8)--(3,13);
    \draw (-3.4,3)--(19,3);

    \node[] at (8,10) {\wred{$\times$}};
    \node[] at (10,8) {\wred{$\times$}};
    \node[] at (12,6) {\wred{$\times$}};

    \tikzset{ every path/.style={draw=red!50!black,  line width=1pt}}
     \draw (7,13) -- (7,9) -- (9,9) -- (9,7) -- (11, 7) -- (11,5) -- (19,5);

    \node[] at (8,8) {\db{$\bullet$}};
    \node[] at (16,6) {\db{$\bullet$}};

    \tikzset{ every path/.style={draw=blue!40!black,dashed, line width=1.4pt}}
     \draw (7,13) -- (7,7) -- (15,7) -- (15,5) -- (19, 5);

    \begin{pgfonlayer}{background}
      \filldraw[fill=blue!20,draw=none] (3,13) -- (7,13) -- (7,7) -- (15,7) -- (15,5) -- (19, 5) -- (19,3) -- (3,3) -- cycle;
    \filldraw[draw=none,pattern={Lines[
      distance=2mm,
      angle=45,
      line width=0.1mm]}, pattern color=red!90!black] (3,13) -- (7,13) -- (7,9) -- (9,9) -- (9,7) -- (11, 7) -- (11,5) -- (19,5) -- (19,3) -- (3,3) -- cycle;
     \end{pgfonlayer}

\end{tikzpicture} 
\caption{Boundaries of non-\dEuclid[2] (resp.\ non-\dManhattan[2]) profiles with a given number of voters and alternatives. Each \db{blue} bullet (resp.\ \wred{red} cross) represents the existence of such a non-\dEuclid[2] (resp. non-\dManhattan[2]) profile.}\label{fig:overview}
\end{figure}

\paragraph{Paper structure.}
The paper is organized as follows:
\cref{sec:defi} introduces necessary definitions and notations. 
In \cref{sec:Manhattan-positive} we show that every profile with $m$ alternatives and $n$ voters is \dManhattan{} whenever $d \ge \min(n, m - 1)$.
These results extend for an arbitrary $\ell_p$ norm for every $p \in \mathds{R}_{\geq 1} \cup \{\infty\}$.
In \cref{sec:vconfigs} we develop our forbidden substructures---the \bet-configuration, the \ext-configuration, and the \worstinconsistentconfig---and prove that they constrain \dManhattan[2] embeddings. These are the central technical tools of the paper and we believe them to be of independent interest for future research on recognizing \dManhattan[2] profiles.
In \cref{sec:Manhattan-negative},
we apply these forbidden substructures to prove that our smallest non-\dManhattan[2] profiles are indeed not \dManhattan[2], and we verify via a computer program that all strictly smaller profiles are \dManhattan[2], yielding a tight characterization.
In \cref{sec:other_domains} we discuss the relation between \dManhattan{} preferences and other restricted preference structures.
We conclude with future research directions in \cref{sec:conclude}.
For a better presentation, proofs of the results and additional materials marked with (\appsymb) are deferred to the appendix.

\section{Preliminaries}\label{sec:defi}%

Given a non-negative integer~$t$, we use \myemph{$[t]$} to denote the set~$\{1,\dots,t\}$.
Let~$\px$ denote a vector of length~$d$ or a point in a $d$-dimensional space, and
let $i$ denote an index $i\in [d]$.
We use \myemph{$\px[i]$} to refer to the $i^{\text{th}}$~value in~$\px$.

Let $\aaa\coloneqq [m]$ be a set of alternatives. %
A \myemph{preference order}~$\pref$ of $\aaa$ is a linear order (a.k.a.\ permutation or ranking) of~$\aaa$; a linear order is a binary relation which is total, irreflexive, and transitive.
For two distinct alternatives~$a$ and $b$, the relation \myemph{$a\pref b$} means that $a$ is preferred to (or in other words, ranked higher than) $b$ in~$\pref$.
An alternative~$c$ is the \myemph{most-preferred} alternative in~$\pref$
if for each alternative~$b\in \aaa \setminus \{c\}$ it holds that $c \pref b$.
Let $\pref$ be a preference order over~$\aaa$. 
For a subset~$B\subseteq \aaa$ of alternatives and an alternative~$c$ not in $B$, 
we use \myemph{$B\pref c$} (resp.\ \myemph{$c \pref B$}) to denote that
for each $b\in B$ it holds that $b\pref c$ (resp.\ $c\pref b$).
A \myemph{preference profile} (or \myemph{profile} in short) $\ppp$ specifies the preference orders of a number of voters over a set of alternatives.
Formally, \myemph{$\ppp \coloneqq (\aaa, \vvv, \rrr)$},
where $\aaa$ denotes the set of $m$ alternatives,
$\vvv$ denotes the set of $n$~voters,
and $\rrr \coloneqq (\pref_1, \dots, \pref_n)$ is a collection of $n$ preference orders
such that each voter~$v_i\in \vvv$ ranks the alternatives according to the preference order~$\pref_i$ on~$\aaa$.
We may omit the subscript~$i$ from~$\pref_i$ if it is clear from the context. 
Throughout the paper, if not explicitly stated otherwise, we assume~$\ppp$ is a preference profile of the form~$(\aaa,\vvv,\rrr)$.
For notational convenience, for each alternative~$a\in \aaa$ and each voter~$v_i\in \vvv$,
let \myemph{$\rank_{i}(a)$} denote the rank of alternative~$a$ in the preference order~$\pref_i$,
which is the number of alternatives that are preferred to~$a$ by voter~$v_i$, i.e.,
$\rank_{i}(a)=|\{b \in \aaa \mid b\pref_i a\}|$.
For instance, if voter~$v_i$ has preference order~$2 \succ_i 3 \succ_i 1 \succ_i 4$, then $\rank_i(3) = 1$.

Given a $d$-dimensional vector~$\px \in \dspace$ and an $p$-norm with $p \in \mathds{R}_{\geq 1}$,
let \myemph{$\dis{\px}_{p}$} denote the $p$-norm of~$\px$,
i.e., $\dis{\px}_{p} = (| \px[1]| ^p+\dots+| \px[d]| ^p)^{1/p}$,
and let \myemph{$\dis{\px}_{\infty}$} denote the $\infty$-norm of~$\px$,
i.e., $\dis{\px}_{p} = \max\{\px[i]\}_{i\in [d]}$.
Given two points~$\pu, \pw$ in~$\dspace$ and $p \in \mathds{R}_{\geq 1} \cup \{\infty\}$, we use the $p$-norm of $\pu-\pw$, i.e., \myemph{$\dis{\pu-\pw}_{p}$}, to denote the $\ell_p$-distance of $\pu$ and $\pw$.
By convention, we use \myemph{Manhattan}, \myemph{Euclidean}, and \myemph{Max} distances to refer to $\ell_1$-, $\ell_2$-, and $\ell_{\infty}$-distances, respectively.

\begin{figure}
	\def \xs {.3}
	\def \xy {.2}
	\centering
	\begin{tikzpicture}[black]
		\foreach \x / \y / \n / \nn / \typ / \p / \dx / \dy in {-4/0/u/u/voterU/{below left}/0/-4, 6/0/v/v/voterW/{above right}/-1/-1} {
			\node[\typ] at (\x*\xs,\y*\xy) (\n) {};
			\node[\p = \dx pt and \dy pt of \n] {$\nn$};
		}

		\coordinate (y) at (1*\xs,2*\xy);
		\coordinate (x) at (1*\xs,-2*\xy);

		\gettikzxy{(v)}{\vx}{\vy}
		\gettikzxy{(u)}{\ux}{\uy}
		\gettikzxy{(x)}{\xx}{\xy}
		\gettikzxy{(y)}{\yx}{\yy}

		\draw[gray,dashed] (v) -- (\ux, \vy);
		\draw[gray,dashed] (u) -- (\ux, \vy);
		\draw[gray,dashed] (v) -- (\vx, \uy);
		\draw[gray,dashed] (u) -- (\vx, \uy);

		\draw[bistyle] (\xx,\xy-20) -- (x) -- (y) -- (\yx,\yy+20);
		
	\end{tikzpicture}
	\quad~\quad
	\begin{tikzpicture}[black]
		\foreach \x / \y / \n / \nn / \typ / \p / \dx / \dy in {-4/-2/u/u/voterU/{below left}/-1/-4, 6/2/v/v/voterW/{above right}/-1/-1} {
			\node[\typ] at (\x*\xs,\y*\xy) (\n) {};
			\node[\p = \dx pt and \dy pt of \n] {$\nn$};
		}

		\coordinate (y) at (-1*\xs,2*\xy);
		\coordinate (x) at (3*\xs,-2*\xy);

		\gettikzxy{(v)}{\vx}{\vy}
		\gettikzxy{(u)}{\ux}{\uy}
		\gettikzxy{(x)}{\xx}{\xy}
		\gettikzxy{(y)}{\yx}{\yy}

		\draw[gray,dashed] (v) -- (\ux, \vy);
		\draw[gray,dashed] (u) -- (\ux, \vy);
		\draw[gray,dashed] (v) -- (\vx, \uy);
		\draw[gray,dashed] (u) -- (\vx, \uy);

		\draw[bistyle] (\xx,\xy-20) -- (x) -- (y) -- (\yx,\yy+20);
		
	\end{tikzpicture}
	\quad~\quad
	\begin{tikzpicture}[black]
		\def \xs {.2}
		\def \xy {.2}
		\foreach \x / \y / \n / \nn / \typ / \p / \dx / \dy in {-4/-2/u/u/voterU/{below left}/-1/-4, 0/2/v/v/voterW/{above right}/-1/-1} {
			\node[\typ] at (\x*\xs,\y*\xy) (\n) {};
			\node[\p = \dx pt and \dy pt of \n] {$\nn$};
		}

		\coordinate (y) at (-4*\xs,2*\xy);
		\coordinate (x) at (0*\xs,-2*\xy);

		\gettikzxy{(v)}{\vx}{\vy}
		\gettikzxy{(u)}{\ux}{\uy}
		\gettikzxy{(x)}{\xx}{\xy}
		\gettikzxy{(y)}{\yx}{\yy}

		\draw[gray,dashed] (v) -- (\ux, \vy);
		\draw[gray,dashed] (u) -- (\ux, \vy);
		\draw[gray,dashed] (v) -- (\vx, \uy);
		\draw[gray,dashed] (u) -- (\vx, \uy);

		\draw[bistyle] (\xx+20,\xy) -- (x);
		\draw[bistyle] (\yx-20,\yy) -- (y);
		\draw[bistyle] (\xx,\xy-20) -- (x) -- (y) -- (\yx,\yy+20);
		
		\draw[fill=darkgreen!50,draw=none] (\xx+20,\xy) -- (x) -- (\xx,\xy-20) -- (\xx+20,\xy-20) -- cycle;
		\draw[fill=darkgreen!50,draw=none] (\yx-20,\yy) -- (y) -- (\yx,\yy+20) -- (\yx-20,\yy+20) -- cycle;
	\end{tikzpicture}
	\caption{The bisector (in green) between points~$u$ and $v$ under the Manhattan distance. The green lines and areas extend to infinity.
	We also see the bounding box $\BB(u,v)$ in the middle figure.}
	\label{fig:L1bisect}
\end{figure}

  \begin{figure}
    \centering
    \def \xs {.1}
    \def \xy {.1}
    \begin{tikzpicture}[black]
      \foreach \x / \y / \n / \nn / \typ / \p / \dx / \dy in {-4/-2/u/u/voterU/{below left}/-1/-4, 6/2/v/v/voterW/{above right}/-1/-1} {
        \node[\typ] at (\x*\xs,\y*\xy) (\n) {};
        \node[\p = \dx pt and \dy pt of \n] {$\nn$};
      }

      \node[alter,red,fill=red] at (2*\xs,-6*\xy) (x) {};
      \node[alter,red,fill=red] at (-2*\xs,6*\xy) (y) {};

      \drawcircle{x}{u}{colorV}
      \drawcircle{x}{v}{colorW}

      \node[alter,red,fill=red] at (x) {};
      \node[alter,red,fill=red] at (y) {};

    \end{tikzpicture}\quad
    \begin{tikzpicture}[black]

      \foreach \x / \y / \n / \nn / \typ / \p / \dx / \dy in {-4/-2/u/u/voterU/{below left}/-4/-3, 6/2/v/v/voterW/{above right}/-1/-1} {
        \node[\typ] at (\x*\xs,\y*\xy) (\n) {};
        \node[\p = \dx pt and \dy pt of \n] {$\nn$};
      }

      \coordinate (x) at (0*\xs,-4*\xy);
      \coordinate (y) at (-4*\xs,4*\xy);

      \drawcircle{x}{u}{colorV}
      \drawcircle{x}{v}{colorW}

      \gettikzxy{(v)}{\vx}{\vy}
      \gettikzxy{(x)}{\xx}{\xy}
      \node[alter,red,fill=red] at (x) {};
      \pgfmathsetlengthmacro\disRX{abs(\xx-\vx)+abs(\xy-\vy)}
      \path[draw,red,line width=2pt] (y) -- (\vx-\disRX,\vy);
    \end{tikzpicture}\quad
    \begin{tikzpicture}[black]
      \def \xs {.12}
       \def \xy {.12}

      \foreach \x / \y / \n / \nn / \typ / \p / \dx / \dy in {-4/-2/u/u/voterU/{below left}/-4/-3, 6/2/v/v/voterW/{above right}/-1/-1} {
        \node[\typ] at (\x*\xs,\y*\xy) (\n) {};
        \node[\p = \dx pt and \dy pt of \n] {$\nn$};
      }

      \coordinate (x) at (0*\xs,-2*\xy);
      \coordinate (y) at (-4*\xs,2*\xy);

      \drawcircle{x}{u}{colorV}
      \drawcircle{x}{v}{colorW}
      \path[draw,red,line width=2pt] (y) -- (x);
    \end{tikzpicture}\quad
    \begin{tikzpicture}[black]
       \def \xs {.12}
       \def \xy {.12}

      \foreach \x / \y / \n / \nn / \typ / \p / \dx / \dy in {-4/-2/u/u/voterU/{below left}/-4/-3, 2/-2/v/v/voterW/{above right}/-1/-1} {
        \node[\typ] at (\x*\xs,\y*\xy) (\n) {};
        \node[\p = \dx pt and \dy pt of \n] {$\nn$};
      }

      \coordinate (x) at (-4*\xs,-6*\xy);
      \coordinate (y) at (-4*\xs, 2*\xy);

      \drawcircle{x}{u}{colorV}
      \drawcircle{x}{v}{colorW}

     \gettikzxy{(v)}{\vx}{\vy}
      \pgfmathsetlengthmacro\disRX{abs(\xx-\vx)+abs(\xy-\vy)}
      \path[draw,red,line width=2pt] (y) -- (\vx-\disRX,\vy);
      \path[draw,red,line width=2pt] (x) -- (\vx-\disRX,\vy);
    \end{tikzpicture}

    \caption{The intersection  (in red) of two circles under the Manhattan distance in \ensuremath{\mathds{R}^{2}} can be two points, one point and one line segment, one line segment, or two line segments.}
    \label{fig:L1unit}
\end{figure}

\paragraph{Basic geometric notation.}
Throughout this paper, 
we use lower case letters in boldface to denote points in a space. 
Given two points~$\pq$ and $\pr$, we introduce the following notions:
Let \myemph{$\BB(\pq,\pr)$} denote the set of points which are contained in the (smallest) rectilinear bounding box of points~$\pq$ and $\pr$, i.e.,
$\BB(\pq,\pr)\coloneqq \{\px\in \dspace \mid \min\{\pq[i],\pr[i]\} \le \px[i] \le \max\{\pq[i],\pr[i]\} \text{ for all } i\in [d]\}$. %
See \cref{fig:L1bisect} for illustration.
The \myemph{perpendicular bisector} (bisector in short) between two points $\pq$ and $\pr$ wrt.\ a $p$-norm is a set~$\Mhalfspace_p(\pq,\pr)$ of points which each have the same distance to both~$\pq$ and $\pr$.
Formally, $\Mhalfspace_{p}(\pq,\pr)\coloneqq \{\px \in \dspace \mid \dis{\px-\pq}_p = \dis{\px-\pr}_p\}$.
In a $d$-dimensional space, a bisector of two points under the Manhattan distance (i.e., $1$-norm) can itself be a $d$-dimensional object, while a bisector under Euclidean distances is always $(d-1)$-dimensional; see e.g., \cref{fig:L1bisect} (right).

A \myemph{sphere} around $\pq$ of distance $s \geq 0$ is a set consisting of all points of distance $s$ to $\pq$.
Formally, it is the set $\{\px \in \dspace \mid \dis{\px - \pq}_p = s\}$. In two dimensions, we call a sphere a \myemph{circle}.
For $d=2$, the Manhattan distance of two points is equal to the length of a shortest path between them on a rectilinear grid.
Hence, under Manhattan distances, a circle is a square rotated at a $45^{\circ}$ angle from the coordinate axes. %
The intersection of two Manhattan-circles can range from two points to two segments as depicted in \cref{fig:L1unit}. 

\paragraph{The two-dimensional case.}
In a two-dimensional space, the vertical line and the horizontal line crossing any point divide the space into four non-disjoint quadrants: the north-east, south-east, north-west, and south-west quadrants.
Given a point~$\pq$, we use $\dNE(\pq)$, $\dSE(\pq)$, $\dNW(\pq)$, and $\dSW(\pq)$ to denote these four quadrants.
Formally,
\myemph{$\dNE(\pq)\coloneqq\{\pz \in \rz^2\mid \pz[1]\ge \pq[1] \wedge \pz[2]\ge \pq[2]\}$},
\myemph{$\dSE(\pq)\coloneqq\{\pz \in \rz^2\mid \pz[1]\ge \pq[1] \wedge \pz[2]\le \pq[2]\}$},
\myemph{$\dNW(\pq)\coloneqq\{\pz \in \rz^2\mid \pz[1]\le \pq[1] \wedge \pz[2]\ge \pq[2]\}$}, and
\myemph{$\dSW(\pq)\coloneqq\{\pz \in \rz^2\mid \pz[1]\le \pq[1] \wedge \pz[2]\le \pq[2]\}$}. %

\paragraph{Embeddings.}
The $d$-dimensional geometric representation under $p$-norm
 models the preferences of the voters over the alternatives 
using the %
$\ell_p$-distance.
We recall that a shorter distance indicates a stronger preference.

\newcommand{\cyclefig}{
	\begin{scope}[shift={(-1,0)}]
		\node[voter](v) at (2,2) {};
		\node[above right = -2pt and -1pt of v] {$v$};
		\node[alter](a) at (3,3) {};
		\node[above right = -2pt and -1pt of a] {$a$};
		\draw[colorV]  (2,0) -- (4,2) -- (2,4) -- (0,2) -- cycle;
	\end{scope}
}

\newcommand{\cyclefigeuc}{
	\begin{scope}[shift={(3,0)}]
		\node[voter](v) at (2,2) {};
		\node[above right = -2pt and -1pt of v] {$v$};
		\node[alter](a) at (3,3) {};
		\node[above right = -2pt and -1pt of a] {$a$};
		\draw[colorV] (v) circle[radius=1.41];
	\end{scope}
}

\newcommand{\cyclefigmax}{
	\begin{scope}[shift={(6,0)}]
		\node[voter](v) at (2,2) {};
		\node[above right = -2pt and -1pt of v] {$v$};
		\node[alter](a) at (3,3) {};
		\node[above right = -2pt and -1pt of a] {$a$};
		\draw[colorV]  (3,3) -- (3,1) -- (1,1) -- (1,3) -- cycle;
\end{scope}
}

\begin{figure}
\centering
\begin{tikzpicture}[scale=1, black]

		\cyclefig
		\cyclefigeuc
		\cyclefigmax
\end{tikzpicture}\label{fig:circles}
\caption{A circle around $v$ whose radius is the distance of $v$ and $a$ in \dManhattan[2], \dEuclid[2], and \dMax[2] spaces, respectively.}
\end{figure}

\begin{definition}[$d$-dimensional geometric embeddings under $p$-norm]\label{def:embeddings}
  Let $\ppp \coloneqq (\aaa, \vvv\coloneqq\{v_1, \dots,  v_n\}, \rrr\coloneqq(\pref_1, \dots, \pref_n))$ be a profile.   
  Let~$E\colon \aaa \cup \vvv \to \dspace$ be an embedding of the alternatives and the voters into a $d$-dimensional space. %
  Given $p \in \mathds{R}_{\geq 1} \cup \{\infty\}$, we say $\ppp$ is \myemph{$d$-dimensional geometric under $p$-norm} if there is an embedding $E$ such that for every voter $v_i \in V$,  for each two alternatives~$a,b\in \aaa$, it holds that
  \begin{align*}
    a \pref_i  b \text{ if and only if } \dis{E(a)-E(v_i)}_p < \dis{E(b)-E(v_i)}_p.
  \end{align*}
  In this case, we say $E$ is an embedding under $p$-norm. 
If $p = 1$ (resp.\ $p = 2$, $p = \infty$), we say $E$ is a \myemph{\dManhattan{}} (resp.\ \myemph{\dEuclid{}}, \myemph{\dMax{}}) embedding and the profile~$\ppp$ is \myemph{\dManhattan} (resp.\ \myemph{\dEuclid}, \myemph{\dMax{}}). %
\end{definition}

The following proposition allows us to extend any result we obtain of the (non-)existence of \dManhattan[2] embeddings to \dMax[2] embeddings and vice versa. The same claim has been made by \citet[Proposition~2]{EST2021Euclidlp}.

\begin{proposition}[\cite{Lee1980Voronoui}]\label{prop:max_man_eq}
  There is a natural isometry between $\mathds{R}^{2}$ under $1$-norm and $\mathds{R}^{2}$ under $\infty$-norm.
\end{proposition}

For intuition, observe that a circle in \dManhattan[2] space is a rotated and scaled version of a circle in \dMax[2] space, see \cref{fig:circles} for an illustration.

  The definition of embeddings can also be extended to the case where the preference orders in $\rrr$ are not necessarily strict.
  
\begin{remark}  
  We may also allow the preference orders to be weak orders, in which case the preferences may contain \myemph{indifferences} and we will use $\prefeq$ to refer to preference orders with indifferences.
  More formally, we write $a\prefeq b$ to refer to the case that $a$ is weakly prefer to~$b$.
  We use \myemph{$\pref$} to refer to the asymmetric part, that is, $a\pref b$ means $a$ is strictly preferred to~$b$,
  and \myemph{$\sim$} to refer to the symmetric part, that is, $a\sim b$ means that $a$ and $b$ are considered indifferent.
  The definition of rank function~$\rank(a)$ stays the same, i.e., it refers to the number of alternatives that are strictly preferred to alternative~$a$.

  The definition of $d$-dimensional geometric profiles will be extended as follows:
  Given $p \in \mathds{R}_{\geq 1} \cup \{\infty\}$, we say $\ppp$ is \myemph{$d$-dimensional geometric under $p$-norm} if there is an embedding $E$ such that for every voter $v_i \in V$,  for each two alternatives~$a,b\in \aaa$, it holds that
  \begin{align*}
    a \prefeq_i  b \text{ if and only if } \dis{E(a)-E(v_i)}_p \le \dis{E(b)-E(v_i)}_p.
  \end{align*}
If the preference orders do not contain indifferences, we say the profile and the preference orders are \myemph{strict}.
\end{remark}

In this paper we focus primarily on strict preferences. Unless stated otherwise, all of our results assume the preferences to be strict.
In \cref{appsec:sec:Manhattan-positive} we provide some positive results for the case with indifferent preferences and $p$-norms with $p > 1$.

\section{Manhattan Embedding Existence for Large Dimensions}\label{sec:Manhattan-positive}\appendixsection{sec:Manhattan-positive}

In this section, we show that for sufficiently high dimension~$d$, i.e., $d\ge \min(n, m-1)$, every profile with $n$ voters and $m$~alternatives is \dManhattan, even if the voters may have indifferent preferences. The same result holds for \dEuclid profiles by Bogomolnaia and Laslier~\cite{BoLa2006}.
The idea behind our proof for $d =n$ is similar to the one for \dEuclid[n] preferences by~\citet{BoLa2006}.
The proof for $d=m-1$ is however different from the \dEuclid[(m-1)] case. While the proof for the \dEuclid case relies on abstract geometric properties, it is relatively straightforward to give a full concrete construction of the \dManhattan case.
In the appendix of this section we show that for every $p \in \mathds{R}_{> 1} \cup \{\infty\}$ (the case $p=1$ being the Manhattan results of this section), every profile with $n$ voters is $n$-dimensional geometric under $p$-norm and every profile with~$m$ alternatives is $(m-1)$-dimensional geometric under $p$-norm.

\begin{figure}
  \centering
\def\mx{-14}
\def\mmx{14}
\def\my{-2}
\def\mmy{13}
  \begin{tikzpicture}[scale=0.3, black]
    \tkzInit[xmax=\mmx,ymax=\mmy,xmin=\mx,ymin=\my]
    \exPosiTmod

    \begin{pgfonlayer}{background}  
      \tkzGrid[color=gray!20]
      \path[draw, thick] (\mx,0) -- (\mmx,0);
      \path[draw, thick] (0,\my) -- (0,\mmy);
    \end{pgfonlayer}
    
    \coordinate (start) at (-9,2);
    \coordinate (end) at (-20,-2);
    \coordinate (ss) at (25,0);
    \coordinate (ee) at (0,23);

    \begin{pgfonlayer}{background}

    \foreach \v / \co in {v1/lineA}  {
        \foreach \c in {a1,a2,a3,a4,a5} {
         \drawcircleD{\c}{\v}{\co}{start}{end}{ss}{ee}
        }
      }
    \end{pgfonlayer}
    
    \coordinate (start) at (-23,-5);
    \coordinate (end) at (0,-2);
    \coordinate (ss) at (9,0);
    \coordinate (ee) at (0,17);

    \begin{pgfonlayer}{background}
      \foreach \v / \co in {v2/lineB}  {
        \foreach \c in {a1,a2,a3,a4,a5} {
          \drawcircleD{\c}{\v}{\co}{start}{end}{ss}{ee}
        }
      }
    \end{pgfonlayer}
  \end{tikzpicture}
  \caption{Illustration for \cref{ex:d=n-Manhattan}.}\label{fig:d=n-Manhattan}
\end{figure}

\paragraph{Embedding with $d = n$ voters.}
\begin{theorem}\label{thm:d=n->dManhattan}
Every profile with $n$~voters is \dManhattan[n], even when the preference orders may contain indifferences.
\end{theorem}
\begin{proof}

 Let $\ppp=(\aaa, \vvv, (\pref_i)_{i\in [n]})$ be a profile with~$m$ alternatives $\aaa=[m]$ and $n$~voters~$\vvv$.

Conceptually, the proof is similar to that of \citet{BoLa2006}.
By embedding the voters correctly, we can find $m$ spheres of increasing sizes around each alternative so that all of these spheres intersect with all the spheres of the other voters.
These intersection points correspond to different combinations of voter ranks for alternatives.
See \cref{ex:d=n-Manhattan} and \cref{fig:d=n-Manhattan} for an example: The smallest blue solid circle corresponds to alternatives for which $v_1$ has rank~$0$, the second smallest blue circle to alternatives for which $v_1$ has rank~$1$ and so on.
Similarly, the smallest red dashed circle corresponds to alternatives for which $v_2$ has rank~$0$, the second smallest red circle to alternatives for which $v_1$ has rank $1$ and so on.
Consider for example the alternative 1, which in this example satisfies $\rank_{1}(1) = 0$ and $\rank_{2}(1) = 2$.
Hence it is placed in the intersection of the blue circle closest to $v_1$ and the red circle third closest to $v_2$.

Now let us describe our formal construction.
First we embed the $n$ voters in $n - 1$ dimensions so that for every voter $v_i \in \vvv$, the first~$i - 1$ coordinates are $m$, the coordinates from $i$ to $n - 1$ are $-m$ and the last coordinate is $0$. More formally, for every $v_i \in \vvv$:
\[E(v_i)[z] \coloneqq \begin{cases}
m, &\text{ if } z \leq i - 1,\\
-m, &\text{ if } i \leq z \leq n -1, \\
0, &\text{ if }z = n.
\end{cases}\] 

We embed the alternatives in the following points:
\[E(j)[z] \coloneqq \begin{cases}
\rank_z(j) - \rank_{z+1}(j), &\text{ if } z \neq  n,\\
\rank_1(j) + \rank_n(j), &\text{ if }  z = n.
\end{cases}\]
These are the intersections of the spheres of radius $m(n - 1) + 2 \rank_i(j)$ around $E(v_i)$, $v_i \in \vvv$.
We can show this by computing the distance $\Mdis{E(v_i) - E(j)}$ for every $v_i \in \vvv, j \in \aaa$:
\begin{align}
\Mdis{E(v_i) - E(j)}  &= \sum_{z =1}^n|E(v_i)[z] - E(j)[z]| \label{eq:ndim1}\\
& = \sum_{z = 1}^{i - 1}|m - (\rank_z(j) - \rank_{z+1}(j))| + \sum_{z = i}^{n - 1}|- m - (\rank_z(j) - \rank_{z+1}(j))|   + |0 - (\rank_1(j) + \rank_n(j))| \label{eq:ndim3}\\
& = \left( \sum_{z = 1}^{i - 1}m - \rank_z(j) + \rank_{z+1}(j) \right) + \left(\sum_{z = i}^{n - 1} m + \rank_z(j) - \rank_{z+1}(j)\right)  + \rank_1(j) + \rank_n(j) \label{eq:ndim4}\\
& = (n - 1) m + \left( \sum_{z = 1}^{i - 1}-\rank_z(j) + \rank_{z+1}(j)\right) + \left(\sum_{z = i}^{n - 1} \rank_z(j) - \rank_{z+1}(j)\right) + \rank_1(j) + \rank_n(j) \label{eq:ndim6}\\
&= (n - 1)m + 2 \rank_i(j).\label{eq:ndim5}
\end{align}

For Step~\eqref{eq:ndim5}, observe that if $2 \leq i \leq n - 1$, then the terms in the first sum cancel each other so that in the end we have $-\rank_1(j) + \rank_i(j)$. Similarly, the second sum simplifies to the form $\rank_i(j) - \rank_n(j)$. %

Since $m(n - 1) + 2 \rank_i(j)$ is linear in the ranks, this proves the statement. This clearly holds even if preferences may contain indifferences.

\end{proof}

By \cref{thm:d=n->dManhattan}, we obtain that any profile with two voters is \dManhattan[2].
The following example provides an illustration.
\begin{example}\label{ex:d=n-Manhattan}
  \stepcounter{myprofilecounter}
  Consider profile~$\ppp_{\themyprofilecounter}$ with $2$ voters and $5$ alternatives:
  \begin{align*}
    v_1\colon  1 \succ 2 \succ 3 \succ 4 \succ 5, \qquad  v_2\colon  4 \succ 3 \succ 1 \succ 5 \succ 2.
  \end{align*}
  By the proof of \cref{thm:d=n->dManhattan}, $E(v_1) = (-m, 0)$ and $E(v_2) = (m, 0)$. For every alternative $j \in \aaa$, the $E(j) = (\rank_1(j) - \rank_2(j), \rank_1(j) + \rank_2(j))$. Also see \cref{fig:d=n-Manhattan} for an illustration.
  \begin{center}
    \begin{tabular}{l|cc|cccccc}
      $x\in \vvv\cup \aaa$ & $v_1$ & $v_2$ & $1$ & $2$ & $3$ & $4$ & $5$ \\\hline\\[-2ex]
      $E(x)[1]$ & $-5$ & $0$ & $-2$ & $-3$ & $1$ & $3$ & $1$\\
      $E(x)[2]$ & $5$ & $0$ & $2$ & $5$ & $3$ & $3$ & $7$\\
      \end{tabular}
    \end{center}
    \stepcounter{myprofilecounter}
  \end{example}

\paragraph{Embedding with $d = m-1$ alternatives.}
\begin{theorem}\label{thm:d=m-1->Manhattan}
Every profile with $m$ alternatives is \dManhattan[(m-1)].
\end{theorem}

\begin{proof}
We may assume that $m \geq 2$, because otherwise we have one alternative and the profile can trivially be embedded in a point.

 Let $\ppp=(\aaa, \vvv, (\pref_i)_{i\in [n]})$ be a profile with~$m \geq 2$ alternatives $\aaa=[m]$ and $n$~voters~$\vvv$. 
The idea is to place every alternative except $m$ on its own axis.
Then it is straightforward to choose the placement of a voter so that any possible preference order of the first $m - 1 $ is embedded: As the alternatives are on their own axes, we can move closer to one alternative without changing how close we are to a different alternative.
Finally, we embed $m$ to the origin.
By choosing how far an alternative is from the origin, we can respect the voter's preferences regarding $m$.

  More precisely, define an embedding $E\colon \aaa\cup \vvv \to \mathds{N}_0$ such that
  alternative~$m$ is embedded in the origin coordinate, i.e., $E(m)[z]=0$ for all $z\in [m - 1]$.
  For each alternative~$j\in [m - 1]$ and each coordinate~$z\in [m - 1]$, we have $E(j)[z]\coloneqq 2m$ if $z=j$, and $E(j)[z]\coloneqq 0$ otherwise.

  Then, the embedding of each voter~$v_i\in \vvv$ is defined as follows: $\forall~z\in [m - 1]\colon $
  \begin{align*}
   E(v_i)[z]  & 
    \coloneqq
    \begin{cases}
      2m  - \rank_i(z), & \text{ if } \rank_i(z) < \rank_i(m),\\
      m -\rank_i(z), & \text{ if } \rank_i(z) > \rank_i(m).
    \end{cases}
  \end{align*}
  Observe that $0\le E(v_i)[j] \le 2m$. 
  Before we show that $E$ is \dManhattan[2] for~$\ppp$, let us establish a simple formula for the distance between a voter and an alternative.
 
  \begin{restatable}{clm}{clmmManhattanr}\label{clm:m-Manhattan}
    \clmmManhattan
  \end{restatable}

  \begin{proof}[Proof of
    \cref{clm:m-Manhattan}]
    \renewcommand{\qed}{~\hfill~$\diamond$}
    The case with $j=m$ is straightforward since alternative~$m$ is embedded at the origin.
    The proof for $j\neq m$ is also straightforward by a direct application of the definition:
    \begin{align*}
      \Mdis{E(v_i)-E(j)}
      & = \sum_{z\in [m - 1]} | E(v_i)[z]-E(j)[z]| = \left(\sum_{z\in [m -1]\setminus \{j\}}| E(v_i)[z]| \right) + | E(v_i)[j]-E(j)[j]| \\
      & = \left(\sum_{z\in [m - 1]\setminus \{j\}}| E(v_i)[z]| \right) + (2m-E(v_i)[j])  = \Mdis{E(v_i)} + 2(m-E(v_i)[j]). 
    \end{align*} This concludes the proof.\end{proof}    

  Now, we proceed with the proof.  Consider an arbitrary voter~$v_i\in \vvv$ and let $j,k\in [m]$ be two consecutive alternatives in the preference order~$\succ_i$ such that  $\rank_i(j) = \rank_i(k)-1$.

It is clear from \cref{clm:m-Manhattan} and the voter embedding that if $j \succ_i k \succ_i m$ or $m \succ_i j \succ_i k$, then $\Mdis{E(v_i) - E(j)} < \Mdis{E(v_i) - E(k)}$.
    It remains to consider the cases when $m \in \{j,k\}$.
   \begin{description}

    \item[Case 1:] $k=m$ and thus $E(v_i)[j]=2m-\rank_i(j)$.
    Then,  by \cref{clm:m-Manhattan} and by definition, it follows that
    \begin{align*}
      \Mdis{E(v_i)-E(j)} &- \Mdis{E(v_i)-E(k)}  = 2(m - E(v_i)[j]) =2\rank_i(j) - 2m < 0.
    \end{align*}
    Note that the last inequality holds since $\rank_i(j)=\rank_i(k)-1 < m$.

    \item[Case 2:] $j=m$ and thus $E(v_i)[k]=m-\rank_i(k)$.
    Then,  by \cref{clm:m-Manhattan} and by definition, it follows that
    \begin{align*}
      \Mdis{E(v_i)-E(j)} &- \Mdis{E(v_i)-E(k)}  = -2(m - E(v_i)[k]) = - 2\rank_i(k) < 0.
    \end{align*}
  \end{description}

  Since in all cases, we show that $\Mdis{E(v_i)-E(j)} - \Mdis{E(v_i)-E(k)} < 0$,
  embedding~$E$ is indeed \dManhattan[(m-1)] for~$\ppp$.
\end{proof}

  We can extend  \cref{thm:d=m-1->Manhattan} to profiles where the voters preferences are not necessarily strict. The construction requires an additional case and is deferred to appendix.
  \begin{restatable}[\appsymb]{proposition}{statementmwithties}\label{rm:mwithties}
    Every profile with $m$ alternatives and with possibly indifferent preferences is \dManhattan[(m-1)].
  \end{restatable}
  
  \appendixproofwithstatement{rm:mwithties}{\statementmwithties*}{
  	\begin{proof}
  		We amend the embedding of the voters as follows:
  		The embedding of each voter~$v_i\in \vvv$ is defined as follows: $\forall~z\in [m - 1]\colon $
  		\begin{align*}
  			E(v_i)[z]  & 
  			\coloneqq
  			\begin{cases}
  				2m  - \rank_i(z), & \text{ if } \rank_i(z) < \rank_i(m),\\
  				\myemph{$m$} & \myemph{\text { if }} \myemph{$\rank_i(z) = \rank_i(m)$},\\
  				m -\rank_i(z), & \text{ if } \rank_i(z) > \rank_i(m).
  			\end{cases}
  		\end{align*}
  		
  		\cref{clm:m-Manhattan} and its proof remain unchanged. Recall that \cref{clm:m-Manhattan} states the following:
  		\clmmManhattan
  		
  		We need to show that for every pair of alternatives $j,k \in \aaa$, a voter $v_i \in \vvv$ (i) if  $j \sim_i k$, then $\Mdis{E(v_i) - E(j)} = \Mdis{E(v_i) - E(k)}$, and (ii) if $j \succ_i k$, then $\Mdis{E(v_i) - E(j)} < \Mdis{E(v_i) - E(k)}$.
  		
  		We start by showing Property (i).
  		It is clear that if $j \sim_i k$ and $m \notin \{j,k\}$, then $\Mdis{E(v_i) - E(j)}$ and $\Mdis{E(v_i) - E(k)}$ depend only on how $v_i$ ranks them, and thus the distances must be equal. 
  		If $m \in \{j,k\}$, assume without loss of generality that $k = m$, then $\Mdis{E(v_i) - E(j)} = \Mdis{E(v_i)} + 2(m - E(v_i)[j]) = \Mdis{E(v_i)} + 2 (m - m) = \Mdis{E(v_i) - E(m)}$, as required.
  		
  		The proof of Property (ii) is mainly shown in \cref{thm:d=m-1->Manhattan}. The only case that is not covered is when either $j$ or $k$ is ranked the same as $m$. 
  		If $\rank_i(j) = \rank_i(m)$, then $m \succ_i k$. We obtain that $\Mdis{E(v_i) - E(j)} = \Mdis{E(v_i)} + 2(m - E(v_i)[j]) = \Mdis{E(v_i)} + 2 (m - m) = \Mdis{E(v_i)}$, and $\Mdis{E(v_i) - E(k)} > \Mdis{E(v_i)}$, since $E(v_i)[k] < m$.
  		
  		Otherwise if $\rank_i(k) = \rank_i(m)$, then $j \succ_i m$. We obtain that $\Mdis{E(v_i) - E(k)} = \Mdis{E(v_i)}$, and $\Mdis{E(v_i) - E(j)} < \Mdis{E(v_i)}$, since $E(v_i)[j] > m$.
  		
  		This concludes the proof.
  		
  	\end{proof}
  }

\noindent \cref{thm:d=m-1->Manhattan} implies that every profile with $3$~alternatives is \dManhattan[2].
The following example illustrates a corresponding Manhattan embedding. %
\begin{example}  \label{ex:yes-v6-m3}
  The following profile~$\ppp_{\themyprofilecounter}$ with $6$~voters and $3$~alternatives is \dManhattan[2].
  \begin{align*}
  v_1 \colon & 1 \succ 2 \succ 3, & v_3 \colon & 2 \succ 1 \succ 3, & v_5 \colon & 3 \succ 1 \succ 2,\\
  v_2 \colon & 1 \succ 3 \succ 2, & v_4 \colon & 2 \succ 3 \succ 1, & v_6 \colon & 3 \succ 2 \succ 1. 
  \end{align*}
  \stepcounter{myprofilecounter}
  One can check that the embedding~$E$ given in \cref{fig:d=m+1-Manhattan} is \dManhattan[2] for~$\ppp_{\themyprofilecounter}$.
\end{example}

\begin{figure}
  \centering
\begin{subfigure}[c]{0.45\textwidth}\centering
  \begin{tikzpicture}[scale=0.6, black]
    \begin{pgfonlayer}{background}
      \tkzInit[xmax=7,ymax=7,xmin=-2,ymin=-2]
      \tkzGrid[color=gray!20]
    \end{pgfonlayer}
    \exPosi
    \begin{pgfonlayer}{background}  
    	\path[draw, blue!50!white, dashed] (7,0) -- (-2,0);
    	\path[draw, blue!50!white, dashed] (0,7) -- (0,-2);
    \end{pgfonlayer}
    \coordinate (start) at (-5,2);
    \coordinate (end) at (-4,-5);
    \coordinate (ss) at (6,0);
    \coordinate (ee) at (0,7);

    \begin{pgfonlayer}{background}
    
    \draw[bistyle, draw=orange!70!black] (7, 3) to (-2,3);
    \draw[bistyle, draw=red!50!blue] (3, -2) to (3,7);
    \draw[bistyle] (0,0) to (6,6);
    \draw[fill=green!80!black,draw=none] (0,0) -- (-2,0) -- (-2,-2) -- (0,-2) -- cycle;
    
    \draw[fill=green!80!black,draw=none] (7,7) -- (7,6) -- (6,6) -- (6,7) -- cycle;

    \end{pgfonlayer}

    \foreach \c / \d / \x / \y / \colo in {
    3/1/4.5/-1.6/{red!50!blue},
    2/3/0.1/3.3/{orange!70!black},
    2/1/0.1/-1.1/{green!70!black}} {
    \node[bistyle, fill=white, inner sep = 1pt,text=\colo] at (\x,\y) {\footnotesize Bisector of \c\ and \d};    
  }
  
    \coordinate (start) at (-9,0);
    \coordinate (end) at (-4,-5); 
    \coordinate (ss) at (8,0); %
    \coordinate (ee) at (0,9); %
  \end{tikzpicture}\end{subfigure}
\begin{subfigure}[c]{0.5\textwidth}
    \begin{tabular}{l|cccccc|ccc}
      $x\in \vvv\cup \aaa$ & $v_1$ & $v_2$& $v_3$& $v_4$& $v_5$& $v_6$  & $1$ & $2$ & $3$ \\\hline\\[-2ex]
      $E(x)[1]$ & $6$ & $6$  &$5$ & $1$ & $2$ & $1$ & $6$ & $0$ & $0$\\
      $E(x)[2]$ & $5$ & $1$  & $6$ & $6$ & $1$ & $2$ & $0$ & $6$ & $0$\\
      \end{tabular}\end{subfigure}

 \caption{Illustration for \cref{ex:yes-v6-m3}; %
the solid colored lines and areas are the bisectors of pairs of alternatives whereas the dashed blue lines are the coordinate axes.} \label{fig:d=m+1-Manhattan}
\end{figure}

In Appendix~\ref{appsec:sec:Manhattan-positive} we extend our previous results for an arbitrary $p \in \mathds{R}_{> 1} \cup \{\infty\}$.
We also show that the proof from \citet[Proposition 4]{BoLa2006} which attempts to construct a \dEuclid\ embedding when $d = n$ does not work when there are indifferent preferences and provide a fixed construction.

Having established that every profile admits a \dManhattan{} embedding when $d \ge \min(n, m-1)$, we now turn to the more challenging question: what prevents profiles from being embeddable in low dimensions, specifically $d=2$?

\toappendix{
  \citet[Proposition 4]{BoLa2006} provided a construction to show that every preference profile with $n$ voters and possibly indifferent preferences is \dEuclid[n]. 
  Their idea was to embed each voter on a distinct and private axis and then embed each alternative such that the distance between the voter and the alternative on the private axis respects the rank of the alternative for the voter.  
  Unfortunately, this does not work for the case with indifferent preferences, as we will see in the next example.
  In  \cref{prop:BoLaFix}, we show how to fix it. 

  \begin{example}\label{ex:BoLaBug}
    \setcounter{vtwocounter}{\themyprofilecounter}
    \citeauthor{BoLa2006} proposed the following embedding~$E$ for showing that every preference profile with $n$ voters and possibly indifferent preferences is \dEuclid[n]. 
    For every $v_i \in \vvv$ and dimension $z \in [n]$,
    let \[E(v_i)[z] \coloneqq \begin{cases}
      M, &\text{ if } z = i,\\
      0, &\text{ otherwise}.
    \end{cases}\] 
  Here $M$ is some sufficiently large positive real value. 

  For every alternative $j \in \aaa$ and dimension~$z \in [n]$ (which also corresponds to a voter),
    let
    $E(j)[z] \coloneqq -|\{k \in \aaa  \mid k \succeq_z j\}|$.\footnote{Note that they had a typo and used ``$j \succeq_z k$'' in their original definition. See the proof of Proposition 4 in their paper~\cite{BoLa2006}.}

    Unfortunately, the embedding of the alternatives is problematic as the overall distance from an alternative to a voter does not only depend on its rank in the voter's preferences. To see this, consider the following profile with $2$ voters and $2$ alternatives.
    \stepcounter{myprofilecounter}
    \begin{align*}
      \ppp_{\thevtwocounter}   \colon \quad v_1\colon  1 \sim 2, \qquad   v_2\colon 1 \succ 2.
    \end{align*}
    
  By their construction, the embeddings are as follows:
  \begin{center}
    \begin{tabular}{ccccc}
      $v_i \in \vvv$ & $E(v_i)$ & & $j \in \aaa$ & $E(j)$\\\midrule
      $v_1$ & $(M,0)$ & & $1$ & $(-2,-1)$\\
      $v_2$ & $(0,M)$ & & $2$ & $(-2,-2)$
    \end{tabular}
  \end{center}
  As one can see, $v_1$'s distance to~$1$ is not the same as his distance to~$2$ since 
  $\Edis{E(v_1) - E(1)} = \sqrt{(M + 2)^2 + 1}$ and $\Edis{E(v_1) - E(2)} = \sqrt{(M + 2)^2 + 2^2}$.
  This contradicts $v_1 \colon 1 \sim 2$.
\end{example}

\newcommand{\veri}[1]{[#1]}
\newcommand{\Nzero}{\mathds{N} \cup \{0\}}
\newcommand{\LL}{\ensuremath{\textsf{L}}}

To fix their problem \cite{BoLa2006} (see \cref{ex:BoLaBug}), we provide a new proof that also works for indifferent preferences and for all $p$-norms with $p \in \mathds{R}_{>1} \cup \{\infty\}$.
The key idea is to seek a point for each alternative~$j\in \aaa$ whose distance to each voter~$v_i$ corresponds to $\LL+\rank_i(j)$ (for a large enough~$\LL$) so that the preferences are represented exactly by distance comparisons.
To show that such point exists, we need to show that the $\ell_p$-spheres centered at the $E(v_i)$ with radii $\LL+\rank_i(j)$ have non-empty intersection. 
We accomplish this by rewriting the sphere-intersection problem as a continuous self-map on a compact convex box and applying Brouwer's fixed point theorem~\cite{Border1985}.

\begin{restatable}{proposition}{propBoLaFix}\label{prop:BoLaFix}
  For every $p \in \mathds{R}_{> 1} \cup \{\infty\}$, every profile with $n$ voters and with possibly indifferent preferences is $n$-dimensional geometric under $p$-norm.
  In particular, this implies that it is \dEuclid[n]. 
\end{restatable}

\begin{proof}
  Let $\LL > 1$ be some large positive number (chosen later).
  Similarly to Bogomolnaia and Laslier~\cite{BoLa2006}, we embed the voters on coordinate axes. Formally, 
  For every $v_i \in \vvv$
  and dimension $z \in [n]$,
  let \[E(v_i)\veri{z} \coloneqq
    \begin{cases}
      \LL, &\text{ if } z = i,\\
      0, &\text{ otherwise}.
    \end{cases}\] 
  In the following, we aim to find a point~$\px_j\in \RR^{n}$ such that for every voter~$v_i\in \vvv$: %
  \begin{align}
    \dis{E(v_i) - \px_j}_p = \LL + \rank_i(j). \label{eq:sphere-system}
  \end{align}
  If such a point~$\px$ exists, then the distance comparisons from~$E(v_i)$ exactly coincide with rank comparison, including ties.
  As already discussed, we will apply Brouwer's fixed point theorem by defining a continuous function~$f_j\colon K \to K$ on a non-empty, compact, and convex domain~$K$.

  \begin{description}
    \item[A compact and convex domain.] Define the box~$K\coloneqq [-m, 1]^{n}\subseteq \RR^n$. This set is non-empty, compact (closed and bounded), and convex.
    \item[A continuous function.] Define $f_j\colon K \to \RR^n$ as follows. For each~$\px \in K$ and coordinate~$i \in [n]$,
    let 
    \begin{align}
      f_j(\px)[i] \coloneqq -\rank_i(j) + \left(\dis{E(v_i)-\px}_p - (\LL - \px[i])\right).\label{def:f_j}
    \end{align}
  \end{description}
  If we can show that
  $f_j$ is continuous and
  $f_j(K) \subseteq K$,
  then we can apply Brouwer's fixed point theorem and conclude that there exists a point~$\px^{\star}\in K$ with $f_j(\px^{\star}) = \px^{\star}$.
  Then, for each~$i\in [n]$, we have by \eqref{def:f_j} that
  \[
    \px^{\star}[i] = -\rank_i(j) + \left(\dis{E(v_i)-\px^{\star}}_p - (\LL - \px^{\star}[i])\right).
  \]
  Rearranging yields
  \[
    \dis{E(v_i)-\px^\star}_p = \LL + \rank_i(j) \text{ for all voters~}v_i\in \vvv, 
  \]
  which is exactly \eqref{eq:sphere-system}.

  Before we show the two properties, let us recall one fact from calculus.
  \begin{fact}\label{fact:concave-function-approx}
    Let $p \in \mathds{R}_{> 1}$ and $g(y) = \sqrt[p]{y}$ for $y\ge 0$.
    Then, for all $y,s \ge 0$, it holds that 
    \begin{align}
      g(y+s) \le g(y) + g'(y)\cdot s = y^{\frac{1}{p}} + \frac{1}{p}y^{\frac{1}{p}-1}\cdot s.
    \end{align}
  \end{fact}
   \begin{proof}[Proof of
    \cref{fact:concave-function-approx}]
    \renewcommand{\qed}{~\hfill~$\diamond$}
    Since $p > 1$ and $y \ge 0$, we have
    $g''(y) = \frac{1}{p}(\frac{1}{p}-1)y^{\frac{1}{p}-2} \le 0$.
    So $g$ is concave on~$(0,\infty)$.
    Hence, $g$ lies below its tangent line. That is, for all $y\ge 0$ and $s\ge 0$, 
    $g(y+s) \le g(y) + g'(y)\cdot s = y^{\frac{1}{p}} + \frac{1}{p}y^{\frac{1}{p}-1}\cdot s$, as desired.
  \end{proof}

  Now, we are ready to show the two properties. We note that the following claim uses \cref{fact:concave-function-approx} for finite~$p > 1$; for $p = \infty$ and large~$\LL$, one has $\dis{E(v_i)-\px}_{\infty} = \LL - \px[i]$ directly, so $f_j$ is constant and the claim holds trivially.
  \begin{clm}\label{clm:n-voters-n-geometric-f}
    The function~$f_j$ from \eqref{def:f_j} is continuous and $f_j(K) \subseteq K$ for some~$\LL$.
  \end{clm}
  \begin{proof}[Proof of
    \cref{clm:n-voters-n-geometric-f}]
    \renewcommand{\qed}{~\hfill~$\diamond$}
    For each~$i\in [n]$, each coordinate map~$\px \mapsto \px[i]$ is continuous, and the function~$\px \mapsto E(v_i)-\px$ is affine and hence continuous. 
    The $p$-norm is continuous as a composition of continuous operations.
    Therefore, $f_j$ is continuous.

    It remains to show the second part of the statement.
    Fix a point~$\px\in K$ and a coordinate~$i\in [n]$.
    We aim to show that $0\le \dis{E(v_i)-\px}_p - (\LL-\px[i]) \le 1$.
    By definition, $\px[i] \le 1$, so $\LL-\px[i] \ge \LL-1 > 0$, and
    \begin{align}
      \label{eq:nvoters-distance}
      \dis{E(v_i)-\px}_p^p = |\LL-\px[i]|^p + \sum_{z \in [n]\setminus \{i\}}|\px[z]|^p = (\LL-\px[i])^p + \sum_{z \in [n]\setminus \{i\}}|\px[z]|^p.  
    \end{align}
    Apply \cref{fact:concave-function-approx} with $y =  (\LL-\px[i])^p$ and $s=|\px[z]|^p$, we get
    \begin{align}
      (\LL-\px[i]) \le \dis{E(v_i)-\px}_p &=\left((\LL-\px[i])^p+\sum_{z\in [n]\setminus \{i\}}|\px[z]|^p\right)^{1/p}\nonumber\\
      & \le (\LL-\px[i]) + \frac{1}{p}(\LL-\px[i])^{1-p}\sum_{z\in[n]\setminus \{i\}}|\px[z]|^p.
    \end{align}
    Subtracting~$(\LL-\px[i])$ gives
    \begin{align}
      0\le \dis{E(v_i)-\px}_p - (\LL-\px[i])&\le \frac{1}{p}(\LL-\px[i])^{1-p}\sum_{z\in[n]\setminus \{i\}}|\px[z]|^p
                                              \le \frac{(n-1)m^p}{p(\LL-1)^{p-1}} \eqqcolon \delta(\LL).\label{def:delta}
    \end{align}
    The last inequality holds since $\px[i] \in [-m, 1]$.
    Clearly, we can choose~$\LL >1$ so that $\delta(\LL) \le 1$.
    For instance, it suffices to take~$\LL \ge 1 + \big(\frac{(n-1)m^p}{p}\big)^{\frac{1}{p-1}}$.
    Then, for all~$\px\in K$, using~\eqref{def:delta},
    we get
    \begin{align*}
      f_j(\px)[i] = -\rank_i(j) + \left(\dis{E(v_i)-\px}_p - (\LL - \px[i])\right) \le -\rank_i(j) + \delta(\LL) \le 0+1=1,
    \end{align*}
    and
    \begin{align*}
      f_j(\px)[i] = -\rank_i(j) + \left(\dis{E(v_i)-\px}_p - (\LL - \px[i])\right) \ge -\rank_i(j) + 0 \ge -m.
    \end{align*}
    Hence, $f_j(\px)\subseteq K$.
  \end{proof}
  By \cref{clm:n-voters-n-geometric-f}, we infer that there exists a point~$\px^{\star}$ such that $f_j(\px^{\star}) = \px^{\star}$, as desired.

  Repeating the above construction independently for each alternative~$j \in \aaa$ defines $E$ on all alternatives in~$\aaa$.
  By construction, for every voter $v_i$ and alternatives $j, k\in \aaa$,
  \[
    \dis{E(v_i)-E(j)}_p - \dis{E(v_i)-E(k)}_p = \rank_{i}(j) - \rank_i(k),
  \]
  so
  \[
    \dis{E(v_i)-E(j)}_p \le \dis{E(v_i)-E(k)}_p \text{ if and only if }
     \rank_{i}(j) \le \rank_i(k) \text{ if and only if }
     j\prefeq_i k.
\]
In particular, if $j\sim_i k$, then $ \rank_{i}(j) = \rank_i(k)$, and the two distances are equal.
\end{proof}

Finally, we extend our result for $m$ alternatives to arbitrary $p \in \mathds{R}_{>1} \cup \{\infty\}$.
In the proof we first observe that our construction from \cref{thm:d=m-1->Manhattan} works for every $p \in  \mathds{R}_{\geq 1}$.
However, the case where $p = \infty$  requires a new construction, which we also present.
\begin{restatable}[\appsymb]{proposition}{promminusonevotersallp}
For every $p \in \mathds{R}_{> 1} \cup \{\infty\}$, every profile with $m$ alternatives is $(m-1)$-dimensional geometric under $p$-norm.\label{pro:mminusonevotersallp}
\end{restatable}

\begin{proof}

We first show the proof for $p \in \mathds{R}_{\geq 1}$.
As we mentioned, we will show that the construction from \cref{thm:d=m-1->Manhattan} also works for $p \in \mathds{R}_{\geq 1}$.
To do this, we show that the location of the bisectors between alternatives does not depend on our choice of $p$ as long as $p \geq 1$.
Since $p = 1$ corresponds to \dManhattan{} preferences and we have shown that our construction works for \dManhattan{} preferences, the result follows.

\begin{restatable}{clm}{statementbisector}\label{clm:eqbisector}
Let $E$ be the embedding constructed in the proof of \cref{thm:d=m-1->Manhattan}.
For every pair of alternatives $i, j \in \aaa \setminus \{m \}$, for every $p \in \mathds{R}_{\geq 1}$, the bisector of $E(i)$ and $E(j)$ restricted to $[0,2m]^{m-1}$ is $ \{ \vect{x} \in [0,2m]^{m-1} \mid \vect{x}[i] = \vect{x}[j]\}$. Additionally, the bisector of $E(i)$ and $E(m )$ restricted to $[0,2m]^{m-1}$ is  $ \{ \vect{x} \in [0,2m]^{m-1} \mid \vect{x}[i] = m\}$. 
\end{restatable}

\begin{proof}\renewcommand{\qedsymbol}{$\diamond$}
	
We first compute the bisector of $E(i)$ and $E(j)$ restricted to $[0,2m]^{m-1}$. This is the set of points $\vect{x} \in [0,2m]^{m-1}$ such that:
\begin{align*}
	\dis{E(i) - \vect{x}}_p &= \dis{E(j) - \vect{x}}_p &\iff \\
	\left(\sum_{z = 1}^{m-1} |E(i)[z] - \vect{x}[z]|^p\right)^{\frac{1}{p}} &= \left(\sum_{z = 1}^{m-1} |E(j)[z] - \vect{x}[z]|^p\right)^{\frac{1}{p}} &\iff \\
	\sum_{z = 1}^{m-1} |E(i)[z] - \vect{x}[z]|^p &= \sum_{z = 1}^{m-1} |E(j)[z] - \vect{x}[z]|^p &\iff \\
	\left(\sum_{z \in [m-1] \setminus \{i\}} |\vect{x}[z]|^p\right) + |2m - \vect{x}[i]|^p &= \left(\sum_{z \in [m-1] \setminus \{j\}} |\vect{x}[z]|^p\right) + |2m - \vect{x}[j]|^p &\iff \\
	\left(\sum_{z \in [m-1] \setminus \{i\}} \vect{x}[z]^p\right) + (2m - \vect{x}[i])^p &= \left(\sum_{z \in [m-1] \setminus \{j\}} \vect{x}[z]^p\right) + (2m - \vect{x}[j])^p &\iff \\
	\vect{x}[j]^p + (2m - \vect{x}[i])^p &= \vect{x}[i]^p + (2m - \vect{x}[j])^p & \iff \\
	\vect{x}[i]^p - (2m - \vect{x}[i])^p &=  \vect{x}[j]^p - (2m - \vect{x}[j])^p.
\end{align*}  

Let us assume, without loss of generality, that $\vect{x}[i] \leq \vect{x}[j]$. Let $\delta = \vect{x}[j] - \vect{x}[i]$. We know that $\delta \in [0, 2m - \vect{x}[i]]$, because $ \vect{x}[j] \leq 2m$ and $\vect{x}[j] \geq \vect{x}[i]$. Let us rewrite the equation as:
\[ \vect{x}[i]^p - (2m - \vect{x}[i])^p =  (\vect{x}[i] + \delta)^p - (2m - \vect{x}[i] - \delta)^p\]

The equality holds if and only if $\delta = 0$, i.e., $\vect{x}[i] = \vect{x}[j]$. To see why, first observe that the equation holds when $\delta = 0$. Moreover, since the RHS strictly increases as $\delta$ increases (recall that $0 \leq \delta \leq 2m - \vect{x}[i]$), $\delta = 0$ must be the only solution. Thus we obtain that the bisector is $ \{ \vect{x} \in [0,2m]^{m-1} \mid \vect{x}[i] = \vect{x}[j]\}$.

Now we compute the bisector of $E(i)$ and $E(m)$ restricted to $[0,2m]^{m-1}$. This is the set of points $\vect{x} \in [0,2m]^{m-1}$ such that:
\begin{align*}
	\dis{E(i) - \vect{x}}_p &= \dis{E(m) - \vect{x}}_p &\iff \\
	\left(\sum_{z = 1}^{m-1} |E(i)[z] - \vect{x}[z]|^p\right)^{\frac{1}{p}} &= \left(\sum_{z = 1}^{m-1} |E(m)[z] - \vect{x}[z]|^p\right)^{\frac{1}{p}} &\iff \\
	\sum_{z = 1}^{m-1} |E(i)[z] - \vect{x}[z]|^p &= \sum_{z = 1}^{m-1} |E(m)[z] - \vect{x}[z]|^p &\iff \\
	\left(\sum_{z \in [m-1] \setminus \{i\}} |\vect{x}[z]|^p\right) + |2m - \vect{x}[i]|^p &= \sum_{z \in [m-1]} |\vect{x}[z]|^p &\iff \\
	\left(\sum_{z \in [m-1] \setminus \{i\}} \vect{x}[z]^p\right) + (2m - \vect{x}[i])^p &= \sum_{z \in [m-1]} \vect{x}[z]^p &\iff \\
	(2m - \vect{x}[i])^p &= \vect{x}[i]^p  & \iff \\
	2m - \vect{x}[i] &= \vect{x}[i] & \iff \\
	\vect{x}[i] & = m.
\end{align*}  

Thus the bisector of $E(i)$ and $E(m)$ restricted to $[0,2m]^{m-1}$ is indeed $ \{ \vect{x} \in [0,2m]^{m-1} \mid \vect{x}[i] = m\}$.   
\end{proof}

Bisectors uniquely define the areas where the voters have the same preferences over the alternatives.
Since the location of the bisectors does not depend on $p$ and $E$ is a \dManhattan[(m-1)]-embedding, i.e., an embedding under $1$-norm, the embedding $E$ must also be an embedding under $p$-norm for every $p \in \mathds{R}_{\geq 1}$. This concludes the proof.

However, the construction of \cref{thm:d=m-1->Manhattan} does not work for $p = \infty$. In the proof of \cref{thm:d=m-1->Manhattan}, our alternative placement ensured that the bisector locations do not depend on the $p$-norm as long as $p$ is in $\mathds{R}_{\geq 1}$. However, when $p = \infty$, the bisector locations are no longer the same.
We show a different construction for this case:

		\newcommand{\secondrank}{\hat{r}}
		
		We may again assume that $m \geq 2$, because a profile with one alternative can be embedded trivially.
		
		The idea of this proof is similar to the one for \cref{thm:d=m-1->Manhattan}: We embed the alternatives on their own axes, except for the alternative $m$. This allows us to control the distance to the different alternatives independently.
		However, as we do this, we can no longer embed $m$ in the origin, because any sphere that contains two alternatives other than $m$ must also contain the origin. Instead we embed $m$ a bit further away from the origin.
		We also need to use different values for the coordinates than in \cref{thm:d=m-1->Manhattan}.
		
		We again embed the alternatives on their own axes as follows: for every $j \in [m], z \in [m - 1]$ let $E(j)[z] = 20m$ if $z = j$, and $E(j)[z] = 0$ otherwise. The location of $m$ is different: Let $E(m)[z] = -12m$ for every $z \in [m-1]$.
		
		We again embed voters on the axes based on their rank of the respective alternatives. 
		Let $v_i \in \vvv$ be an arbitrary voter. For every dimension $z \in [m - 1]$
		\[E(v_i)[z] =
		\begin{cases}5m - \rank_i(z), &\text{ if } \rank_i(z) < \rank_i(m),\\
			2m - \rank_i(z), &\text{ if } \rank_i(z) > \rank_i(m),\\
			3m & \text{ if } \rank_i(z) = \rank_i(m) \text{ and } \rank_i(m) > 0,\\
			4m & \text{ if } \rank_i(z) = \rank_i(m) \text{ and } \rank_i(m) = 0.
		\end{cases}\]

		Now let us compute the distances between $v_i$ and the alternatives. For every alternative $j \in \aaa \setminus \{m\}$: 
		\begin{align}
			\Maxdis{E(v_i) - E(j)} &= \max_{z \in [m-1]}|E(v_i)[z] - E(j)[z]|
			= 20m - E(v_i)[j] \label{eq:m+1maxcoord}\\ 
			&= \begin{cases}15m + \rank_i(j), &\text{ if } \rank_i(j) < \rank_i(m)\\
				18m + \rank_i(j), &\text{ if } \rank_i(j) > \rank_i(m),\\
				17m, & \text{ if } \rank_i(j) = \rank_i(m) \text{ and } \rank_i(m) > 0,\\
				16m  & \text{ if } \rank_i(j) = \rank_i(m) \text{ and } \rank_i(m) = 0.
			\end{cases} \label{eq:m+1maxdist}
		\end{align}
		Equation~\eqref{eq:m+1maxcoord} follows from the fact that for every dimension $z \in [m-1]$, we have that $|E(v_i)[z] - 20m| > |E(v_i)[z] - 0|$. Also recall that $E(j)[z] = 20m$ if $z = j$ and $E(j)[z] = 0$ otherwise.

		For $m$, we have that
		\[
		\Maxdis{E(v_i) - E(m)} = \max_{z \in [m-1]}|E(v_i)[z] - E(m)[z]|
		= 12m + \max_{z \in [m-1]}E(v_i)[z].
		\]
		
		We proceed to show that $E$ is a \dMax embedding of $\ppp$.
		Let $v_i \in \vvv$ be an arbitrary voter, and let $j, k \in \aaa$ be an arbitrary pair of alternatives.

		If $m \notin \{j,k\}$, then it is clear from Equation~\eqref{eq:m+1maxdist} that $j \succ_i k$ if and only if $\Maxdis{E(v_i) - E(j)} < \Maxdis{E(v_i) - E(k)}$, as the distance increases as the rank increases. Similarly, $j \sim_i k$ if and only if $\Maxdis{E(v_i) - E(j)} = \Maxdis{E(v_i) - E(k)}$.
		
		We continue to the case where $m \in \{j,k\}$. Without loss of generality, let $m = k$.
		
		We distinguish between two cases:
		\begin{description}
			\item[Case 1:] $\rank_i(m) > 0$. This implies that there exists an alternative $s \in \aaa \setminus \{m\}$ with $\rank_i(s) = 0$. Thus we have that $\max_{z \in [m-1]}E(v_i)[z] = 5m$, and therefore $\Maxdis{E(v_i) - E(m)} =  12m + 5m = 17m$.
			
			If $v_i \colon j \succ m$, then $\Maxdis{E(v_i) - E(j)} \stackrel{\eqref{eq:m+1maxdist}}{\leq} 16m < 17m = \Maxdis{E(v_i) - E(m)}$, as required. Similarly, if $v_i \colon m  \succ j$, then $\Maxdis{E(v_i) - E(j)} \stackrel{\eqref{eq:m+1maxdist}}{\geq} 18m > 17m =\Maxdis{E(v_i) - E(m)}$, as required.
			If $v_i \colon m \sim j$, then  $\Maxdis{E(v_i) - E(j)} \stackrel{\eqref{eq:m+1maxdist}}{=} 17m = \Maxdis{E(v_i) - E(m)}$, as required.
			
			\item[Case 2:] $\rank_i(m) = 0$. 
			
			Since there is no alternative $s \in \aaa \setminus \{m\}$ such that $\rank_i(s)= 0$, we have that $\max_{z \in [m-1]}E(v_i)[z] = 4m$, and thus $\Maxdis{E(v_i) - E(m)} =  12m + 4m = 16m$.
			
			Since $\rank_i(m) = 0$, for every $j \in [m-1]$, $v_i \colon m \succeq j$.
			
			If $v_i \colon m \succ j$, then
			we observe that $\Maxdis{E(v_i) - E(j)} \stackrel{\eqref{eq:m+1maxdist}}{\geq} 18m > 16m =\Maxdis{E(v_i) - E(m )}$, as required.
			If $v_i \colon m \sim j$, then $\Maxdis{E(v_i) - E(j)} \stackrel{\eqref{eq:m+1maxdist}}{=} 16m =\Maxdis{E(v_i) - E(m )}$, as required.
			
		\end{description}

\end{proof}

} %

\section{Forbidden Substructures for 2-Manhattan Embeddings}\label{sec:vconfigs}
In this section, we develop the central technical contributions of this paper: forbidden substructures that constrain how voters can be placed in any \dManhattan[2] embedding.
We introduce two types of configurations---the \bet-configuration and the \ext-configuration---that describe preference patterns among three voters preventing certain relative placements.
We also introduce the \worstinconsistentconfig, which interacts with the geometry of \dMax[2] bounding boxes.
In \cref{subsec:technical}, we prove that these configurations impose necessary conditions on \dManhattan[2] embeddings.
These results are then applied to derive our non-embeddability results in \cref{sec:Manhattan-negative}.

We begin by defining three possible geometric relationships between three voters in a 2-dimensional embedding. In \cref{subsec:technical}, we show that if certain preference patterns (the \bet- and \ext-configurations, \cref{def:3voters-forbidden-profiles-B,def:3voters-forbidden-profiles-E}) are present among three voters, then specific properties from \cref{def:3voters-configurations} below are ruled out, constraining the set of feasible embeddings.

\begin{definition}[\bet- and \ext-properties]\label{def:3voters-configurations}
  Let $\ppp$ be a profile containing at least~$3$ voters called $u,v,w$ and let $E$ be an embedding for~$\ppp$. Then, $E$ satisfies

\begin{compactitem}[--]
  \item  the \myemph{$(v,u,w)$-\bet}-property\footnote{\bet\ stands for ``between''} if
  $E(v) \in \BB(E(u),E(w))$ and %
  
  \item  the \myemph{$(v,u,w)$-\ext}-property\footnote{\ext\ stands for ``external} if
  there exists $(i,j)$ with $\{i,j\}=\{1,2\}$
  such that
  \begin{align*}
    &\min\{E(v)[i], E(w)[i]\} \le E(u)[i] \le \max\{E(v)[i], E(w)[i]\} \quad \text{ and }\\
    &\min\{E(u)[j], E(v)[j]\}  \le E(w)[j] \le \max\{E(u)[j], E(v)[j]\}.
  \end{align*}
\end{compactitem}
See \cref{fig:config-v3} for an illustration of the two properties.
If $E$ does not satisfy the $(v,u,w)$-\bet-property (-\ext-property) we say it \myemph{violates} the $(v,u,w)$-\bet-property (resp. -\ext-property).

For brevity's sake, by symmetry, we omit voters~$u$ and $w$ and just speak of the \myemph{$v$-\bet}-property (resp.\ \myemph{$v$-\ext}-property)
if $u,v,w$ are the only voters contained in~$\ppp$ and $E$ satisfies the $(v,u,w)$-\bet-property (resp.\ the $(v,w,u)$-\bet-property). 
\end{definition}

\begin{figure}  \captionsetup[subfigure]{justification=centering}
  \centering 
\begin{tikzpicture}[black]
    \drawgridA
    \foreach \x / \y / \n / \nn / \typ / \p / \dx / \dy in {2/2/u/u/voter/{above left}/-2/-1, 3/3/v/v/voterV/above left/-1/-1, 4/4/w/w/voter/above left/-1/-1} {
      \node[\typ] at (\x\y) (\n) {};
      \node[\p = \dx pt and \dy pt of \n] {$\nn$};
    }
    \drawreg
    \begin{pgfonlayer}{background}
      \addN
    \end{pgfonlayer}
    \node[below = 0ex of 31] {(BE)};
  \end{tikzpicture}~
  \begin{tikzpicture}[black]
    \drawgridA
    \foreach \x / \y / \n / \nn / \typ / \p / \dx / \dy in {2/4/u/u/voter/{above left}/-2/-1, 3/3/v/v/voterV/above left/-1/-1, 4/2/w/w/voter/above left/-2/-4} {
      \node[\typ] at (\x\y) (\n) {};
      \node[\p = \dx pt and \dy pt of \n] {$\nn$};
    }
    \drawreg
    \node[below = 0ex of 31] {(BE)};
    \begin{pgfonlayer}{background}
    \end{pgfonlayer}
\end{tikzpicture}~
\begin{tikzpicture}[black]
    \drawgridA
    \foreach \x / \y / \n / \nn / \typ / \p / \dx / \dy in {2/2/w/w/voter/{above left}/-2/-1, 3/3/v/v/voterV/above left/-1/-1, 4/4/u/u/voter/above left/-1/-1} {
      \node[\typ] at (\x\y) (\n) {};
      \node[\p = \dx pt and \dy pt of \n] {$\nn$};
    }
    \drawreg
    \node[below = 0ex of 31] {(BE)};
    \begin{pgfonlayer}{background}
    \end{pgfonlayer}
  \end{tikzpicture}~
  \begin{tikzpicture}[black]
    \drawgridA
    \foreach \x / \y / \n / \nn / \typ / \p / \dx / \dy in {2/4/w/w/voter/{above left}/-2/-1, 3/3/v/v/voterV/above left/-1/-1, 4/2/u/u/voter/above left/-2/-4} {
      \node[\typ] at (\x\y) (\n) {};
      \node[\p = \dx pt and \dy pt of \n] {$\nn$};
    }
    \drawreg
    \node[below = 0ex of 31] {(BE)};
    \begin{pgfonlayer}{background}
    \end{pgfonlayer}
  \end{tikzpicture}
  
\begin{tikzpicture}[black]
  \drawgridA
  \foreach \x / \y / \n / \nn / \typ / \p / \dx / \dy in {2/3/u/u/voter/{above left}/-2/-1, 4/2/v/v/voterV/above left/-1/-1, 3/4/w/w/voter/above left/-1/-1} {
    \node[\typ] at (\x\y) (\n) {};
    \node[\p = \dx pt and \dy pt of \n] {$\nn$};
  }
  
  \begin{pgfonlayer}{background}
    \foreach \s / \t in {12/52,13/53,14/54,21/25,31/35,41/45} {
      \path[draw,lines] (\s) edge (\t);
    }
    \drawreg
    \addN
  \end{pgfonlayer}
  \node[below = 0ex of 31] {(EX)}; %
\end{tikzpicture}~
\begin{tikzpicture}[black]
  \drawgridA
  \foreach \x / \y / \n / \nn / \typ / \p / \dx / \dy in {2/3/u/u/voter/{above left}/-2/-1, 4/4/v/v/voterV/above left/-1/-1, 3/2/w/w/voter/above left/-1/-4} {
    \node[\typ] at (\x\y) (\n) {};
    \node[\p = \dx pt and \dy pt of \n] {$\nn$};
  }
  
  \begin{pgfonlayer}{background}
    \foreach \s / \t in {12/52,13/53,14/54,21/25,31/35,41/45} {
      \path[draw,lines] (\s) edge (\t);
    }
    \drawreg
    \begin{pgfonlayer}{background}
    \end{pgfonlayer}
  \end{pgfonlayer}
  \node[below = 0ex of 31] {(EX)}; %
\end{tikzpicture}~
\begin{tikzpicture}[black]
  \drawgridA
  \foreach \x / \y / \n / \nn / \typ / \p / \dx / \dy in {2/3/w/w/voter/{above left}/-2/-3, 4/2/v/v/voterV/above left/-1/-1, 3/4/u/u/voter/above left/-1/-1} {
    \node[\typ] at (\x\y) (\n) {};
    \node[\p = \dx pt and \dy pt of \n] {$\nn$};
  }
  
  \begin{pgfonlayer}{background}
    \foreach \s / \t in {12/52,13/53,14/54,21/25,31/35,41/45} {
      \path[draw,lines] (\s) edge (\t);
    }
    \drawreg
    \begin{pgfonlayer}{background}
    \end{pgfonlayer}
  \end{pgfonlayer}
  \node[below = 0ex of 31] {(EX)}; %
\end{tikzpicture}~
\begin{tikzpicture}[black]
  \drawgridA
  \foreach \x / \y / \n / \nn / \typ / \p / \dx / \dy in {2/3/w/w/voter/{above left}/-2/-1, 4/4/v/v/voterV/above left/-1/-1, 3/2/u/u/voter/above left/-1/-3} {
    \node[\typ] at (\x\y) (\n) {};
    \node[\p = \dx pt and \dy pt of \n] {$\nn$};
  }
  
  \begin{pgfonlayer}{background}
    \foreach \s / \t in {12/52,13/53,14/54,21/25,31/35,41/45} {
      \path[draw,lines] (\s) edge (\t);
    }
    \drawreg
    \begin{pgfonlayer}{background}
    \end{pgfonlayer}
  \end{pgfonlayer}
  \node[below = 0ex of 31] {(EX)}; %
\end{tikzpicture}
~
\begin{tikzpicture}[black]
  \drawgridA
  \foreach \x / \y / \n / \nn / \typ / \p / \dx / \dy in {4/3/u/u/voter/{above left}/-2/-1, 2/2/v/v/voterV/above left/-1/-1, 3/4/w/w/voter/above left/-1/-1} {
    \node[\typ] at (\x\y) (\n) {};
    \node[\p = \dx pt and \dy pt of \n] {$\nn$};
  }
  
  \begin{pgfonlayer}{background}
    \foreach \s / \t in {12/52,13/53,14/54,21/25,31/35,41/45} {
      \path[draw,lines] (\s) edge (\t);
    }
    \drawreg
    \begin{pgfonlayer}{background}
    \end{pgfonlayer}
  \end{pgfonlayer}
  \node[below = 0ex of 31] {(EX)}; %
\end{tikzpicture}~
\begin{tikzpicture}[black]
  \drawgridA
  \foreach \x / \y / \n / \nn / \typ / \p / \dx / \dy in {4/3/u/u/voter/{above left}/-2/-1, 2/4/v/v/voterV/above left/-1/-1, 3/2/w/w/voter/above left/-2/-4} {
    \node[\typ] at (\x\y) (\n) {};
    \node[\p = \dx pt and \dy pt of \n] {$\nn$};
  }
  
  \begin{pgfonlayer}{background}
    \foreach \s / \t in {12/52,13/53,14/54,21/25,31/35,41/45} {
      \path[draw,lines] (\s) edge (\t);
    }
    \drawreg
    \begin{pgfonlayer}{background}
    \end{pgfonlayer}
  \end{pgfonlayer}
  \node[below = 0ex of 31] {(EX)}; %
\end{tikzpicture}~
\begin{tikzpicture}[black]
  \drawgridA
  \foreach \x / \y / \n / \nn / \typ / \p / \dx / \dy in {4/3/w/w/voter/{above left}/-2/-1, 2/2/v/v/voterV/above left/-1/-1, 3/4/u/u/voter/above left/-1/-1} {
    \node[\typ] at (\x\y) (\n) {};
    \node[\p = \dx pt and \dy pt of \n] {$\nn$};
  }
  
  \begin{pgfonlayer}{background}
    \foreach \s / \t in {12/52,13/53,14/54,21/25,31/35,41/45} {
      \path[draw,lines] (\s) edge (\t);
    }
    \drawreg
    \begin{pgfonlayer}{background}
    \end{pgfonlayer}
  \end{pgfonlayer}
  \node[below = 0ex of 31] {(EX)}; %
\end{tikzpicture}~
\begin{tikzpicture}[black]
  \drawgridA
  \foreach \x / \y / \n / \nn / \typ / \p / \dx / \dy in {4/3/w/w/voter/{above left}/-2/-1, 2/4/v/v/voterV/above left/-1/-1, 3/2/u/u/voter/above left/-1/-3} {
    \node[\typ] at (\x\y) (\n) {};
    \node[\p = \dx pt and \dy pt of \n] {$\nn$};
  }
  
  \begin{pgfonlayer}{background}
    \foreach \s / \t in {12/52,13/53,14/54,21/25,31/35,41/45} {
      \path[draw,lines] (\s) edge (\t);
    }
    \drawreg
    \begin{pgfonlayer}{background}
    \end{pgfonlayer}
  \end{pgfonlayer}
  \node[below = 0ex of 31] {(EX)}; %
\end{tikzpicture}
\caption{Two possible types of embeddings illustrating the properties in \cref{def:3voters-configurations} (the numbering will be used in the proofs of Lemmas \ref{clm:n5-m4} and \ref{lem:ext-property}). (BE) means ``between'' while (EX) ``external''.}\label{fig:config-v3}
\end{figure}

Note that there are four possible types of embeddings which satisfy the $(v,u,w)$-\bet-property\ (see the first row in \cref{fig:config-v3}).
They are however equivalent up to mirroring.
Analogously, there are eight possible types of embeddings which satisfy the $(v,u,w)$-\ext-property (see the last two rows in \cref{fig:config-v3}).
Moreover, any embedding for three voters~$u, v, w$ must satisfy the $u$-, $v$- or $w$-\ext-property, or the $u$-, $v$- or $w$-\bet-property, although it may satisfy more than one of these (consider for example three voters at the same point).
However, each of these embeddings satisfying the $(v,u,w)$-\bet-property (resp.\ $(v,u,w)$-\ext-property) forbids certain types of preference structures. %
The following two configurations describe preferences whose existence precludes an embedding from satisfying either the \bet-property or the \ext-property for some voters,
as we will show in \cref{lem:bet-property,lem:ext-property}. 

Intuitively, a $(v,u,w)$-\bet-configuration forbids voter $v$ to be embedded within the bounding box of voters $u$ and~$w$:
\begin{definition}[\bet-configurations]\label{def:3voters-forbidden-profiles-B}
  A profile~$\ppp$ with $3$ voters~$u,v,w$ and $3$ alternatives~$a,b,x$ is 
  a \myemph{$(v,u,w)$-\bet-configuration}
  if the following holds:
   \begin{align*}
    u,w\colon b \succ x \succ a,  \text{ and }  v\colon  a \succ x \succ b. %
  \end{align*}%
\end{definition}

A $(v,u,w)$-\ext-configuration forbids $v$ to be embedded outside of the bounding box of voters~$u$ and $w$:
\begin{definition}[\ext-configurations]\label{def:3voters-forbidden-profiles-E}
  A profile~$\ppp$ with $3$ voters~$u,v,w$ and $6$ alternatives~$x, a,b,c,d,e$ ($c,d,e$ not necessarily distinct) is
  a \myemph{$(v,u,w)$-\ext-configuration}
  if the following holds:
  \begin{alignat*}{4}
    u\colon & a \succ x \succ b,  \quad & c\succ x, &\quad& & d\succ x\\
    v\colon & \{a,b\} \succ x, &  &&& x \succ \{d,e\},\\
    w\colon & b \succ x \succ a, &  c\succ x, &&& e\succ x.
  \end{alignat*}%
\end{definition}

\begin{example}\label{ex:non-bet-ext}
  \setcounter{betcounter}{\themyprofilecounter} \stepcounter{myprofilecounter}
  \setcounter{extcounter}{\themyprofilecounter} \stepcounter{myprofilecounter}
  Consider two profiles~$\ppp_{\thebetcounter}$
  and $\ppp_{\theextcounter}$ which satisfy the following:  
  \begin{alignat*}{6}
\ppp_{\thebetcounter}\colon  &  v_1\colon & 1 \succ 2 \succ 3, &\qquad \ppp_{\theextcounter}\colon & &v_1 \colon & \{1,2\} \succ 3 \succ 4, \\ 
     &v_2\colon & 3 \succ 2 \succ 1, &&&v_2 \colon & \{1,4\} \succ 3 \succ 2, \\
     &v_3\colon & 3 \succ 2 \succ 1,&&&v_3 \colon & \{2, 4\} \succ 3 \succ 1.
  \end{alignat*}
  
  \noindent Clearly, 
  $\ppp_{\thebetcounter}$ is a $(v_1,v_2,v_3)$-\bet-configuration.
  Further, one can verify that~$\ppp_{\theextcounter}$ contains a $(v_1,v_2,v_3)$-, $(v_2,v_1,v_3)$-, and $(v_3,v_1,v_2)$-\ext-configuration,
  by setting $(a,b,x,c,d,e)\coloneqq (1,2,3,4,4,4)$, $(a,b,x,c,d,e)\coloneqq(1,4,3,2,2,2)$, and $(a,b,x,c,d$, $e)\coloneqq (2,4,3,1,1,1)$, respectively.
\end{example}

The next configuration is a restriction of the worst-diverse configuration.
The latter is used to characterize the so-called single-peaked preferences~\citep{BH11}.

\begin{definition}[\Worstinconsistentconfig]\label{def:threeworst}
  A profile~$\ppp$ is an \myemph{\worstinconsistentconfig} if for \emph{every} triple of alternatives $\{x, y, z\} \subseteq \aaa$ there are three voters $u, v, w \in \vvv$ which form a worst-diverse configuration, i.e., their preferences satisfy
  \begin{align*}
    u\colon &\{x, y\} \succ z, &   v\colon &\{x, z\} \succ y,  & w\colon &\{y, z\} \succ x.
  \end{align*}
\end{definition}

\subsection{Necessary Conditions from Forbidden Substructures}\label{subsec:technical}

In this subsection, we show how the voter configurations restrict the possible \dManhattan[2] embeddings.
For brevity's sake, given an embedding~$E$ and a voter~$v\in \vvv$ (resp.\ an alternative~$a\in \aaa$), we use boldface~$\pv$ (resp.~$\pa$) to denote the embedding~$E(v)$ (resp.\ $E(a)$).

\begin{lemma}\label{lem:two-votes}
  Let $\ppp$ be a profile admitting a \dManhattan[2] embedding~$E$.
  For every two voters~$r,s$ and two alternatives~$x,y$ the following holds:
  \begin{compactenum}[(i)]
    \item \label{lem:not-inside} If $r,s\colon y \succ x$,
  then $\px\notin \BB(\pr, \ps)$.
  \item \label{lem:not-outside-corner}  If $r\colon x \succ y$ and $s\colon y\succ x$, 
  then $\ps\notin \BB(\pr, \px)$.\footnote{Equivalently, for each $\Pi \in \{\dNE$, $\dNW,\dSE,\dSW\}$, if $r\colon x \succ y$ and $s\colon y\succ x$ and $\ps \in \Pi(\pr)$, then $\px \notin \Pi(\ps)$.}
  \end{compactenum}
\end{lemma}

\begin{figure}
  \centering
     \begin{tikzpicture}[scale=0.3, black]
    \foreach \x / \y / \n / \nn / \typ / \p / \dx / \dy in {0/0/r/r/voter/{above left}/-2/-1, 3/2/x/x/alter/above left/-1/-1, 6/4/s/s/voter/above left/-1/-1,0/4/rs/\;/hide/above left/-1/-1, 6/0/sr/\;/hide/above left/-1/-1} {
      \node[\typ] at (\x , \y) (\n) {};
      \node[\p = \dx pt and \dy pt of \n] {$\nn$};
    }
   
    \coordinate (start) at (-6,-4);
    \coordinate (end) at (-6,-7);
    \coordinate (ss) at (6,0);
    \coordinate (ee) at (0,7);
    
    \begin{pgfonlayer}{background}
    	\drawcirclefill{x}{r}{fill=green!20,draw=green};
    	\drawcirclefill{x}{s}{fill=green!20,draw=green};
    \end{pgfonlayer}

    \draw[linesdark] (0,0) -- (0,4) -- (6,4) -- (6,0)  -- cycle;
\end{tikzpicture}\quad\quad
\begin{tikzpicture}[scale=0.3, black]    
        \foreach \x / \y / \n / \nn / \typ / \p / \dx / \dy in {2/2/r/r/voter/{above left}/-2/-1, 4/3/s/s/voter/above left/-1/-1, 6/4/x/x/alter/above right/-1/-1, 2/4/rx/\;/hide/above left/-1/-1, 6/2/xr/\;/hide/above left/-1/-1} {
      \node[\typ] at (\x , \y) (\n) {};
      \node[\p = \dx pt and \dy pt of \n] {$\nn$};
    }
    \begin{pgfonlayer}{background}
    	\drawcirclefill{x}{r}{fill=red!20,draw=red};
    	\drawcirclefill{x}{s}{fill=green!20,draw=green};
    \end{pgfonlayer}
    
    \draw[linesdark] (2,2) -- (2,4) -- (6,4) -- (6,2) -- cycle;
\end{tikzpicture}

     \caption{Some illustrations for the idea of \cref{lem:two-votes}. Both figures have two voters $r$ and $s$ and circles of radius $\Mdis{\pr - \px}$ and $\Mdis{\ps - \px}$ centered at $\pr$ and $\ps$, respectively. Left: If both $r$ and $s$ satisfy the premises of \cref{lem:two-votes}\eqref{lem:not-inside} and $\px \in \BB(\pr, \ps)$, then $y$ needs to be strictly inside both circles, which is impossible. Right: If both $r$ and $s$ satisfy the premises of \cref{lem:two-votes}\eqref{lem:not-outside-corner} and $\ps \in \BB(\pr, \px)$, then $y$ needs to be strictly inside the inner circle and outside of the outer circle, which is impossible.
     }\label{fig:two-voters-idea}
\end{figure}

{\begin{proof}
  Let $\ppp$, $E$, $r,s$, and $x,y$ be as defined.
  Both statements follow from using simple calculations and the triangle inequality of Manhattan distances. The idea is also illustrated in \cref{fig:two-voters-idea}.
  
  For Statement~\eqref{lem:not-inside}, suppose, towards a contradiction, that $r,s\colon y \succ x$ and $\px\in \BB(\pr, \ps)$.
  By the definition of Manhattan distances, this implies that
  \begin{align}
    \Mdis{\pr-\px}+\Mdis{\px-\ps} = \Mdis{\pr-\ps}.\label{eq-lemma1-1}
  \end{align}
  By the preferences of voters~$r$ and $s$ we infer that
$    \Mdis{\pr-\py}+\Mdis{\ps-\py} < \Mdis{\pr-\px}+\Mdis{\ps-\px} \stackrel{\eqref{eq-lemma1-1}}{=} \Mdis{\pr-\ps},$
  a contradiction to the triangle inequality of $\Mdis{\cdot}$.

  For Statement~\eqref{lem:not-outside-corner},  suppose, towards a contradiction, that $r\colon x \succ y$ and $s\colon y\succ x$ and $\ps\in \BB(\pr, \px)$.
  By the definition of Manhattan distances, this implies that
  \begin{align}
    \Mdis{\pr-\px}=\Mdis{\pr-\ps} + \Mdis{\ps-\px}.\label{eq-lemma1-2}
  \end{align}
   By the preferences of voters~$r$ and $s$ we infer that %
   \begin{align*}
     \Mdis{\pr-\ps}+\Mdis{\ps-\py}  < \Mdis{\pr-\ps}+\Mdis{\ps-\px}
     \stackrel{\eqref{eq-lemma1-2}}{=} \Mdis{\pr-\px} < \Mdis{\pr-\py},
   \end{align*}
  a contradiction to the triangle inequality of $\Mdis{\cdot}$.\hfill~ 
\end{proof}
}
\noindent The following is a summary of the differences between the coordinates wrt.\ the preferences.

\begin{observation}\label{obs:pref-relation}
  Let profile~$\ppp$ admit a \dManhattan[2] embedding~$E$. %
  For each voter~$s$ and each two alternatives~$x,y$ with $s\colon x\pref y$,
  the following holds:

  \begin{compactenum}[(i)]
    \item\label{obs:pref-NE} If $\py\in \dNE(\ps)$, then $\py[1]+\py[2]>\px[1]+\px[2]$.
    \item\label{obs:pref-NW} If $\py\in \dNW(\ps)$, then $-\py[1]+\py[2]>-\px[1]+\px[2]$.
    \item\label{obs:pref-SE} If $\py\in \dSE(\ps)$, then $\py[1]-\py[2]>\px[1]-\px[2]$.
    \item\label{obs:pref-SW} If $\py\in \dSW(\ps)$, then $-\py[1]-\py[2]>-\px[1]-\px[2]$.
  \end{compactenum}
\end{observation}

\begin{proof}
  All proofs are straightforward by evoking the definition of Manhattan embedding.
  Hence, we only showcase how to prove the first statement.
  Let $\ppp,E,s,x,y$ be as defined. %
  Assume that $\py\in \dNE(\ps)$.
  Then, by the Manhattan property and the fact that $s\colon x\pref y$, it follows that
  \begin{align*}
    (\py[1]-\ps[1])+(\py[2]-\ps[2])
    & = \Mdis{\py-\ps}  > \Mdis{\px-\ps}  
     = | \px[1]-\ps[1]|  + | \px[2]-\ps[2]| \\
    &\ge   (\px[1]-\ps[1]) + (\px[2]-\ps[2])\\
    \Rightarrow     \py[1] + \py[2] & >   \px[1]+\px[2],
  \end{align*}
  as desired.\hfill~ 
\end{proof}

The next technical lemma excludes two alternatives from being put in the same quadrant region of some voters; see \cref{fig:bet-property} for an illustration.

\begin{lemma}\label{lem:bet-property-Ntogether}
  Let $\ppp$ be a profile admitting a \dManhattan[2] embedding~$E$.
  Let $r,s,t$ and~$x,y$ be $3$ voters and $2$ alternatives in~$\ppp$, respectively.
  The following holds.
  \begin{compactenum}[(i)]
  \item\label{lem:Ntogether1} For each $\Pi \in \{\dNE$, $\dNW,\dSE,\dSW\}$, it holds that if  $r\colon x\succ y$ and $s\colon y\succ x$ and $\px \in \Pi(\ps)$, then $\py\notin \Pi(\pr)$.
  \item\label{lem:Ntogether2} For each $\Pi \in \{\dNW, \dSE\}$, it holds that if $r,t\colon x\succ y$, $s\colon y \succ x$, $\pr\in \dSW(\ps)$, $\pt \in \dNE(\ps)$, and $\px\in \Pi(\ps)$, then $\py\notin \Pi(\ps)$.
  \end{compactenum}
\end{lemma}

\begin{figure}
  \centering
  \begin{tikzpicture}[black, scale=0.4]

    \foreach \x / \y / \n / \nn / \typ / \p / \dx / \dy in {0/0/r/r/voter/{above left}/-2/-1, 2/3/x/x/alter/above right/-1/-1, 3/2/s/s/voter/above right/-1/-1} {
      \node[\typ] at (\x , \y) (\n) {};
    }
    
    \coordinate (start) at (-6,-4);
    \coordinate (end) at (-6,-7);
    \coordinate (ss) at (9,0);
    \coordinate (ee) at (0,9);
    
    \begin{pgfonlayer}{background}
        \drawcirclefill{x}{r}{fill=red!20,draw=red};
        \drawcirclefill{x}{r}{redpattern}
    	\drawcirclefill{x}{s}{draw=green, fill=green, fill opacity = 0.2};
        \draw[linesdark] (-6,2) -- (6,2);
        \draw[linesdark] (3,-6) -- (3,6);
        \filldraw[pattern={Lines[
      distance=2mm,
      angle=45,
      line width=0.5mm]},draw=none, pattern color=black!30] (-6,6) -- (-6,0) -- (0,0) -- (0,6) -- cycle;
    \end{pgfonlayer}

    \foreach \x / \y / \n / \nn / \typ / \p / \dx / \dy in {0/0/r/r/voter/{above left}/-2/-1, 2/3/x/x/alter/above right/-1/-1, 3/2/s/s/voter/above right/-1/-1} {
      \node[\typ] at (\x , \y) (\n) {};
      \node[\p = \dx pt and \dy pt of \n, inner sep=1pt] {$\nn$};
    }

    \node[] at (-4.5,5) {$\dNW$};
    
  \end{tikzpicture}\qquad\qquad
\begin{tikzpicture}[black, scale=0.4]    
    \foreach \x / \y / \n / \nn / \typ / \p / \dx / \dy in {1/-0.5/r/r/voter/{above left}/-2/-1, 4/-1/x/x/alter/right/-1/-1, 3/2/s/s/voter/above right/-1/-1, 4/4/t/t/voter/{above left}/-2/-1} {
      \node[\typ] at (\x , \y) (\n) {};
    }
    
    \coordinate (start) at (-6,-4);
    \coordinate (end) at (-6,-7);
    \coordinate (ss) at (9,0);
    \coordinate (ee) at (0,9);

    \begin{pgfonlayer}{background}
      \draw[linesdark] (-4,2) -- (10,2);
      \draw[linesdark] (3,10) -- (3,-5);
      \drawcirclefill{x}{r}{draw=red, fill=red, fill opacity = 0.2};
      \drawcirclefill{x}{r}{redpattern}
      \drawcirclefill{x}{t}{draw=red, fill=red, fill opacity = 0.2};
      \drawcirclefill{x}{t}{redpatterntilt}
      \drawcirclefill{x}{s}{draw=green, fill=green, fill opacity = 0.2};
    \end{pgfonlayer}
    
    \foreach \x / \y / \n / \nn / \typ / \p / \dx / \dy in {1/-0.5/r/r/voter/{above left}/-2/-1, 4/-1/x/x/alter/right/-1/-1, 3/2/s/s/voter/above right/-1/-1, 4/4/t/t/voter/{above left}/-2/-1} {
      \node[\p = \dx pt and \dy pt of \n] {$\nn$};
    }
   
    \node[] at (9,-4) {$\dSE$};
    
\end{tikzpicture}
\caption{Some illustrations for \cref{lem:bet-property-Ntogether}.  Left: The green circle around $s$ has the radius $\Mdis{\ps - \px}$ and the red circle around $r$ has the radius $\Mdis{\pr - \px}$. Due to the premises of \cref{lem:bet-property-Ntogether}\eqref{lem:Ntogether1}, the alternative $y$ has to be inside the green circle but outside of the lined red circle. One can verify that this area never intersects with $\dNW(\pr)$, if $\px \in \dNW(\ps)$. Right: The green circle around $s$ has the radius $\Mdis{\ps - \px}$ and the lined red circles around $r$ and $t$ have the radius $\Mdis{\pr - \px}$ and $\Mdis{\pt - \px}$, respectively. Due to the premises of \cref{lem:bet-property-Ntogether}\eqref{lem:Ntogether2}, the alternative $y$ has to be inside the green circle but outside of the lined red circles. One can verify that this area never intersects with $\dSE(\ps)$, if $\px \in \dSE(\ps)$.}\label{fig:bet-property}
\end{figure}

\begin{proof}
  Let $\ppp,E,r,s,t,x,y$ be as defined. %
  The first statement follows directly from applying \cref{obs:pref-relation}.
  Hence, we only prove the case with~$\Pi=\dNW$.
  For a contradiction, suppose that $\px \in \dNW(\ps)$ and $\py \in \dNW(\pr)$.
  Since $r\colon x\pref y$  and $\py \in \dNW(\pr)$, by \cref{obs:pref-relation}\eqref{obs:pref-NW} we have that $\py[2]-\py[1] > \px[2]-\px[1]$.
  Since $s \colon y \pref x$ and $\px \in \dNW(\ps)$, by \cref{obs:pref-relation}\eqref{obs:pref-NW} we have that $\px[2]-\px[1] > \py[2]-\py[1]$.
  However, these two statements contradict each other.
  
  \noindent  Statement~\eqref{lem:Ntogether2}: We only show the case with $\Pi=\dNW$ as the other case is symmetric. For a contradiction, suppose that $\px, \py \in \dNW(\ps)$. %
  Since $r,t\colon x \succ y$, $s\colon y \succ x$, $\px\in \dNW(\ps)$, by the first statement,
  we have $\py \notin \dNW(\pr)\cup \dNW(\pt)$.
  However, since $\py\in \dNW(\ps)$, it follows that $\py\in \BB(\pr,\pt)$,
  a contradiction to \cref{lem:two-votes}\eqref{lem:not-inside}.\hfill~ 
\end{proof}

The next two lemmas specify the relation between a \bet-configuration and the \bet-property, and between a \ext-configuration and the \ext-property, respectively.

\begin{lemma}\label{lem:bet-property}
  If a profile contains a $(v,u,w)$-\bet-configuration, then no \dManhattan[2] embedding satisfies the $(v,u,w)$-\bet-property.
\end{lemma}

  \begin{proof}
  Suppose, towards a contradiction, that $\ppp$ is a profile which contains a $(v,u,w)$-\bet-configuration
  and admits a \dManhattan[2] embedding~$E$, such that $E$ satisfies the $(v,u,w)$-\bet-property, for $3$ voters~$u,v,w$.
  Let $a,b,x$ be the $3$ alternatives defined in the $(v,u,w)$-\bet-configuration~(see \cref{def:3voters-forbidden-profiles-B}).
  By symmetry and by the preferences of~$u$ and $w$, the embedding~$E$ corresponds to one of the four possible types of illustrations labeled with (BE) in \cref{fig:config-v3}.
  Since they are equivalent up to mirroring, let us assume that $E$ corresponds to the top left illustration of \cref{fig:config-v3}.  
  Since there are $3$ voters, we can divide the two-dimensional space into $16$ subspaces by drawing a vertical and horizontal line through each voter's embedded point.
  We enumerate these regions and use $R_i$ to refer to region~$i$, $i\in [16]$.

  First, using \cref{lem:two-votes}\eqref{lem:not-inside} (setting $(r,s,y)\coloneqq (u,w,b)$),
  we infer that alternative~$x$ cannot be embedded in~$R_6$, $R_7$, $R_{10}$, or $R_{11}$.
  Moreover, using \cref{lem:two-votes}\eqref{lem:not-outside-corner} (setting $(r,s,y)\coloneqq (u,v,a)$),
  we infer that alternative~$x$ cannot be embedded in $R_3$, $R_4$, $R_7$, or $R_8$.
  Similarly, using \cref{lem:two-votes}\eqref{lem:not-outside-corner} (setting $(r,s,y)\coloneqq (w,v,b)$),
  we infer that alternative~$x$ cannot be embedded in~$R_9$, $R_{10}$, $R_{13}$, or $R_{14}$.
  This implies that $x$ is in one of the regions~$R_1$, $R_2$, $R_5$, $R_{12}$, $R_{15}$ or $R_{16}$.
  By exchanging the two coordinates and the roles of $u$ and $w$ and the roles of $a$ and $b$,
  respectively, we know that if $E$ embeds alternative~$x$ in~$R_5$ (resp.\ $R_1$ or~$R_2$),
  then there exists another Manhattan embedding which embeds~$x$ in~$R_{15}$ (resp.\ $R_{16}$ or $R_{12}$),
  and vice versa.
  Hence, without loss of generality, assume that $E$ embeds~$x$ in~$R_1$, $R_2$, or $R_5$.
  Note that this implies that~$\px \in \dNW(\pv)$.

  Similarly, using \cref{lem:two-votes}\eqref{lem:not-outside-corner} (setting~$(r,s,x,y)=(u,v,b,a)$ and $(r,s,x,y)=(w,v,b,a)$), %
  we infer that $\pv \notin \BB(\pu, \pb)\cup \BB(\pw, \pb)$.
  This implies that $\pb\notin \dNE(\pv)\cup \dSW(\pv)$. 
  Since $\px \in \dNW(\pv)$, by \cref{lem:bet-property-Ntogether}\eqref{lem:Ntogether2}~(wrt.\ alternatives~$x$ and $b$), it follows that $\pb\notin \dNW(\pv)$.
  This implies that $\pb\in \dSE(\pv)$.

  Let us consider alternative~$a$.
  On the one hand, since $u,w\colon x\succ a$ and $v\colon a \succ x$,
  by \cref{lem:two-votes}, 
  it follows that $\pa \notin \BB(\pu,\pw)\cup \dNE(\pw)\cup \dSW(\pu)$.
  Altogether, it follows that $\pa \in \dSE(\pw)\cup\dNW(\pw)\cup \dSE(\pu)\cup \dNW(\pu)$. 
  
  On the other hand, since $v\colon a \succ b$, $u, w\colon b \succ a$, and $\pb \in \dSE(v)$,
  by  \cref{lem:bet-property-Ntogether}\eqref{lem:Ntogether1},
  it follows that $\pa \notin \dSE(u)\cup \dSE(w)$.
  Analogously, since $v\colon a \succ x$, $u, w\colon x \succ a$, and $\px \in \dNW(v)$,
  by  \cref{lem:bet-property-Ntogether}\eqref{lem:Ntogether1},
  it follows that $\pa \notin \dNW(u)\cup \dNW(w)$.
  
  This results in having no place to embed alternative~$a$, a contradiction.\hfill~ 
\end{proof}

\begin{lemma}\label{lem:ext-property}
  If a profile contains a $(v,u,w)$-\ext-configuration, then no \dManhattan[2] embedding satisfies the $(v,u,w)$-\ext-property.
\end{lemma} 

\begin{proof}
  Suppose, for the sake of contradiction, that there exists a profile~$\ppp$
  which contains a $(v,u,w)$-\ext-configuration
  and admits a \dManhattan[2] embedding~$E$ such that $E$ satisfies the $(v,u,w)$-\ext-property, for $3$ voters~$v,u,w$.
  Let $x,a,b,c,d,e$ be the $6$ alternatives defined in the $(v,u,w)$-\ext-configuration~(see \cref{def:3voters-forbidden-profiles-E}).
  Observe that the preferences of $u$ and $w$ are symmetric in the sense that if we exchange the roles of $a$ and $b$, and also the roles of~$d$ and $e$, then we arrive at a new $(v,u,w)$-\ext-configuration\
  for~$\ppp$.
  Hence, up to rotation and mirroring,
  we can assume that embedding~$E$ corresponds to the first embedding of the second row of~\cref{fig:config-v3}.
  Since there are $3$ voters, we can divide the two-dimensional space into $16$ subspaces by drawing a vertical and horizontal line through each voter's embedded point. 
  We enumerate these regions as in the first embedding of the second row of~\cref{fig:config-v3} and use $R_i$ to refer to region~$i$, $i\in [16]$. 
  We aim to show by contradiction that $x$ cannot be embedded in any region.  

  First, using \cref{lem:two-votes}\eqref{lem:not-inside} (setting $(r,s,y)\coloneqq (u,w,c)$),
  we infer that alternative~$x$ cannot be embedded in~$R_6$.
  Analogously, repeatedly using \cref{lem:two-votes}\eqref{lem:not-inside} (setting $(r,s,y)\coloneqq (u,v,a)$ and  $(r,s,y)\coloneqq (v,w,b)$, respectively), we infer that $x$ cannot be embedded in regions~$R_7$, $R_{10}$ or $R_{11}$.
  Further, using \cref{lem:two-votes}\eqref{lem:not-outside-corner} (setting  $(r,s,y)\coloneqq (v,u,d)$),
  we infer that alternative~$x$ cannot be embedded in regions~$R_1$ and $R_5$.
  Again, using \cref{lem:two-votes}\eqref{lem:not-outside-corner} repeatedly (setting $(r,s,y)\coloneqq (v,w,e)$, $(r,s,y)\coloneqq (u,w,b)$, and $(r,s,y)\coloneqq (w,u,a)$, and $(r,s,y)\coloneqq (u,v,b)$, respectively),
  we further infer that alternative~$x$ cannot be embedded in regions~$R_1$--$R_4$, $R_9$, $R_{13}$, and $R_{16}$.

  This implies that $x$ can only be embedded in~$R_8$, $R_{12}$, $R_{14}$, or $R_{15}$.
  To this end, since $\pv\in \dSE(\pu)\cap \dSE(\pw)$,
  by \cref{lem:two-votes}\eqref{lem:not-outside-corner} (setting $(r,s,x,y)=(v,w,a,x)$ and $(r,s,x,y)=(v,u,b,x)$, respectively), we observe that
  \begin{align}\label{eq:ext-property-ab}
    \pa[2]\le \pw[2] \text{ and } \pb[1]\ge \pu[1].
  \end{align}

  If $E$ embeds $x$ in regions~$R_{14}$--$R_{15}$, then
  \begin{align}
  \px \in \dSW(\pv)\cap \dSE(\pu).\label{eq:ext-property-xvu}
  \end{align}
  Since $v\colon a\succ x$, $w\colon x \succ a$,
  by \cref{lem:bet-property-Ntogether}\eqref{lem:Ntogether1},
  it follows that
  $\pa \notin \dSW(\pw)$.
  By~\eqref{eq:ext-property-ab}, it follows that $\pa \in \dSE(\pw)$.
  Since $w\colon x \succ a$ and $u\colon a \succ x$,
  by \cref{lem:bet-property-Ntogether}\eqref{lem:Ntogether1},
  it follows that $\px \notin \dSE(\pu)$, a contradiction to~\eqref{eq:ext-property-xvu}.
  
  Analogously, we also obtain a contradiction if $x$ is embedded in region~$R_8$ or $R_{12}$ by focusing on voter~$u$ and alternative~$b$.
  Assume that
  \begin{align}
    \px \in  \dNE(\pv)\cap \dSE(\pw).\label{eq:ext-property-xvw}
  \end{align}
  Since $v\colon b\succ x$ and $u\colon x \succ b$,
  by  \cref{lem:bet-property-Ntogether}\eqref{lem:Ntogether1},
  it follows that
  $\pb \notin \dNE(\pu)$.
  By~\eqref{eq:ext-property-ab}, it follows that $\pb \in \dSE(\pu)$.
  Since $u\colon x \succ b$ and $w\colon b \succ x$, by \cref{lem:bet-property-Ntogether}\eqref{lem:Ntogether1}, it follows that $\px \notin \dSE(\pw)$, a contradiction to \eqref{eq:ext-property-xvw}.\hfill~ 
\end{proof}

\section{Smallest Non-\dManhattan[2] Profiles}\label{sec:Manhattan-negative}

In this section, we apply the forbidden substructures from \cref{sec:vconfigs} to identify minimally non-\dManhattan[2] profiles.
We show that for $n\in \{3,4,5\}$ voters, the smallest non-\dManhattan[2] profile has $9-n$ alternatives (\cref{thm:no-n3-m6,thm:no-n4-m5,thm:no-n5-m4}).
Each proof proceeds by showing that any hypothetical \dManhattan[2] embedding would necessarily violate the constraints established by our \bet- and \ext-configurations or the \worstinconsistentconfig.
For brevity's sake, given an embedding~$E$ and a voter~$v\in \vvv$ (resp.\ an alternative~$a\in \aaa$), we use boldface~$\pv$ (resp.\ $\pa$) to denote the embedding~$E(v)$ (resp.\ $E(a)$).

\subsection{The Instance with $3$ Voters and $6$ Alternatives}\label{subsec:n3-m6}

Using the \bet- and \ext-configurations from \cref{sec:vconfigs} together with \cref{lem:bet-property,lem:ext-property}, we prove \cref{thm:no-n3-m6} with the help of \cref{ex:no-n3-m6}.
 
\begin{example}  \label{ex:no-n3-m6}
  \setcounter{vthreecounter}{\themyprofilecounter}
  The following profile~$\ppp_{\thevthreecounter}$ with $3$ voters and $6$ alternatives is not \dManhattan[2]. 
  \stepcounter{myprofilecounter}
  \begin{align*}
 \ppp_{\thevthreecounter} \colon   v_1\colon & 1 \succ 2 \succ 3 \succ 4 \succ 5 \succ 6,\\
    v_2\colon & 1 \succ 4 \succ 6 \succ 3 \succ 5 \succ 2,\\
    v_3\colon & 6 \succ 5 \succ 2 \succ 3 \succ 1 \succ 4. 
  \end{align*}
\end{example}
\begin{theorem}\label{thm:no-n3-m6}
  There exists a non-\dManhattan[2] profile with $3$ voters and $6$ alternatives.
\end{theorem}

\begin{proof}
  Consider profile~$\ppp_{\thevthreecounter}$ given in \cref{ex:no-n3-m6}.
  Suppose, towards a contradiction, that $E$ is a \dManhattan[2] embedding for~$\ppp_{\thevthreecounter}$.
  Since each embedding for $3$ voters must satisfy one of the two properties in \cref{def:3voters-configurations}, 
  we distinguish between two cases: there exists a voter who is embedded inside the bounding box of the other two, or there is no such voter.
  \begin{description}
    \item[Case 1:] There exists a voter~$v_i$, $i\in [3]$, such that $E$ satisfies the $v_i$-\bet-property. 
    Since $\ppp_{\thevthreecounter}$ contains a $(v_1,v_2,v_3)$-\bet-configuration wrt.\ $(a,b,x)=(2,6,5)$,
    by \cref{lem:bet-property} it follows that $E$ violates the $v_1$-\bet-property.
    Analogously, since $\ppp$ contains a $(v_2,v_1,v_3)$-\bet-configuration regarding~$a=4,b=2,x=3$, and 
    $(v_3,v_1,v_2)$-\bet-configuration with~$a=5,b=1,x=3$,
    neither does $E$ satisfy the $v_2$-\bet-property or the $v_3$-\bet-property.

    \item[Case 2:] There exists a voter~$v_i$, $i\in [3]$, such that $E$ satisfies the $v_i$-\ext-property.
    Consider the subprofile~$\ppp'$ restricted to the alternatives~$1,2,$ $3,6$.
    We claim that this subprofile contains an \ext-configuration, which by \cref{lem:ext-property} precludes the existence of such a voter~$v_i$ with the $v_i$-\ext-property:
    
    First, since $\ppp'$ contains a $(v_3,v_1,v_2)$-\ext-configuration~(setting $(u,v,w)$ $\coloneqq (v_1,v_3,v_2)$ and $(x,a,b,c,d,e)=(3,2,6,1,1,1)$),
    by \cref{lem:ext-property}, it follows that $E$ violates the $v_3$-\ext-property.
    In fact, $\ppp'$ also contains a $v_2$-\ext-configuration~(setting $(u,v,w)\coloneqq (v_1,v_2,v_3)$ and $(x,a,b,c,d,e)=(3,1,6,2,2,2)$) and a $v_1$-\ext-configuration~(setting $(u,v,w)\coloneqq (v_2,v_1,v_3)$ and $(x,a,b,c,d,e)=(3,1,2,6,6,6)$).
    By \cref{lem:ext-property}, it follows that $E$ violates the $v_2$-\ext-property and the $v_1$-\ext-property.
    \end{description}
  Summarizing, we obtain a contradiction for~$E$. \hfill~ 
\end{proof}

\subsection{The Instance with $4$ Voters and $5$ Alternatives}\label{subsec:n4-m5}

In this section, we show that a profile with $4$ voters and $5$ alternatives may not be \dManhattan[2]. We will achieve this by considering \dMax[2] embeddings since the arithmetic for \dMax[2] is simpler; recall that by \cref{prop:max_man_eq} a profile is \dManhattan[2] if and only if it is \dMax[2]. It is, however, possible to follow similar steps for \dManhattan[2] preferences and obtain an analogous proof.

\begin{example}
  \setcounter{vfourcounter}{\themyprofilecounter}
  The following profile~$\ppp_{\themyprofilecounter}$ with $5$~alternatives contains an \worstinconsistentconfig and will be shown to be not \dMax[2].
\begin{align*}
\ppp_{\themyprofilecounter} \colon  v_1\colon & 1 \succ 2 \succ 3 \succ 4 \succ 5, \\
  v_2 \colon & 1 \succ 2 \succ 3 \succ 5 \succ 4, \\
  v_3\colon & 1 \succ 4 \succ  5 \succ  3 \succ 2, \\
  v_4 \colon & 2 \succ 4 \succ 5 \succ 3 \succ 1.
\end{align*}\stepcounter{myprofilecounter}\label{ex:no-n4-m5}
\end{example}

The proof consists of two main steps:
We first prove that every profile with at least $5$~alternatives which contains an \worstinconsistentconfig is not \dMax[2],
which is obtained via \cref{lem:bb_no_last,lem:5_c_forces_bb}. Then we proceed to show that the example below with $4$ voters and $5$ alternatives is such a profile.

We first show the two lemmas. The first one shows the significance of bounding boxes for \dMax{} embeddings.

\begin{lemma}\label{lem:bb_no_last}
  Let $\ppp$ be a profile admitting a \dMax{} embedding~$E$.
  If $\pz \in \BB(\px, \py)$, then every voter~$v$ satisfies $z \succ_v x$ or $z \succ_v y$.
\end{lemma}

  \begin{proof}
Assume that we have three alternatives $x, y$ and $z$, and a \dMax{} embedding~$E$ such that $\pz \in \BB(\px, \py)$.
Let $v$ be an arbitrary voter.

Consider an arbitrary dimension~$i \in [d]$.
Since $\pz \in \BB(\px, \py)$, we have $\min\{\px[i], \py[i]\} \leq \pz[i] \leq \max\{\px[i], \py[i]\}$.
Let $\pr \coloneqq \argmin_{\{x, y\}} \{\px[i], \py[i]\}$ and $\ps \coloneqq \argmax_{\{x, y\}} \{\px[i], \py[i]\}$; if $\px[i] = \py[i]$, then we set $(\pr, \ps) \coloneqq (\px, \py)$.
Therefore, we have
\begin{equation}\pr[i] \leq \pz[i] \leq \ps[i].\label{eq:maxbetweennes}\end{equation}
We have the following two cases regarding the relative order of $\pv[i]$ and $\pz[i]$:
\begin{description}
  \item[Case 1:] $\pv[i] \leq \pz[i]$.
  Then, $\pv[i] \leq \ps[i]$ and $| \pz[i] - \pv[i]|  = \pz[i] - \pv[i] \stackrel{\eqref{eq:maxbetweennes}}{\leq} \ps[i] - \pv[i] = | \ps[i] - \pv[i]|  \leq \max\{| \px[i] - \pv[i]| , | \py[i] - \pv[i]| \}$.
  \item[Case 2:] $\pv[i] > \pz[i]$.
  Then, $\pv[i] > \pr[i]$ and  $| \pz[i] - \pv[i]|  = \pv[i] - \pz[i] \stackrel{\eqref{eq:maxbetweennes}}{\leq} \pv[i] - \pr[i] = | \pr[i] - \pv[i]|  \leq \max\{| \px[i] - \pv[i]| , | \py[i] - \pv[i]| \}$. 
\end{description}
In both cases, it holds that $| \pz[i] - \pv[i]|  \leq \max\{| \px[i] - \pv[i]| , | \py[i] - \pv[i]| \}$. As this holds for an arbitrary $i \in [d]$, it holds for every $i \in [d]$.
Therefore,
\begin{align*}
\Maxdis{\pv - \pz} = \max_{i \in [d]}| \pz[i] - \pv[i]|  & \leq \max_{i \in [d]} \biggl( \max\{| \px[i] - \pv[i]| , | \py[i] - \pv[i]| \} \biggr)  = \max\{ \Maxdis{\pv - \px}, \Maxdis{\pv - \py}\}.
\end{align*}

This implies that $\Maxdis{\pv - \pz} < \Maxdis{\pv - \px}$ or $\Maxdis{\pv - \pz} < \Maxdis{\pv - \py}$ and thus by the definition of \dMax{}, $z \succ_v x$ or $z \succ_v y$, as desired.
\end{proof}
\begin{remark}
  For $d = 2$, a result equivalent to \cref{lem:bb_no_last} has also been proven independently by Escoffier et al.~\cite{EST2021Euclidlp} for \dManhattan{} embeddings.
  The result is equivalent due to the natural isometry between \dMax[2] and \dManhattan[2] embeddings, see \cref{prop:max_man_eq}. 
\end{remark} 

The next lemma describes a geometrical property for point sets of cardinality at least five.
\begin{lemma}\label{lem:5_c_forces_bb}
For each point set~$\sss$ of $5$~points in $\mathds{R}^{2}$, there must exist three distinct points $\px, \py, \pz \in \sss$ such that $\pz \in \BB(\px, \py)$.
\end{lemma}
\begin{proof}
Assume, towards a contradiction, that $\sss \subset \mathds{R}^{2}$ is a point set with five points~$\pr,\ps,\pt,\pu,\pw$, but it contains no distinct points~$\px, \py, \pz \in \sss$ such that $\pz \in \BB(\px, \py)$.
By renaming we assume that \[\pr[1] \leq \ps[1] \leq \pt[1] \leq \pu[1] \leq \pw[1].\]
Without loss of generality, assume that \[\pr[2] \leq \pw[2];\]
note that if $ \pr[2] > \pw[2]$ we can mirror the embedding by $x$-axis.

For each point~$\pa \in \{\ps, \pt, \pu \}$, we have two options for their relative positions regarding $\pr$ and $\pw$ on axis~$2$: $\pa[2] < \pr[2]$ or $\pa[2] > \pw[2]$.
Note that we cannot have $\pr[2] \leq \pa[2] \leq \pw[2]$, because then we would have $\pr[i] \leq \pa[i] \leq \pw[i]$ for every axis $i \in [2]$, which would imply $\pa \in \BB(\pr, \pw)$,
a contradiction to our assumption.

As we have three remaining alternatives $\ps$, $\pt$, and $\pu$, but two options, at least two of them must satisfy the same option. Let $\{\pa, \pb\} \subset \{\ps, \pt, \pu \}$ such that $\pa[1] \leq \pb[1]$ and $\max\{\pa[2], \pb[2]\} < \pr[2]$ or $\min\{\pa[2], \pb[2]\} > \pw[2]$. We have four cases for the possible relative orders of $\pa, \pb, \pr$ and $\pw$ on axis~$2$, which are also illustrated in \cref{fig:5_c_forces_bb}:
\begin{description}
  \item[Case 1:] $\pa[2] \leq \pb[2] < \pr[2]$.
  Then, it follows that $\pa[i] \leq \pb[i] \leq \pw[i] $ for every axis~$i \in [2] $, which implies that $\pb \in \BB(\pa, \pw)$, a contradiction. 
  \item[Case 2:] $\pb[2] \leq \pa[2] < \pr[2]$.
  Then, it follows that $\pr[1] \leq \pa[1] \leq \pb[1]$ and $\pb[2] \leq \pa[2] < \pr[2] $, which implies that $\pa \in \BB(\pr, \pb)$ a contradiction.
  \item[Case 3:] $\pw[2] < \pa[2] \leq \pb[2]$.
  Then, $\pr[i] \leq \pa[i] \leq \pb[i] $ for every axis~$i \in [2] $, which implies that $\pa \in \BB(\pr, \pb)$, a contradiction.
  \item[Case 4:] $\pw[2] < \pb[2] \leq \pa[2]$.
  Then, $\pa[1] \leq \pb[1] \leq \pw[1]$ and $\pw[2] < \pb[2] \leq \pa[2]$, which implies that $\pb \in \BB(\pa, \pw)$, a contradiction.
\end{description} 

\tikzstyle{bbstyle} = [fill=red!20,draw=red, inner sep=-1pt]

\begin{figure}
\centering
\begin{subfigure}[t]{.24\linewidth}
\begin{tikzpicture}[black]
\basicfigfivepoints

	\foreach \x / \y / \n in {1/-2/a,2/-1/b} {
	\node[voter] at (\x , \y) (\n) {};
	\node[above left = -2 pt and -2 pt of \n] {$\n$};
}

\begin{pgfonlayer}{background}
\node[bbstyle, fit=(a)(w)] {};

\end{pgfonlayer}
	\draw[black] (3.5,-2) -- (3.5,1);
\end{tikzpicture}
\subcaption{Case 1.}
\end{subfigure}
\begin{subfigure}[t]{.24\linewidth}
	\begin{tikzpicture}[black]
		\basicfigfivepoints
		
		\foreach \x / \y / \n in {1/-1/a,2/-2/b} {
			\node[voter] at (\x , \y) (\n) {};
			\node[above left = -2 pt and -2 pt of \n] {$\n$};
		}

		\begin{pgfonlayer}{background}
			\node[bbstyle, fit=(r)(b)] {};
			
		\end{pgfonlayer}
	\draw[black] (3.5,-2) -- (3.5,1);
	\end{tikzpicture}
\subcaption{Case 2.}
\end{subfigure}
\begin{subfigure}[t]{.24\linewidth}
	\begin{tikzpicture}[black]
		\basicfigfivepoints
		
		\foreach \x / \y / \n in {1/2/a,2/3/b} {
			\node[voter] at (\x , \y) (\n) {};
			\node[above left = -2 pt and -2 pt of \n] {$\n$};
		}

		\begin{pgfonlayer}{background}
			\node[bbstyle, fit=(r)(b)] {};
			
		\end{pgfonlayer}
			\draw[black] (3.5,0) -- (3.5,3);
	\end{tikzpicture}
	\subcaption{Case 3.}
\end{subfigure}
\begin{subfigure}[t]{.24\linewidth}
	\begin{tikzpicture}[black]
		\basicfigfivepoints
		
		\foreach \x / \y / \n in {1/3/a,2/2/b} {
			\node[voter] at (\x , \y) (\n) {};
			\node[above left = -2 pt and -2 pt of \n] {$\n$};
		}

		\begin{pgfonlayer}{background}
			\node[bbstyle, fit=(a)(w)] {};
			
		\end{pgfonlayer}

	\end{tikzpicture}
	\subcaption{Case 4.}
\end{subfigure}
\caption{The four cases in the proof of \cref{lem:5_c_forces_bb}.}\label{fig:5_c_forces_bb}
\end{figure}

As all cases lead to a contradiction, our original assumption must have been false. This concludes the proof.
\end{proof}

Now, we are ready to show our second main result.
\begin{theorem}\label{thm:no-n4-m5}
  There exists a non-\dManhattan[2] profile with $4$ voters and $5$ alternatives.
\end{theorem}

\begin{proof}
Suppose, towards a contradiction, that we have a profile $\ppp$ with at least $5$ alternatives $\{a, b, c, d, e\}$ which contains an \worstinconsistentconfig and is \dMax[2] with a \dMax[2] embedding~$E$.

As we have $5$ alternatives, by \cref{lem:5_c_forces_bb} there must be a triple $\{ x, y, z\} \subset \{a, b, c, d, e\}$ such that $\pz \in \BB(\px, \py)$. This together with \cref{lem:bb_no_last} implies that no voter $v$ can satisfy $\{x, y\} \succ_v z$. However, this is a contradiction to our assumption that $\ppp$ contains an \worstinconsistentconfig. Therefore we cannot have a profile $\ppp$ with at least $5$ alternatives which contains an \worstinconsistentconfig and has a \dMax[2] embedding~$E$.

One can verify that profile~$\ppp_{\thevfourcounter}$ given in \cref{ex:no-n4-m5} with $5$ alternatives and~$4$ voters contains an \worstinconsistentconfig,
and is not \dMax[2]: The alternatives~$1, 2, 4$, and $5$ are ranked last by voters~$v_4, v_3, v_2$, and $v_1$, respectively. Therefore we can pick the corresponding voters for every triple involving only the alternatives $1, 2, 4$ and~$5$. It is straightforward to verify that there is a worst-diverse configuration for every triple of alternatives involving $3$ as well.
Thus we have shown that there is a profile with $4$ voters and $5$ alternatives that is not \dMax[2]. By \cref{prop:max_man_eq} it is also not \dManhattan[2].
\end{proof}

\subsection{The Instance with $5$ Voters and $4$ Alternatives}\label{subsec:n5-m4}
In this section, we focus on \cref{thm:no-n5-m4}.
The proof is based on the following example.
\begin{example}  \label{ex:no-n5-m4}
  \setcounter{vfivecounter}{\themyprofilecounter}
  Any profile~$\ppp_{\themyprofilecounter}$ satisfying the following
  is not \dManhattan[2].
  \begin{align*}
 \ppp_{\themyprofilecounter} \colon   v_1\colon & 1 \succ 2 \succ 3 \succ 4,\\
    v_2\colon & 1 \succ 4 \succ 3 \succ 2, \\
    v_3\colon & \{2, 4\} \succ 3 \succ 1,\\
    v_4\colon & 3 \succ 2 \succ 1 \succ 4,\\
    v_5\colon & 3 \succ 4 \succ 1 \succ 2.
  \end{align*}
  \stepcounter{myprofilecounter}
\end{example}
Before we proceed with the proof, %
we show a technical but useful lemma.

\begin{figure}\centering
	
\def\xx{1.5}
\def\xs{0.11}
\def\yy{0.83}
\def\ys{0.11}
\begin{tikzpicture}[black]
	\drawgridA
	\foreach \x / \y / \n / \nn / \typ / \p / \dx / \dy in {2/2/u/u/voter/{above left}/-2/-1, 3/3/v/v/voterV/above left/-1/-1, 4/4/w/w/voter/above left/-1/-1} {
		\node[\typ] at (\x\y) (\n) {};
		\node[\p = \dx pt and \dy pt of \n] {$\nn$};
	}
	\drawreg

	\begin{pgfonlayer}{background}

				        \filldraw[pattern={Lines[
		distance=2mm,
		angle=45,
		line width=0.5mm]},draw=none, pattern color=black!20] (\xx,\yy) -- (\xx,\ys) -- (\xs,\ys) -- (\xs,\yy) -- cycle;	
			\addN

	\end{pgfonlayer}
\end{tikzpicture}
\caption{Embedding for $u, v,$ and $w$ in the proof of \cref{clm:n5-m4}. The striped region is $\dSW(v)$.}\label{fig:n5-m4}
\end{figure}
  \begin{lemma}\label{clm:n5-m4}
    Let $\ppp$ be a  profile with $4$ voters~$u,v,w,r$ and $4$ alternatives~$a$, $b$, $c$, $d$ satisfying the following:
    \begin{align*}
      u\colon & \{a,b\} \succ c \succ d, \\
      v\colon & \{b,d\} \succ c \succ a, \\
      w\colon & \{a,d\} \succ c \succ b, \\
      r\colon & c \succ \{a,b\} \succ d.
    \end{align*}
  If $E$ is a \dManhattan[2] embedding for~$\ppp$ with $\pv \in \BB(\pu, \pw)$, then $\pv\in \BB(\pr, \pw)$.
 \end{lemma}

 \begin{proof}%
   Let $\ppp,u,v,w,r,a,b,c,d,E$ be as defined such that $\pv \in \BB(\pu, \pw)$.
   Without loss of generality assume that $\pu[1]\le \pv[1] \le \pw[1]$ and $\pu[2]\le \pv[2] \le \pw[2]$.
   We divide the two-dimensional space into 16 subspaces, enumerate these regions as in the top left configuration of~\cref{fig:config-v3}, also shown again in \cref{fig:n5-m4} and use $R_i$ to refer to region~$i$, $i\in [16]$. 
   To prove the statement, we will show that if $\pr\in \dNW(\pv)\cup \dSE(\pv)\cup \dNE(\pv)$, 
   then $E$ is not \dManhattan[2].

  Before we proceed, we establish where the individual alternatives can be embedded.
  First, by the preferences of~$u$ and $w$ regarding $c$ and $a$,
  and by \cref{lem:two-votes}\eqref{lem:not-inside}, we obtain that
  $\pc\notin \BB(\pu,\pw)$.
  Further, by the preferences of~$u$ and~$v$ regarding $c$ and $d$ and by \cref{lem:two-votes}\eqref{lem:not-outside-corner},
  we infer that since $\pv \in \dNE(\pu)$, it holds that
  $\pc\notin  \dNE(\pv)$.
  Analogously, due to the preference of $v$ and $w$ regarding $b$ and $c$, we have that since $\pw \in \dSW(\pv)$, it holds that
  $\pc\notin \dSW(\pv)$.
  Together, we infer that $\pc \in R_1\cup R_2\cup R_5\cup R_{12}\cup R_{15}\cup R_{16}$.
  By symmetry, assume that $\pc \in R_1\cup R_2 \cup R_5$, implying that $\pc \in \dNW(\pv)$.

  Similarly, we obtain that $\pa \in \dNW(\pv)\cup \dSE(\pv)$.
  By the preferences of $u,v,w$ regarding $c$ and $a$ and by \cref{lem:bet-property-Ntogether}\eqref{lem:Ntogether2},
  we infer that $\pa \in \dSE(\pv)$ since~$\pc\in \dNW(\pv)$.

  Now, we distinguish between three cases regarding the relative position of voter~$r$.
  \begin{description}
    \item[Case 1:] $\pr \in \dNW(\pv)$. Since $r\colon a \succ d$ and $v\colon d \succ a$, by \cref{lem:two-votes}\eqref{lem:not-outside-corner}, it follows that $\pv\notin \BB(\pr, \pa)$, a contradiction to $\pr \in \dNW(\pv)$ and $\pa \in \dSE(\pv)$.
    \item[Case 2:] $\pr \in \dSE(\pv)$. This case is analogous to the first case. We consider $c$ and $d$ instead. Since $r\colon c \succ d$ and $v\colon d \succ c$, by \cref{lem:two-votes}\eqref{lem:not-outside-corner}, it follows that $\pv\notin \BB(\pr,\pc)$, a contradiction to our assumption as well.
    \item[Case 3:] $\pr\in \dNE(\pv)$. 
    Let us consider alternative~$d$.
    By the preferences of $u$ and $r$, and by \cref{lem:two-votes}\eqref{lem:not-inside},
    we obtain that $\pd\notin \BB(\pu,\pr)$.
    By \cref{lem:two-votes}\eqref{lem:not-outside-corner} (considering the preferences of $u$ and $v$ regarding $c$ and $d$) we infer that since $\pu \in \dSW(\pv)$, it holds that $\pd \notin \dSW(\pu)$.
    Analogously by considering the preferences of $r$ and $v$ regarding $c$ and $d$ we infer that $\pd \notin \dNE(\pr)$.
    Moreover, by \cref{lem:bet-property-Ntogether}\eqref{lem:Ntogether1}
    (considering the preferences of $u,v,r$ regarding $c$ and $d$) and since~$\pc \in \dNW(\pv)$,
    we infer that $\pd\notin  \dNW(\pu)\cup \dNW(\pr)$.
    By \cref{lem:bet-property-Ntogether}\eqref{lem:Ntogether2}
    (considering the preferences of $u,v,r$ regarding $c$ and $d$) and since~$\pc \in \dNW(\pv)$,
   	we further infer that $\pd\notin  \dNW(\pv)$.
   	Hence $\pd \in \dSE(\pu) \cup \dSE(\pr)$.
    
    However, this is a contradiction: Since $v\colon d \succ a$ and $u,r\colon a \succ d$, and~$\pa \in \dSE(\pv)$, by
    \cref{lem:bet-property-Ntogether}\eqref{lem:Ntogether1},
    it follows that $\pd \notin \dSE(\pu)\cup \dSE(\pr)$.
  \end{description}
  Summarizing, this implies that $\pr\in \dSW(\pv)$, and hence $\pv \in \BB(\pr,\pw)$.
\end{proof}

\begin{theorem}\label{thm:no-n5-m4}
  There exists a non-\dManhattan[2] profile with 5 voters and $4$ alternatives.
\end{theorem}

\begin{proof}
  We show that profile~$\ppp_{\thevfivecounter}$ given in \cref{ex:no-n5-m4} is not \dManhattan[2].
  Suppose, towards a contradiction, that $\ppp_{\thevfivecounter}$ admits a \dManhattan[2] embedding~$E$. %
For the sake of brevity, we use $\pc_1, \dots, \pc_5$ to refer to $E(1), \dots, E(5)$.

  First, we observe that one of voters~$v_1$, $v_2$, and $v_3$ is embedded within the bounding box defined by the other two since
   the subprofile of $\ppp_{\thevfivecounter}$ restricted to voters~$v_1$, $v_2$, and $v_3$ is equivalent to profile~$\ppp_{\theextcounter}$ which, by \cref{lem:ext-property}, violates the \ext-property (for each of $v_1$, $v_2$, and $v_3$, respectively).
   We distinguish between two cases.

   \begin{description}
     \item[Case 1:] $\pv_2\in \BB(\pv_1,\pv_3)$ or $\pv_1\in \BB(\pv_2,\pv_3)$.
    Note that these two subcases are equivalent in the sense that if we exchange the roles of alternatives~$2$ and $4$, i.e., 
    $1\mapsto 1$, $3 \mapsto 3$, $2\mapsto 4$, and $4 \mapsto 2$,
    we obtain an equivalent (in terms of the Manhattan property) profile where the roles of voters~$v_1$ and $v_2$ (resp.\ $v_4$ and $v_5$) are exchanged.
    Hence, it suffices to consider the case of $\pv_2\in \BB(\pv_1,\pv_3)$.
    Without loss of generality, assume that $\pv_1[1]\le \pv_2[1]\le \pv_3[1]$ and $\pv_1[2]\le \pv_2[2]\le \pv_3[2]$; see \cref{fig:no-n5-m4-case1v1}.

   \begin{figure}
   \centering
     \captionsetup[subfigure]{justification=centering}
     \begin{subfigure}[t]{.45\linewidth}
       \centering
       \begin{tikzpicture}[black]
         \drawgridB
         \drawregNN

         \foreach \x / \y / \n / \nn / \typ / \p / \dx / \dy in {2/2/v1/{v_1}/voterV/above right/-1/-2, 3/3/v2/{v_2}/voterW/above right/-1/-2, 4/4/v3/{v_3}/voterW/above right/-1/-1} {
           \node[\typ, fill=black,draw=black] at (\x\y) (\n) {};
           \node[\p = \dx pt and \dy pt of \n] {$\nn$};
         }
       \end{tikzpicture}
       \caption{}\label{fig:no-n5-m4-case1v1}
     \end{subfigure}
     \begin{subfigure}[t]{.45\linewidth}
       \centering
       \begin{tikzpicture}[black]
         \drawgridB
         \drawregNN

         \foreach \x / \y / \n / \nn / \typ / \p / \dx / \dy in {2/2/v4/{v_4}/voterV/above right/-1/-2, 3/3/v2/{v_2}/voterW/above right/-1/-2, 4/4/v3/{v_3}/voterW/above right/-1/-1} {
           \node[\typ, fill=black,draw=black] at (\x\y) (\n) {};
           \node[\p = \dx pt and \dy pt of \n] {$\nn$};
         }
       \end{tikzpicture}
       \caption{}\label{fig:no-n5-m4-case1v4}
     \end{subfigure}
     \caption{Illustration of possible embeddings for \cref{thm:no-n5-m4} and for the case where $\pv_2\in \BB(\pv_1,\pv_3)$ (see the left figure).
       We will show that it implies that $\pv_4\in \dSW(\pv_2)$ (see the right figure).}
   \end{figure}

    Then, by \cref{clm:n5-m4} (setting $(u,v,w,r)\coloneqq (v_1,v_2,v_3,v_4)$),
    we obtain that $\pv_2\in \BB(\pv_4,\pv_3)$.
    This implies that $\pv_4[1]\le \pv_2[1]$ and $\pv_4[2]\le \pv_2[2]$; see \cref{fig:no-n5-m4-case1v1}. 

    By the preferences of $v_4$, $v_2$, and $v_3$ regarding alternatives~$2$ and $1$, and by \cref{lem:two-votes}\eqref{lem:Ntogether2},
    it follows that $\pv \notin \BB(\pv_3, 2) \cup \BB(\pv_4,2)$ and hence $\pc_2\in \dNW(\pv_2)\cup \dSE(\pv_2)$. 
    With the same voters and alternatives and \cref{lem:two-votes}\eqref{lem:Ntogether1} we obtain that 
by the preferences of $v_4$, $v_2$, and $v_3$ regarding alternatives~$2$ and $1$, we obtain that $ \pc_1\notin \BB(\pv_3,\pv_4)$.
Combining this and applying \cref{lem:two-votes}\eqref{lem:Ntogether2} again, we obtain that $\pv_3, \pv_4 \notin \BB(\pv_2, \pc_1)$ and hence
    \begin{align}
      \pc_1\notin \BB(\pv_3,\pv_4)\cup \dNE(\pv_3)\cup \dSW(\pv_4).\label{eq:n5-m4:case1-c1}
    \end{align}
    Similarly, regarding the preferences over~$3$ and $1$, it follows that  $\pv_2 \notin \BB(\pv_3, 3) \cup \BB(\pv_4,3)$ and hence $\pc_3\in \dNW(\pv_2)\cup \dSE(\pv_2)$.
    By \cref{lem:bet-property-Ntogether}\eqref{lem:Ntogether2} (considering the preferences of $v_1,v_2$ and $v_3$ regarding alternatives~$2$ and $3$),
    we further infer that either $\pc_2\in \dNW(\pv_2)$ and $\pc_3\in \dSE(\pv_2)$ or  $\pc_2\in \dSE(\pv_2)$ and $\pc_3\in \dNW(\pv_2)$. 
    By symmetry, we only consider the case of $\pc_2\in \dNW(\pv_2)$ and $\pc_3\in \dSE(\pv_2)$. 
    
    On the one hand, by the preferences of $v_3$ and $v_2$ (resp.\ $v_4$ and $v_2$) regarding $1$ and $3$  and by \cref{lem:bet-property-Ntogether}\eqref{lem:Ntogether1} and the fact that $\pc_3\in \dSE(\pv_2)$,
    it follows that $\pc_1\notin  \dSE(\pv_3)$ (resp.\ $\pc_1\notin \dSE(\pv_4)$). 
    On the other hand, by the preferences of~$v_4$ and~$v_2$ (resp.\ $v_3$ and~$v_2$) regarding $1$ and $2$ and by \cref{lem:bet-property-Ntogether}\eqref{lem:Ntogether1} and  $\pc_2\in \dNW(\pv_2)$,
    it follows that $\pc_1\notin  \dNW(\pv_3)$ (resp.\ $\pc_1\notin \dNW(\pv_4)$). 
    Together, this leads to a contradiction to \eqref{eq:n5-m4:case1-c1}.

      \begin{figure}
        \begin{subfigure}[t]{.33\linewidth}
            \centering
            \begin{tikzpicture}[black]
              \drawgridU

	    \foreach \x / \y / \n / \nn / \typ / \p / \dx / \dy in {2/3/v1/{v_1}/voterV/below left/-1/-1,  4/4/v3/{v_3}/voterW/above right/-1/-2, 5/6/v2/{v_2}/voterW/above left/-4/-2} {
	      \node[\typ, fill=black,draw=black] at (\x\y) (\n) {};
	      \node[\p = \dx pt and \dy pt of \n] {$\nn$};
	    }
            
             \begin{pgfonlayer}{background}
               \foreach \s / \t in {13/63, 14/64, 16/66, 41/47, 21/27, 51/57} {
                 \path[draw,lines] (\s) edge (\t);
               }
             \end{pgfonlayer}

	  \end{tikzpicture}
          \caption{}\label{fig:no-n5-m4-case2}
        \end{subfigure}~
        \begin{subfigure}[t]{.33\linewidth}
          \centering
	   \begin{tikzpicture}[black]
	    \drawgridU
	    \drawregU

	    \foreach \x / \y / \n / \nn / \typ / \p / \dx / \dy in {2/3/v1/{v_1}/voterV/below left/-1/-1, 3/2/v4/{v_4}/voterV/below right/-1/-1, 4/4/v3/{v_3}/voterW/above right/-1/-2, 6/5/v5/{v_5}/voterW/below left/-2/-3,5/6/v2/{v_2}/voterW/above left/-4/-2} {
	      \node[\typ, fill=black,draw=black] at (\x\y) (\n) {};
	      \node[\p = \dx pt and \dy pt of \n] {$\nn$};
	    }

	  \end{tikzpicture}
          \caption{}\label{fig:no-n5-m4-case2-v-refined}
        \end{subfigure}~
        \begin{subfigure}[t]{.33\linewidth}
          \centering

	   \begin{tikzpicture}[black]
	    \drawgridU
	    \drawregU

	    \foreach \x / \y / \n / \nn / \typ / \p / \dx / \dy in {2/3/v1/{v_1}/voterV/below left/-1/-1, 3/2/v4/{v_4}/voterV/below right/-1/-1, 4/4/v3/{v_3}/voterW/above right/-1/-2, 6/5/v5/{v_5}/voterW/below right/-2/-3,5/6/v2/{v_2}/voterW/above left/-4/-2} {
	      \node[\typ, fill=black,draw=black] at (\x\y) (\n) {};
	      \node[\p = \dx pt and \dy pt of \n] {$\nn$};
	    }

 \gettikzxy{(v3)}{\xs}{\ys}
	    
	    \gettikzxy{(v1)}{\vx}{\vy}

\gettikzxy{(17)}{\onex}{\oney}
\gettikzxy{(71)}{\threex}{\threey}

	    \gettikzxy{(v4)}{\wx}{\wy}
    
	    \node[] at ($(\vx*0.5+\wx*0.5,\vy*0.6+\wy*0.4)$) (c2) {};
	    \node[right = -5pt of c2] {$\pc_2$};
    
	    \gettikzxy{(v2)}{\vvx}{\vvy}
	    \gettikzxy{(v5)}{\wwx}{\wwy}
    
	    \node[] at ($(\vvx*0.8+\wwx*0.2,\vvy*0.4+\wwy*0.6)$) (c4) {};
	    \node[right = -5pt of c4] {$\pc_4$};
    
	    \node[] at ($(\vvx*0.5+\wwx*0.5,\vy*0.3+\wy*0.7)$) (c3) {};
	    \node[right = -1pt of c3] {$\pc_3$};

	    \node[] at ($(\vx*0.8,\vvy*1.1)$) (c1) {};
	    \node[right = 0pt of c1] {$\pc_1$};

	\begin{pgfonlayer}{background}

          \filldraw[pattern={Lines[
		distance=2mm,
		angle=45,
		line width=0.5mm]},draw=none, pattern color=black!20] (\xs, \ys) -- (\onex, \ys) -- (\onex, \oney) -- (\xs, \oney) -- cycle;

   \filldraw[pattern={Lines[
		distance=2mm,
		angle=45,
		line width=0.5mm]},draw=none, pattern color=black!20] (\xs, \ys) -- (\threex, \ys) -- (\threex, \threey) -- (\xs, \threey) -- cycle;

 \filldraw[pattern={Lines[
		distance=2mm,
		angle=45,
		line width=0.5mm]},draw=none, pattern color=black!20] (\vx,\vy) -- (\vx,\wy) -- (\wx,\wy) -- (\wx,\vy) -- cycle;

 \filldraw[pattern={Lines[
		distance=2mm,
		angle=45,
		line width=0.5mm]},draw=none, pattern color=black!20] (\vvx,\vvy) -- (\vvx,\wwy) -- (\wwx,\wwy) -- (\wwx,\vvy) -- cycle;

	\end{pgfonlayer}
	  \end{tikzpicture}
          \caption{}\label{fig:no-n5-m4-case2-refined}
          \end{subfigure}
	  \caption{Illustration of possible embeddings for \cref{thm:no-n5-m4} and for the case
            where $\pv_3\in \BB(\pv_1,\pv_2)$ (see \cref{fig:no-n5-m4-case2}).
            This case implies that $\pv_1,\pv_4, \in \dSW(\pv_3)$ and
            and $\pv_2, \pv_5 \in \dNE(\pv_3)$ (see \cref{fig:no-n5-m4-case2-v-refined})
            such that $\pc_2\in \dSE(\pv_1)\cap \dNW(\pv_4)$
            and $\pc_4\in \dSE(\pv_2) \cap \dNW(\pv_5)$ (see \cref{fig:no-n5-m4-case2-refined}).}
	  \label{fig:no-n5-m4}
	\end{figure}
        
    \item[Case 2:]  $\pv_3\in \BB(\pv_1,\pv_2)$.
    Without loss of generality, assume that $\pv_1[1]\le \pv_3[1]\le \pv_2[1]$ and $\pv_1[2]\le \pv_3[2]\le \pv_2[2]$; see \cref{fig:no-n5-m4-case2} for an illustration.   
    Then, by \cref{clm:n5-m4} (setting $(u,v,w,r)\coloneqq(v_1,v_3,v_2,v_4)$ and $(u,v,w,r)\coloneqq (v_2,v_3,v_1,v_5)$, respectively),
    we obtain that $\pv_3\in \BB(\pv_4,\pv_2)$ and $\pv_3\in \BB(\pv_5,\pv_1)$.
    This implies that
    \begin{align}
      \pv_4[1]\le \pv_3[1] \text{ and }\pv_4[2]\le \pv_3[2], \text{ and }%
      \pv_5[1]\ge \pv_3[1] \text{ and }\pv_5[2]\ge \pv_3[2].\label{eq:n5-m4-v5-v3}
    \end{align}
    See \cref{fig:no-n5-m4-case2-v-refined} for an illustration.
    
    In the remainder of the proof,
    we will show that we can assume the following relative orientation from each of the four alternatives towards voter~$v_3$; see \cref{fig:no-n5-m4-case2:crossingtwiche-a}:
    \begin{align*}
      \pc_1 \in \dNW(\pv_3),
      \pc_2 \in \dSW(\pv_3),
      \pc_3 \in \dSE(\pv_3), \text{ and }
      \pc_4 \in \dNE(\pv_3).
    \end{align*}
    Moreover, among all four alternatives, alternative~$1$ shall be embedded to the westmost and northmost, 
    while alternative~$3$ to the southmost and $4$ the eastmost. 
    Such conditions are, however, not possible to yield a \dManhattan[2] embedding for voters~$v_2$ and $v_4$ respect to the two pairs~$\{1,2\}$ and $\{3,4\}$.

    Before we formally prove this, we give an intuitive idea.    
    Since voters~$v_1$ and $v_5$ are embedded to the southwest and northeast of $v_3$, respectively, and since both prefer $1\succ 2$ and $3 \succ 4$, but $v_3$ prefers $2\succ 1$ and $4 \succ 3$,
    the bisector between alternatives~$1$ and~$2$ and that between alternatives~$3$ and~$4$ \emph{must} ``cross'' twice; see~\cref{fig:no-n5-m4-case2:crossingtwiche-a}.
    This enforces the relative positions of the four alternatives as described above.
    Analogously, due voters~$v_2$ and~$v_4$, and $v_3$'s preferences over~$\{1,4\}$ and $\{3,2\}$,
    the bisector between alternatives~$1$ and~$4$ and the one between alternatives~$3$ and $2$ must also cross twice.
    This is, however, impossible; see \cref{fig:no-n5-m4-case2:crossingtwiche-b}.
   
   \begin{figure}
     \centering
     \begin{subfigure}[c]{0.48\textwidth}\centering
       \begin{tikzpicture}[scale=.4, black]
         \tkzInit[xmax=10,ymax=8,xmin=0,ymin=-4]
         \begin{pgfonlayer}{background}  
           \tkzGrid[color=gray!20]
         \end{pgfonlayer}

         \node[alter] at (2,6) (c1) {};
         \def \ymax {8}
         
         \def \ofourx {4}
         \def \ofoury {-1}
         
         \def \otwox {1}
         \def \otwoy {-9}
         
         \def \othreex {2}
         \def \othreey {-5}
         
         \node[alter] at ($(c1)+(\otwox,\otwoy)$) (c2) {};
         
         \node[alter,red!40!black] at ($(c1)+(\ofourx,\ofoury)$) (c4) {};
         \node[alter,red!40!black] at ($(c4)+(\othreex,\othreey)$) (c3) {};
         
         \node[left = 0pt of c1] {$\pc_1$};
         \node[above  = 0pt and 0pt of c4] {$\pc_4$};
         \node[right = 0pt of c2] {$\pc_2$};
         \node[right = 0pt of c3] {$\pc_3$};

         \draw[BBstyle] (c1) -- ($(c1)+(\otwox,0)$) -- (c2) -- ($(c1)+(0,\otwoy)$) -- (c1);
         \draw[bisectorstyle, name path=BS12] ($(c1)-(2,\otwox*0.5-\otwoy*0.5)$) --  ($(c1)-(0,\otwox*0.5-\otwoy*0.5)$) -- ($(c2)+(0,\otwox*0.5-\otwoy*0.5)$) -- ($(c2)+(7,\otwox*0.5-\otwoy*0.5)$);

         \draw[BBstyle,red!40!black] (c4) -- ($(c4)+(\othreex,0)$) -- (c3) -- ($(c4)+(0,\othreey)$) -- (c4);
         \draw[bisectorstyle,red!40!black, name path=BS34] ($(c4)-(6,\othreex*0.5-\othreey*0.5)$) -- ($(c4)-(0,\othreex*0.5-\othreey*0.5)$) -- ($(c3)+(0,\othreex*0.5-\othreey*0.5)$) -- ($(c3)+(2,\othreex*0.5-\othreey*0.5)$);

         \path[name path=UP] ($(c1)+(0,2)$) -- ($(c4)+(0,2)$);
         
         \path[name intersections={of = BS12 and BS34, name=I}];

         \begin{pgfonlayer}{background}
           \fill[gray!50, opacity=0.7]  ($(c2)+(0,\otwox*0.5-\otwoy*0.5)$)-- ($(c2)+(0,\otwox*0.5-\otwoy*0.5-0.5)$) --
           ($(c4)-(0,\othreex*0.5-\othreey*0.5)$) -- ($(c4)-(0,\othreex*0.5-\othreey*0.5-0.5)$) -- cycle;

           \filldraw[pattern={Lines[
             distance=1mm,
             angle=135,
             line width=.4mm]}, pattern color=green!50, opacity=0.7]
           ($(c1)-(2,\otwox*0.5-\otwoy*0.5)$) -- ($(c1)-(0,\otwox*0.5-\otwoy*0.5)$) -- (I-1) -- ($(c4)-(6,\othreex*0.5-\othreey*0.5)$) -- cycle;
            \filldraw[pattern={Lines[
             distance=1mm,
             angle=45,
             line width=.4mm]}, pattern color=red!50, opacity=0.7] ($(c1)-(2,\otwox*0.5-\otwoy*0.5)$) -- ($(c1)-(0,\otwox*0.5-\otwoy*0.5)$) -- (I-1) -- ($(c4)-(6,\othreex*0.5-\othreey*0.5)$) -- cycle;

           \filldraw[pattern={Lines[
             distance=1mm,
             angle=135,
             line width=.4mm]}, pattern color=green!50, opacity=0.7]
           ($(c3)+(0,\othreex*0.5-\othreey*0.5)$) -- ($(c3)+(2,\othreex*0.5-\othreey*0.5)$) -- ($(c2)+(7,\otwox*0.5-\otwoy*0.5)$) -- (I-2) -- cycle;
           
           \filldraw[pattern={Lines[
             distance=1mm,
             angle=45,
             line width=.4mm]}, pattern color=red!50, opacity=0.7] 
            ($(c3)+(0,\othreex*0.5-\othreey*0.5)$) -- ($(c3)+(2,\othreex*0.5-\othreey*0.5)$) -- ($(c2)+(7,\otwox*0.5-\otwoy*0.5)$) -- (I-2) -- cycle;
          \end{pgfonlayer}

         \node[voter] at (4.5,1.75) (v) {};
         
         \node[above left = 1pt and 0pt of v,fill=white,inner sep=-1pt,circle] {$\pv_3$};

         \node[voter] at (1.5,1.25) (v1) {};
         
         \node[above = 1pt of v1,fill=white,inner sep=-1pt,circle] {$\pv_1$};
    
         \node[voter] at (8.5,2.5) (v5) {};
         
         \node[above right = 0pt and 0pt of v5,fill=white,inner sep=-1pt,circle] {$\pv_5$};

       \end{tikzpicture}
       \caption{}\label{fig:no-n5-m4-case2:crossingtwiche-a}
     \end{subfigure}
     \begin{subfigure}[c]{0.48\textwidth}\centering
       \begin{tikzpicture}[scale=.4, black]
         \tkzInit[xmax=10,ymax=8,xmin=0,ymin=-4]
         \begin{pgfonlayer}{background}  
           \tkzGrid[color=gray!20]
         \end{pgfonlayer}

         \node[alter] at (2,6) (c1) {};
         \def \ymax {8}
         
         \def \ofourx {4}
         \def \ofoury {-1}
         
         \def \otwox {1}
         \def \otwoy {-9}
         
         \def \othreex {2}
         \def \othreey {-5}
         
         \def \oothreex {5}
         \def \oothreey {3}
         
         \node[alter] at ($(c1)+(\otwox,\otwoy)$) (c2) {};
         
         \node[alter,red!40!black] at ($(c1)+(\ofourx,\ofoury)$) (c4) {};
         \node[alter,red!40!black] at ($(c4)+(\othreex,\othreey)$) (c3) {};
         
         \node[left = 0pt of c1] {$\pc_1$};
         \node[right  = 0pt and 0pt of c4] {$\pc_4$};
         \node[left = 0pt of c2] {$\pc_2$};
         \node[right = 0pt of c3] {$\pc_3$};

         \draw[BBstyle,brown] (c1) -- ($(c1)+(\ofourx,0)$) -- (c4) -- ($(c1)+(0,\ofoury)$) -- (c1);
         \draw[bisectorstyle, name path=BS14,brown] ($(c1)+(\ofourx*0.5-\ofoury*0.5,2)$) --  ($(c1)+(\ofourx*0.5-\ofoury*0.5,0)$) -- ($(c4)-(\ofourx*0.5-\ofoury*0.5,0)$) -- ($(c4)-(\ofourx*0.5-\ofoury*0.5,9)$);

         \draw[BBstyle,darkblue] (c2) -- ($(c2)+(\oothreex,0)$) -- (c3) -- ($(c2)+(0,\oothreey)$) -- (c2);
         \draw[bisectorstyle, name path=BS14,darkblue]  ($(c2)+(\oothreex*0.5+\oothreey*0.5,-1)$) --  ($(c2)+(\oothreex*0.5+\oothreey*0.5,0)$) --  ($(c3)-(\oothreex*0.5+\oothreey*0.5,0)$) --  ($(c3)-(\oothreex*0.5+\oothreey*0.5,-8)$);

         \node[voter] at (3.75,1.75) (v) {};
         
         \node[right = 0pt of v] {$\pv_3$};

       \end{tikzpicture}
       \caption{}\label{fig:no-n5-m4-case2:crossingtwiche-b}
     \end{subfigure}
     \caption{Further illustration for the proof of  \cref{thm:no-n5-m4} where $\pv_3\in \BB(\pv_1,\pv_2)$ (also see \cref{fig:no-n5-m4}).
       Left: The bisector (in green) between alternatives~$1$ and~$2$ and the one (in red) between alternatives~$4$ and~$3$
       \emph{must} ``cross'' twice so we can embed voters~$v_3$, $v_1$, and $v_5$.
       Concretely, $v_3$ will be embedded in the middle gray area, and $v_1$ and $v_5$ in the lower and upper area with hatched pattern, respectively.
       Right: The bisector (in brown) between alternatives~$1$ and~$4$ and the one (in blue) between alternatives~$3$ and~$2$ now cannot cross twice anymore.
       Consequently, it is not \dManhattan[2] for voters~$v_2$ and $v_4$ since they prefer $1\succ 4$ and $3\succ 2$. 
     }
   \end{figure}
   Now, we proceed with the proof of the relative positions of the alternatives.
   \begin{clm}\label{clm:no-n5-m4-case2-relative-pos}
    We can assume that  $\pc_1 \in \dNW(\pv_3)$,
      $\pc_2 \in \dSW(\pv_3)\cap \dSE(\pv_1)\cap \dNW(\pv_4)$,
      $\pc_3 \in \dSE(\pv_3)$, and 
      $\pc_4 \in \dNE(\pv_3)\cap \dSE(\pv_2) \cap \dNW(\pv_5)$.
    \end{clm}
   \begin{proof}[Proof of
    \cref{clm:no-n5-m4-case2-relative-pos}]
    \renewcommand{\qed}{~\hfill~$\diamond$}
    By \cref{lem:two-votes}\eqref{lem:not-inside} (setting $(r,s,x,y)=(v_2,v_3,1,4)$ and $(r,s,x,y)=(v_4,v_3,1,4)$, respectively), we infer that $\pv_3\notin \BB(\pv_2,\pc_1)\cup \BB(\pv_4,\pc_1)$.
    This implies that $\pc_1\in \dNW(\pv_3) \cup \dSE(\pv_3)$.
    By symmetry, we can assume that
    \begin{align}
      \pc_1\in \dNW(\pv_3).\label{eq:n5-m4-c1}
    \end{align} 
    Again, by \cref{lem:two-votes}\eqref{lem:not-inside} (setting $(r,s,x,y)=(v_2,v_3,3,2)$ and $(r,s,x,y)=(v_4,v_3,3,2)$, respectively), we infer that $\pv_3\notin \BB(\pv_2,\pc_3)\cup \BB(\pv_4,\pc_3)$.
    This implies that $\pc_1\in \dNW(\pv_3) \cup \dSE(\pv_3)$.
    Then, by \cref{lem:bet-property-Ntogether}\eqref{lem:Ntogether2} (setting $(r,s,t,x$, $y) = (v_1,v_3,v_2,1,3)$) and by \eqref{eq:n5-m4-c1}, we infer that
    \begin{align}
      \pc_3\in\dSE(\pv_3).\label{eq:n5-m4-c3}
    \end{align}
    It remains to show the relative positions for alternatives~$2$ and $4$.
    By \cref{lem:bet-property-Ntogether}\eqref{lem:Ntogether2} (setting $(r,s,t,x,y)=(v_1,v_3,v_5,1,2)$ and $(r,s,t,x,y)=(v_2,v_3,v_4,3,2)$, respectively),
    we infer that $\pc_2 \notin \dNW(\pv_3)\cup \dSE(\pv_3)$ since $\pc_1\in \dNW(\pv_3)$ and $\pc_3 \in \dSE(\pv_3)$.
    In other words, $\pc_2 \in \dNE(\pv_3)\cup \dSW(\pv_3)$.

    By \cref{lem:two-votes}\eqref{lem:not-inside} (setting $(r,s,x,y)=(v_4,v_5,2,3)$,
    $(r,s,x,y)=(v_1,v_2,2,1)$, $(r,s,x,y) = (v_1,v_5, 2,1)$, $(r,s,x,y) = (v_2,v_5, 2,1)$, respectively),
    we infer that $\pc_2 \notin \BB(\pv_4,\pv_5)\cup \BB(\pv_1,\pv_2)\cup \BB(\pv_1,\pv_5) \cup \BB(\pv_2,\pv_5)$.    
    By \cref{lem:two-votes}\eqref{lem:not-outside-corner} (setting $(r,s,x,y)=(v_3,v_2,2,3)$ and $(r,s,x,y)=(v_3,v_4,2,3)$),
    we infer that $\pv_2,\pv_4\notin \BB(\pv_3,\pc_2)$,
    i.e., $\pc_2 \notin \dNE(\pv_2)\cup \dSW(\pv_4)$.
    Again, by \cref{lem:two-votes}\eqref{lem:not-outside-corner} (setting $(r,s,x,y)=(v_3,v_1,2,1)$ and $(r,s,x,y)=(v_3,v_5,2,1)$),
    we infer that $\pv_1,\pv_5 \notin \BB(\pv_3,\pc_2)$,
    i.e., $\pc_2 \notin \dSW(\pv_1)\cup \dNE(\pv_5)$.
    Analogously, since $v_5\colon 2\succ 3$ and $v_3\colon 2 \succ 3$,
    By \cref{lem:bet-property-Ntogether}\eqref{lem:Ntogether1} (setting $(r,s,x,y)=(v_5,v_3,3,2)$),
    we infer that $\pc_2 \notin \dSE(\pv_5)$ since $\pc_3 \in \dSE(\pv_3)$.
Analogously, since $v_1,v_2,v_5\colon 1\succ 2$, and $v_3\colon 2 \succ 1$, by \cref{lem:bet-property-Ntogether}\eqref{lem:Ntogether1},
    $\pc_2 \notin \dNW(\pv_1)\cup  \dNW(\pv_2)\cup  \dNW(\pv_5)$ since $\pc_1 \in \dNW(\pv_3)$.
    Summarizing, the only region possible for $\pc_2$ is $\dSE(\pv_1)\cap \dNW(\pv_4)$.
    This gives
    \begin{align}\label{eq:n5-m4-c2}
      \pc_2 \in\dSE(\pv_1)\cap \dNW(\pv_4)\cap \dSW(\pv_3). 
    \end{align}
    By exchanging the roles of~$2$ and $4$,
    those of $v_1$ and $v_2$, and those of $v_4$ and $v_5$, we can analogously obtain
    \begin{align}\label{eq:n5-m4-c4}
      \pc_4 \in\dSE(\pv_2)\cap \dNW(\pv_5)\cap \dNE(\pv_3). 
    \end{align}
    See \cref{fig:no-n5-m4-case2-refined} for an illustration.
   \end{proof}
   
 \begin{figure}
     \centering

     \begin{subfigure}[c]{0.23\textwidth}\centering
       \begin{tikzpicture}[scale=.4, black]
         \tkzInit[xmax=7,ymax=5,xmin=0,ymin=-1]
         \begin{pgfonlayer}{background}  
           \tkzGrid[color=gray!20]
         \end{pgfonlayer}

         \node[alter] at (2, 3) (x) {};
         \def \ofx {2}
         \def \ofy {-3}
         
         \node[alter] at ($(x)+(\ofx,\ofy)$) (y) {};
         
         \node[left = 0pt of x] {$\px$};
         \node[right = 0pt of y] {$\py$};

         \draw[BBstyle] (x) -- ($(x)+(\ofx,0)$) -- (y) -- ($(x)+(0,\ofy)$) -- (x);
         
         \draw[bisectorstyle, name path=BSxy, darkgreen] ($(x)-(2,\ofx*0.5-\ofy*0.5)$) -- ($(x)-(0,\ofx*0.5-\ofy*0.5)$)
         -- ($(y)+(0,\ofx*0.5-\ofy*0.5)$) -- ($(y)+(3,\ofx*0.5-\ofy*0.5)$) ;

         \node[voter] at (5.5, 3) (w) {};
         \node[above=0pt of w] {$\pw$};

         \node[voter] at (4.5, 2) (v) {};
         \node[right=0pt of v] {$\pv$};

         \node[voter] at (2.5, 1.5) (u) {};
         \node[above=0pt of u] {$\pu$};

       \end{tikzpicture}
       \caption{}
     \end{subfigure}
     \begin{subfigure}[c]{0.23\textwidth}\centering
       \begin{tikzpicture}[scale=.4, black]
         \tkzInit[xmax=8,ymax=5,xmin=0,ymin=-1]
         \begin{pgfonlayer}{background}  
           \tkzGrid[color=gray!20]
         \end{pgfonlayer}

         \node[alter] at (2, 3) (x) {};
         \def \ofx {4}
         \def \ofy {-2}
         
         \node[alter] at ($(x)+(\ofx,\ofy)$) (y) {};
         
         \node[left = 0pt of x] {$\px$};
         \node[right = 0pt of y] {$\py$};

         \draw[BBstyle] (x) -- ($(x)+(\ofx,0)$) -- (y) -- ($(x)+(0,\ofy)$) -- (x);
         
         \draw[bisectorstyle, name path=BSxy, darkgreen] ($(x)+(\ofx*0.5-\ofy*0.5,2)$) -- ($(x)+(\ofx*0.5-\ofy*0.5,0)$)
         -- ($(y)-(\ofx*0.5-\ofy*0.5,0)$) -- ($(y)-(\ofx*0.5-\ofy*0.5,2)$) ;

         \node[voter] at (6.5, 2.5) (v) {};
         \node[right=0pt of v] {$\pv$};

         \node[voter] at (3, 1.5) (u) {};
         \node[left=-2pt of u] {$\pu$};
       \end{tikzpicture}
       \caption{}
     \end{subfigure}
    \begin{subfigure}[c]{0.24\textwidth}\centering
       \begin{tikzpicture}[scale=.4, black]
         \tkzInit[xmax=6,ymax=5,xmin=0,ymin=-1]
         \begin{pgfonlayer}{background}  
           \tkzGrid[color=gray!20]
         \end{pgfonlayer}

         \node[alter] at (2, 0) (x) {};
         \def \ofx {2}
         \def \ofy {3}
         
         \node[alter] at ($(x)+(\ofx,\ofy)$) (y) {};
         
         \node[left = 0pt of x] {$\py$};
         \node[above = 0pt of y] {$\px$};

         \draw[BBstyle] (x) -- ($(x)+(\ofx,0)$) -- (y) -- ($(x)+(0,\ofy)$) -- (x);
         
         \draw[bisectorstyle, name path=BSxy, darkgreen] ($(x)+(-2,\ofx*0.5+\ofy*0.5)$) -- ($(x)+(0,\ofx*0.5+\ofy*0.5)$)
         -- ($(y)-(0,\ofx*0.5+\ofy*0.5)$) -- ($(y)-(-2,\ofx*0.5+\ofy*0.5)$) ;

       \end{tikzpicture}
       \caption{}
     \end{subfigure}
     \begin{subfigure}[c]{0.23\textwidth}\centering
       \begin{tikzpicture}[scale=.4, black]
         \tkzInit[xmax=8,ymax=5,xmin=0,ymin=-1]
         \begin{pgfonlayer}{background}  
           \tkzGrid[color=gray!20]
         \end{pgfonlayer}

         \node[alter] at (2, 1) (x) {};
         \def \ofx {4}
         \def \ofy {2}
         
         \node[alter] at ($(x)+(\ofx,\ofy)$) (y) {};
         
         \node[left = 0pt of x] {$\py$};
         \node[above = 0pt of y] {$\px$};

         \draw[BBstyle] (x) -- ($(x)+(\ofx,0)$) -- (y) -- ($(x)+(0,\ofy)$) -- (x);
         
         \draw[bisectorstyle, name path=BSxy, darkgreen] ($(x)+(\ofx*0.5+\ofy*0.5,-2)$) -- ($(x)+(\ofx*0.5+\ofy*0.5,0)$)
         -- ($(y)-(\ofx*0.5+\ofy*0.5,0)$) -- ($(y)-(\ofx*0.5+\ofy*0.5,-2)$) ;

       \end{tikzpicture}
       \caption{}
     \end{subfigure}

     \caption{Illustration for \cref{clm:no-n5-m4-case2:two-cross}, assuming that $u,v,w,x,y$ satisfy the premises in the first statement. 
       (a): A possible \dManhattan[2] embedding;
       (b): It is not \dManhattan[2] for voter~$w$ since ``$\py[1]-\px[1] < \px[2]-\py[2]$'' does not hold.
       (c): It is not \dManhattan[2] for voter~$v$ since  ``$\px[1] < \py[1]$'' does not hold.
       (d): It is not  \dManhattan[2] for voter~$v$ since neither ``$\px[1] < \py[1]$'' nor ``$\py[1]-\px[1] < \px[2]-\py[2]$'' holds.
     }\label{fig:clm:relative-pos}
   \end{figure} 

   To formally prove that the relative positions as described in \cref{clm:no-n5-m4-case2-relative-pos} are not \dManhattan[2], we will use the following claim.
   Briefly put, it states that given the premises, the alternative that is less preferred by $v$ should be embedded further away in which both alternatives lie on the same side of $v$.
   Moreover, the shorter side of the bounding box formed by the two alternatives must be along the coordinate where both alternatives lie on the same side of $v$; see \cref{fig:clm:relative-pos} for an illustration of the first case.
   \begin{clm}\label{clm:no-n5-m4-case2:two-cross}
     Let $\ppp$ admit a \dManhattan[2] embedding~$E$.
     For every three voters~$u,v,w$ and two alternatives~$x,y$ such that $u,w\colon y\succ x$, $v\colon y\succ x$,
     $\pw\in \dNE(\pv)$, and $\pu\in \dSW(\pv)$, the following holds:
    \begin{compactenum}[(i)]
      \item\label{case:W} If $\px\in \dNW(\pv)$ and $\py\in \dSW(\pv)$, then $\px[1] < \py[1]$  and $\py[1]-\px[1] < \px[2]-\py[2]$.
      \item\label{case:E} If $\px\in \dSE(\pv)$ and $\py\in \dNE(\pv)$, then $\px[1] > \py[1]$ and $\px[1]-\py[1] < \py[2]-\px[2]$.
     \end{compactenum}
   \end{clm}

   \begin{proof}[Proof of
    \cref{clm:no-n5-m4-case2:two-cross}]
    \renewcommand{\qed}{~\hfill~$\diamond$}
    Let $\ppp,E,u,v,w,x,y$ be as defined.
    We only consider the first statement in details as the other one can be shown by transforming the embedding accordingly. %
    We first show that $\py\in \dSE(\pu)$.
    First, since $u,w\colon x\succ y$, by \cref{lem:two-votes}\eqref{lem:not-inside}, we infer that $\py \notin \BB(\pu,\pw)$.
    This implies that $\py \notin \dNE(\pu)$ since $\py,\pu \in \dSW(\pv)$.
    Secondly, by \cref{lem:two-votes}\eqref{lem:not-outside-corner} (setting $(r,s,x,y)=(v,u,y,x)$),
    we infer that $\pu \notin \BB(\pv, \py)$.
    This implies that $\py \notin \dSW(\pu)$.
    Finally, since $\px \in\dNW(\pv)$, by \cref{lem:bet-property-Ntogether}\eqref{lem:Ntogether1} (setting $(r,s,x,y)=(u,v,x,y)$),
    we infer that $\py\notin \dNW(\pu)$.
    Summarizing, we obtain that
    \begin{align}\label{eq:y-SE-w}
      \py \in \dSE(\pu).
    \end{align}
    
    Now, we proceed to show that $\px[1] < \py[1]$.
    Since $v\colon y\succ x$, implying that $\Mdis{\py-\pv} < \Mdis{\px-\pv}$, 
    we infer by $\px \in \dNW(\pv)$ that
    $(\pv[1]-\py[1])+(\pv[2]-\py[2]) < (\pv[1]-\px[1])+(\px[2]-\pv[2])$, i.e.,
    \begin{align}\label{eq:v-x-y}
    2\pv[2] <  -\px[1]+\px[2]+\py[1]+\py[2].
    \end{align}
    Similarly, since $u\colon x \succ y$, implying that $\Mdis{\px-\pu} < \Mdis{\py-\pu}$,
    we infer that
    $|\px[1]-\pu[1]|+|\px[2]-\pu[2]| \stackrel{\eqref{eq:y-SE-w}}{<} (\py[1]-\pu[1])+(\pu[2]-\py[2])$.
    This further implies that
    $(\px[1]-\pu[1])+(\px[2]-\pu[2]) < (\py[1]-\pu[1])+(\pu[2]-\py[2])$, i.e.,
    \begin{align}\label{eq:w-x-y}
      \px[1]+\px[2]-\py[1]+\py[2] < 2\pu[2].
    \end{align}
    Since $\pu[2] < \pv[2]$, combining \eqref{eq:v-x-y} and \eqref{eq:w-x-y}, we immediately obtain that
    $\px[1] < \py[1]$. 

    It remains to show the last part of the statement.
    Intuitively this means that the distance of $\px$ and $\py$ in the first coordinate must be smaller than that in the second coordinate.
    This is due to voter~$w$'s preferences.
    Since~$w$ prefers~$x\succ y$, implying that $\Mdis{\px-\pw} < \Mdis{\py-\pw}$,
    we infer by $\pw \in \dNE(\pv)$ that
    $(\pw[1]-\px[1])+|\pw[2]-\px[2]| < (\pw[1]-\py[1])+(\pw[2]-\py[2])$.
    This further implies that
    $(\pw[1]-\px[1])+(\pw[2]-\px[2]) < (\pw[1]-\py[1])+(\pw[2]-\py[2])$,
    i.e.,
    $\py[1]-\px[1] < \px[2]-\py[2]$, as desired.
    
   Since rotating and flipping an embedding do not change the \dManhattan[2] property, we can apply the following transformation to show the other statement. 
    We first rotate the embedding by 180 degree and then exchange the roles of $u$ and~$w$.    
  \end{proof}
  In fact, the two cases in \cref{clm:no-n5-m4-case2:two-cross} one-to-one correspond to the two pairs of alternatives~
  $(1,2)$ and $(3,4)$.
   Specifically, by \cref{clm:no-n5-m4-case2:two-cross}\eqref{case:W} (setting $(u,v,w,x,y)=(v_1,v_3,v_5,c_1,c_2)$), we immediately obtain that
  \begin{align}\label{eq:c1c2}
    \pc_1[1]<\pc_2[1].
  \end{align} 
  By  \cref{clm:no-n5-m4-case2:two-cross}\eqref{case:E} (setting $(u,v,w,x,y)=(v_1,v_3,v_5,c_3,c_4)$), we immediately obtain that
  \begin{align}\label{eq:c3c4}
    \pc_3[1]>\pc_4[1].
  \end{align} 
  We show that these two inequalities \eqref{eq:c1c2}--\eqref{eq:c3c4} are not possible to embed both voters~$v_2$ and $v_4$.
  On the one hand, since $v_2$ and $v_4$ prefer $1\succ 4$, implying that
  $\Mdis{\pc_1-\pv_2} < \Mdis{\pc_4-\pv_2}$ and  $\Mdis{\pc_1-\pv_4} < \Mdis{\pc_4-\pv_4}$, 
  by \cref{clm:no-n5-m4-case2-relative-pos}, we infer that
  \begin{alignat}{2}
    & &\Mdis{\pc_1-\pv_2} =   (\pv_2[1]-\pc_1[1])+|\pv_2[2]-\pc_1[2]|  <~ & (\pc_4[1]-\pv_2[1])+(\pv_2[2]-\pc_4[2]) = \Mdis{\pc_4-\pv_2}  \nonumber\\ 
    &\Rightarrow &  (\pv_2[1]-\pc_1[1])+(\pv_2[2]-\pc_1[2]) <~  &(\pc_4[1]-\pv_2[1])+(\pv_2[2]-\pc_4[2])\nonumber \\
     & \Leftrightarrow  & 2\pv_2[1]  <~ & \pc_1[1] + \pc_1[2] + \pc_4[1] - \pc_4[2].\label{eq:v2-c1c4} \\[1ex]
    &&\Mdis{\pc_1-\pv_4}=|\pv_4[1]-\pc_1[1]|+(\pc_1[2]-\pv_4[2]) <~& (\pc_4[1]-\pv_4[1])+(\pc_4[2]-\pv_4[2]) =\Mdis{\pc_4-\pv_4}\nonumber \\ 
    & \Rightarrow &  (\pv_4[1]-\pc_1[1])+(\pc_1[2]-\pv_4[2]) <~& (\pc_4[1]-\pv_4[1])+(\pc_4[2]-\pv_4[2])\nonumber\\
    &\Leftrightarrow & 2\pv_4[1] <~& \pc_1[1] - \pc_1[2] + \pc_4[1]+\pc_4[2]. \label{eq:v4-c1c4}
  \end{alignat}%
  On the other hand, since $v_2$ and $v_4$ prefer $3\succ 2$, implying that  $\Mdis{\pc_3-\pv_2} < \Mdis{\pc_2-\pv_2}$ and $\Mdis{\pc_3-\pv_4} < \Mdis{\pc_2-\pv_4}$, 
  by \cref{clm:no-n5-m4-case2-relative-pos}, we infer that
  \begin{alignat}{2}
    & & \Mdis{\pc_3-\pv_2} = |\pc_3[1]-\pv_2[1]| + (\pv_2[2]-\pc_3[2]) <~& (\pv_2[1]-\pc_2[1]) + (\pv_2[2]-\pc_2[2])=\Mdis{\pc_2-\pv_2}\nonumber\\
  &  \Rightarrow &   (\pc_3[1]-\pv_2[1]) + (\pv_2[2]-\pc_3[2]) <~& (\pv_2[1]-\pc_2[1]) + (\pv_2[2]-\pc_2[2])\nonumber\\
   & \Leftrightarrow &  \pc_2[1]+\pc_2[2]+\pc_3[1]-\pc_3[2] <~& 2\pv_2[1].\label{eq:v2-c2c3}\\[1ex]
    && \Mdis{\pc_3-\pv_4} = (\pc_3[1]-\pv_4[1]) + |\pc_3[2]-\pv_4[2]| <~& (\pv_4[1]-\pc_2[1]) + (\pc_2[2]-\pv_4[2]) = \Mdis{\pc_2-\pv_4}\nonumber\\
    &\Rightarrow &   (\pc_3[1]-\pv_4[1]) + (\pc_3[2]-\pv_4[2]) <~& (\pv_4[1]-\pc_2[1]) + (\pc_2[2]-\pv_4[2])\nonumber\\
   & \Leftrightarrow &  \pc_2[1]-\pc_2[2]+\pc_3[1]+\pc_3[2] <~& 2\pv_4[1].\label{eq:v4-c2c3}
  \end{alignat}  
  Adding up \eqref{eq:v2-c1c4}--\eqref{eq:v4-c2c3}, we obtain that 
  $\pc_2[1]+\pc_3[1]< \pc_1[1]+\pc_4[1]$, a contradiction to \eqref{eq:c3c4}--\eqref{eq:c1c2}.
\end{description}
In summary, we show that it is not possible to find a \dManhattan[2] embedding for profile~$\ppp_{\thevfivecounter}$.
\end{proof}

  \subsection{Tightness: All Smaller Profiles Are \dManhattan[2]}
  \label{sec:experiments}
We complement the non-embeddability results above by showing that all strictly smaller profiles are always \dManhattan[2], establishing a tight characterization.
\begin{proposition}\label{prop:n3-m5+n4-m4}
  If $(n,m)=(3,5)$ or $(n,m)=(4,4)$, then 
  each preference profile with at most $n$ voters and at most $m$ alternatives is \dManhattan[2].
\end{proposition}

\begin{proof}
  Since the Manhattan property is monotone, to show the statement, we only need to look at profiles which have either $3$ voters and $5$ alternatives, or $4$ voters and $4$ alternatives.
  We achieve this by using a computer program employing the CPLEX solver that exhaustively searches for all possible profiles with either $3$ voters and $5$ alternatives, or $4$ voters and $4$ alternatives, and provide a \dManhattan[2] embedding for each of them. 
  Since the CPLEX solver accepts constraints on the absolute value of the difference between any two variables, our computer program is a simple one-to-one translation of the \dManhattan constraints given in \cref{def:embeddings}, without any integer variables. Peters \cite{Peters2017} has noted a similar formulation for \dManhattan embeddings.
  The same program can also be used to show that the preference profiles from the examples \ref{ex:no-n3-m6}, \ref{ex:no-n4-m5} and \ref{ex:no-n5-m4} do not admit a \dManhattan[2] embedding.
  
  Following a similar line as in the work of \citet{ChenGrottke2021}, we did some optimization to significantly shrink the search space on all profiles: We only consider profiles with distinct preference orders and we assume that one of the preference orders is~$1 \succ \dots \succ m$.
  Hence, the number of relevant profiles with $n$ voters and $m$ alternatives is $\binom{m!-1}{n-1}$.
  For $(n,m) = (3,5)$ and $(n,m)=(4,4)$, we need to iterate through $7021$ and $1771$ profiles, respectively.
  We implemented a program which, for each of these produced profiles, uses the IBM ILOG CPLEX optimization software package to check and find a \dManhattan[2] embedding.
  The verification is done by going through each voter's preference order and checking the condition given in \cref{def:embeddings}.
  All generated profiles, together with their \dManhattan[2] embeddings and the distances used for the verification, are available at \url{https://owncloud.tuwien.ac.at/index.php/s/s6t1vymDOx4EfU9}.~ 
\end{proof} 

Combining \cref{prop:n3-m5+n4-m4} with \cref{thm:d=n->dManhattan,thm:d=m-1->Manhattan,thm:no-n3-m6,thm:no-n4-m5,thm:no-n5-m4}, we obtain a complete dichotomy for \dManhattan[2]:
a strict preference profile with $n$ voters and $m$ alternatives is \dManhattan[2] if and only if $n \leq 2$, or $m \leq 3$, or $(n \leq 3$ and $m \leq 5)$, or $(n \leq 4$ and $m \leq 4)$.

\section{Relations to Other Preference Structures}\label{sec:other_domains}
In this section, we discuss how \dManhattan preferences relate to other restricted preference structures. We show that \dManhattanness[2] is not comparable with either \SCness or \SPness. However, \dManhattanness implies \dSPhness[(2^{d-1})] and \dMaxness implies \dSPhness. Moreover, on profiles with three voters, if any two voters are \SP, then the profile is \dManhattan[2].

\begin{definition}
Let $\axb = (\ax_1, \dots, \ax_d)$ be a $d$-tuple of linear orders over the alternatives $\aaa$. For three alternatives $a$, $b$, $c \in \aaa$, we write $a \in \BB(b, c, \axb)$ if $a$ is between $b$ and $c$ on every linear order of $\axb$, i.e., for every $i \in [d]$ it holds that either $a \ax_i b \ax_i c$ or $c \ax_i b \ax_i a$. 
\end{definition}
 
\begin{definition}[\cite{Sui2013Multi,Barbera1993Generalized,elkind2022preference}]\label{def:SP}
Let $\ppp$ be a profile. A voter $v_i \in \vvv$ is \dSPh\footnote{Our definition differs slightly from the definition of \dSPh used by e.g., \citet{Barbera1993Generalized,Sui2013Multi} and coincides with the definition of hereditary \dSPh introduced by \citet{elkind2022preference}.} wrt.\ a $d$-tuple of linear orders $\axb$ if for every $a, b, c \in \aaa$ such that $a \in \BB(b, c, \axb)$, we have that $a \succ_i b$ or $a \succ_i c$.
The profile~$\ppp$ is \dSPh wrt.\ $\axb$ if every voter $v_i \in \vvv$ is \dSPh wrt.\ $\axb$. We say~$\ppp$ is \dSPh if there is a $d$-tuple of linear orders such that $\ppp$ is \dSPh\ wrt.\ it.
\end{definition}
Note that \dSPhness[1] is equivalent to \SPness~\cite{Black1948}. Hence we drop ``1-dimensional'' when we refer to \dSPhness[1].

\begin{definition}\label{def:SC}
A profile $\ppp$ is \SC if there exists a linear order $\ax$ of voters $\vvv$ such that for every pair of alternatives $a, b \in \aaa$, and every triple of voters $v_i, v_j, v_k \in \vvv$ such that $v_i \ax v_j \ax v_k$, if $a \succ_i b$ and $a \succ_k b$, then $a \succ_j b$. In other words, no pair of alternatives may ``cross" more than once.
\end{definition}

\subsection{From Manhattan to Single-Peakedness}
We observe that neither \SPness nor \SCness is a necessary condition of \dManhattanness[2]. However, \dManhattanness implies \dSPhness[2^{d-1}] and \dMaxness implies \dSPhness. %

Our first result relies on the characterizations of \SPness and \SCness\ of \citet{BH11} and \citet{BCW12}, respectively.
\begin{proposition}
There is a \dManhattan[2] profile that is neither \SP nor \SC. Moreover,
\begin{compactenum}[(i)]
\item among all \dManhattan[2] and non-\SP  profiles, a smallest one has either $3$ voters and $3$ alternatives, or $2$ voters and $4$ alternatives, and \label{lem:minspsc1}
\item among all \dManhattan[2] and non-\SC profiles, a smallest one has $3$ voters and $3$ alternatives, and \label{lem:minspsc2}
\item among all \dManhattan[2], non-\SP, and non-\SC  profiles, a smallest one has $3$ voters and $3$ alternatives.\label{lem:minspsc3}
\end{compactenum}
\end{proposition}

\begin{proof}
  By the characterization of the \SPness from \citet{BH11}, we know that
  every minimally non-\SP\ profile consists of either $3$ voters and $3$ alternatives, or $2$ voters and $4$ alternatives.
  By \cref{thm:d=n->dManhattan,thm:d=m-1->Manhattan} every profile with $3$ voters and $3$ alternatives, or $2$ voters and $4$ alternatives is \dManhattan[2]. This proves Statement~\eqref{lem:minspsc1}.

  By the characterization of the \SCness from \citet{BCW12}, we know that
  every smallest non-\SC profiles has $3$ voters and $3$ alternatives.
  Moreover, the following profile is neither \SP nor \SC~\cite{BCW12,BH11}, but it is \dManhattan[2] by \cref{thm:d=m-1->Manhattan}:
\setcounter{vsixcounter}{\themyprofilecounter}
\begin{align*}
\ppp_{\thevsixcounter} \colon & v_1 \colon 1 \succ 2 \succ 3, \\
& v_2 \colon 2 \succ 3 \succ 1, \\
& v_3 \colon 3 \succ 1 \succ 2. 
\end{align*}\stepcounter{myprofilecounter}This proves Statements~\eqref{lem:minspsc2} and~\eqref{lem:minspsc3}.
\end{proof}

Since \dManhattan[2] profiles are \dMax[2], this shows there are profiles that are \dMax but neither \SC nor \SP.
However, we show next that every \dMax profile is \dSPh.

\begin{proposition}\label{theorem:sp_max}
Every \dMax profile is \dSPh.
\end{proposition}

\begin{proof}
Assume that a profile $\ppp = (\vvv, \aaa, \rrr)$ has a \dMax{} embedding $E$.

For every dimension $i \in [d]$, create the linear order $\ax_i$ by ordering the alternatives along their $i$-coordinate. In the case of a tie, order the alternatives arbitrarily. In the resulting linear order $\ax_i$, for every $a_j, a_{k} \in \aaa$, if~$a_j \ax_i a_{k}$ then $E(a_j)[i] \leq E(a_{k})[i]$. Let $\axb = (\ax_1, \dots, \ax_d)$.

Now we show $\ppp$ is \dSPh wrt.\ linear orders $\axb$.
Let $a, b, c \in \aaa$ be an arbitrary triple of alternatives such that $b \in \BB(a, c,$ $\axb)$. By definition for every dimension $i \in [d]$, we have that $E(a)[i] \le E(b)[i] \le E(c)[i]$ or $E(c)[i] \le E(b)[i] \le E(a)[i]$ and thus $E(b) \in \BB(E(a), E(c))$. By \cref{lem:bb_no_last} there cannot be a voter $v_i \in \vvv$ such that $\{a, c\} \succ_i b$. As this holds for an arbitrary triple with $a \axb b \axb c$, $\ppp$ is \dSPh{} wrt.\ $\axb$.
\end{proof}

Unfortunately, we do not know whether every \dManhattan profile is \dSPh. However, we obtain the following weaker implication:

\begin{proposition}\label{theorem:sp_manhattan}
Every  \dManhattan{} profile is \dSPh[2^{d-1}].
\end{proposition}

\newcommand{\manhsum}[1]{S_{#1}}
\begin{proof}
Assume that a profile $\mathcal{P} = (\vvv, \aaa, \rrr)$ has a \dManhattan{} embedding $E$.

The idea of the proof relies on the fact that on \dspace\ under $1$-norm the distance of a point $\px$ from a point $\pp$ is given by $\sum_{i = 1}^d |\px[i] - \pp[i]|$. There are $2^d$ ways to break the absolute values in the formula, giving us $2^d$ possible formulas for the distance, and consequently $2^d$ values that are contributed by $\px$ to the equation. For example, in 2-dimensions, the possible values contributed by $\px$ are given by the following four formulae:
\[
\px[1] + \px[2], \quad \px[1] - \px[2], \quad -\px[1] + \px[2],\quad  \text{ and } \quad-\px[1] - \px[2].
\]
We will show that if the value contributed by one point is between the values contributed by two other points according to each of these formulas, then this point cannot be further from $\pp$ than both of the two other points. This way we obtain~$2^d$ axes.
Moreover, we can observe that each of these formulae has another formula that is its negation and vice versa: for example $\px[1] + \px[2]$ and $-\px[1] - \px[2]$. We only need to keep one formula for each of these pairs, because negating a formula does not change whether one value is between two others according to it. This way we obtain $2^{d-1}$ axes.

We proceed to the formal proof.
Let $\Pi = \{-1, 1\}^{d - 1}$ be the set of all ($d-1$)-dimensional $(-1, 1)$-vectors. We create an axis for each of these vectors. Observe that $|\Pi| = 2^{d-1}$. 

For every $\sigma \in \Pi$, let $\manhsum{\sigma} \colon \vvv \cup \aaa \to \mathds{R}$ with $\manhsum{\sigma}(x) = \left( \sum_{z \in 1}^{d - 1} \sigma[z] \cdot E(x)[z]\right) + E(x)[d]$.
Create axis $\ax_{\sigma}$ by ordering every alternative $i \in \aaa$ non-decreasingly by $\manhsum{\sigma}(i)$. Break the ties arbitrarily.
That is, for every $a_j, a_{k} \in \aaa$, if~$a_j \ax_{\sigma} a_{k}$ then $\manhsum{\sigma}(j) \leq \manhsum{\sigma}(k)$. Let $\axb = (\ax_{\sigma})_{\sigma \in \Pi}$.

Now we show $\ppp$ is \dSPh[2^{d-1}] wrt.\ $\axb$.
Let $a, b, c \in \aaa$ be an arbitrary triple of alternatives such that $b \in \BB(a, c,$ $\axb)$.
Assume, towards a contradiction, that there is a voter $v_i \in \vvv$ such that~$\{a, c\} \succ_i b$. 

We proceed in two cases.
\begin{description}
\item[Case 1:] $E(v)[d] \leq E(b)[d]$.
Consider the following linear order $\sigma' \in \Pi$: For every $z \in [d - 1]$, if $E(b)[z] \geq E(v)[z]$, then $\sigma'[z] \coloneqq 1$,  otherwise $\sigma'[z] \coloneqq - 1$.

We obtain that
\begin{align*}
\Mdis{E(v) - E(b)} &= \sum_{z = 1}^{d}|E(v)[z] - E(b)[z]| 
= \left(\sum_{z = 1}^{d - 1}|E(b)[z] - E(v)[z]|\right) + E(b)[d] - E(v)[d]\\
& = \left(\sum_{z = 1}^{d - 1}\sigma'[z]\cdot (E(b)[z] - E(v)[z]) \right) + E(b)[d] - E(v)[d]
= \manhsum{\sigma'}(b) - \manhsum{\sigma'}(v).
\end{align*}
By assumption that $b \in \BB(a, c, \axb)$, there must be an $x \in \{a, c\}$ such that $\manhsum{\sigma'}(b) \leq \manhsum{\sigma'}(x)$.
We obtain that
\begin{align*}
 \manhsum{\sigma'}(b) - \manhsum{\sigma'}(v) &\leq  \manhsum{\sigma'}(x) - \manhsum{\sigma'}(v)  =\left(\sum_{z = 1}^{d - 1}\sigma'[z]\cdot(E(x)[z] - E(v)[z])\right) + E(x)[d] - E(v)[d]\\
&\leq \left(\sum_{z = 1}^{d - 1}|(E(x)[z] - E(v)[z]|\right) + |E(x)[d] - E(v)[d]|\\
& = \Mdis{E(v) - E(x)}.
\end{align*}
Thus $\Mdis{E(v) - E(b)} \leq \Mdis{E(v) - E(x)}$. Since $E$ is \dManhattan{} embedding, this contradicts $\{a, c\}\succ_i b$.

\item[Case 2:] $E(v)[d] > E(b)[d]$.
Consider the following linear order $\sigma' \in \Pi$: For every $z \in [d - 1]$, if $E(b)[z] \geq E(v)[z]$, then $\sigma'[z] \coloneqq -1$, otherwise $\sigma'[z] \coloneqq 1$. Observe this is the opposite of Case 1.\begin{align*}
\Mdis{E(v) - E(b)} &= \sum_{z = 1}^{d}|E(v)[z] - E(b)[z]| 
= \left(\sum_{z = 1}^{d - 1}|E(b)[z] - E(v)[z]|\right) + E(v)[d] - E(b)[d]\\
& = \left(\sum_{z = 1}^{d - 1}\sigma'[z] \cdot (E(v)[z] - E(b)[z])\right) + E(v)[d] - E(b)[d]
= \manhsum{\sigma'}(v) - \manhsum{\sigma'}(b).
\end{align*}

By construction, there must be an $x \in \{a, c\}$ such that $\manhsum{\sigma'}(b) \geq \manhsum{\sigma'}(x)$.
We obtain that
\begin{align*}
\manhsum{\sigma'}(v) - \manhsum{\sigma'}(b) \leq \manhsum{\sigma'}(v) - \manhsum{\sigma'}(x) 
& = \left(\sum_{z = 1}^{d - 1}\sigma'[z] \cdot (E(v)[z] - E(x)[z])\right) + E(v)[d] - E(x)[d]\\
&\leq \left(\sum_{z = 1}^{d - 1}|(E(x)[z] - E(v)[z]|\right) + |E(x)[d] - E(v)[d]|\\
& = \Mdis{E(v) - E(x)}.
\end{align*}

Thus $\Mdis{E(v) - E(b)} \leq \Mdis{E(v) - E(x)}$. Since $E$ is \dManhattan{} embedding, this contradicts $\{a, c\}\succ_i b$.
\end{description}
As both cases lead to a contradiction, this concludes the proof.
\end{proof}

\subsection{From Single-Crossing and Single-Peakedness to Manhattan}
In this subsection, we study how other restricted preference structures relate to \dManhattanness[2]. %

We first characterize some  of the smallest profiles that are \SC but not \dManhattan[2]. However, we do not know whether there is a \SC profile with $3$ voters that is non-\dManhattan[2].

\begin{example}\label{ex:nfourmsixSC}
\setcounter{vsevencounter}{\themyprofilecounter}
The following $\ppp_{\thevsevencounter}$ with $4$ voters and $6$ alternatives will be shown to be \SC and non-\dManhattan[2].
\begin{align*}
\ppp_{\thevsevencounter} \colon \quad &v_1\colon 1 \succ 2 \succ 3 \succ 4 \succ 5 \succ 6,\\
&v_2\colon 1 \succ 2 \succ 6 \succ 3 \succ 4 \succ 5,\\
&v_3\colon 1 \succ 6 \succ 5 \succ 3 \succ 2 \succ 4,\\
&v_4\colon 6 \succ 5 \succ 4 \succ 3 \succ 2 \succ 1.
\end{align*}\stepcounter{myprofilecounter}
\setcounter{veightcounter}{\themyprofilecounter}
The following $\ppp_{\thevsevencounter}$ with $5$ voters and $5$ alternatives will be shown to be \SC and non-\dManhattan[2].
\begin{align*}
\ppp_{\theveightcounter} \colon \quad &v_1\colon 1 \succ 2 \succ 3 \succ 4 \succ 5,\\
&v_2\colon 1 \succ 2 \succ 3 \succ 5 \succ 4,\\
&v_3\colon 1 \succ 2 \succ 5 \succ 4 \succ 3,\\
&v_4\colon 1 \succ 5 \succ 4 \succ 3 \succ 2,\\
&v_5\colon 5 \succ 4 \succ 3 \succ 2 \succ 1.
\end{align*}

\stepcounter{myprofilecounter}
\end{example}

\begin{proposition}\label{thm:smallest_SC_non2man}
There is a profile that is \SC but non-\dManhattan[2]. Moreover,
\begin{compactenum}[(i)]
  \item among all \SC and non-\dManhattan[2] profiles with $6$ alternatives,
  a smallest one consists of~$4$ voters,\label{thm:smallest_SC_non2man:6alt}
  \item  among all \SC and non-\dManhattan[2] profiles with $5$ alternatives,
  a smallest one consists of $5$ voters, and\label{thm:smallest_SC_non2man:5alt}
\item every \SC\ profile with $4$ alternatives is \dManhattan[2].\label{thm:smallest_SC_non2man:4alt}
\end{compactenum}
\end{proposition}

\begin{proof}

By \citet{BCW12}, the number of voters a \SC profile on $4$ alternatives may have is at most $\frac{(4-1)4}{2} + 1 = 7$.
We verify computationally that all \SC profiles with $4$ alternatives and $7$ voters, $5$ alternatives and $4$ voters, and $6$ alternatives and $3$ voters, are \dManhattan[2]. %
 All generated profiles, together with their \dManhattan[2] embeddings and the distances used for the verification, are available at \url{https://owncloud.tuwien.ac.at/index.php/s/s6t1vymDOx4EfU9}.

We proceed to show the numbers of voters given in Statements~\eqref{thm:smallest_SC_non2man:6alt} and~\eqref{thm:smallest_SC_non2man:5alt}  are indeed minimal:

\mypar{The remaining part of Statement~\eqref{thm:smallest_SC_non2man:6alt}.}
The profile $\ppp_{\thevsevencounter}$ from \cref{ex:nfourmsixSC} is \SC along the linear order $v_1 \ax v_2 \ax v_3 \ax v_4$. However, we will show that it is not \dSPh[2].  By the contrapositive of \cref{theorem:sp_manhattan}, this implies it is not \dManhattan[2].

To show that $\ppp_{\thevsevencounter}$ is not \dSPh[2], assume, towards a contradiction, that $\ppp_2$ is \dSPh[2] wrt.\ a pair of linear orders $(\ax_1, \ax_2)$.

We first observe that the reasoning of \cref{lem:5_c_forces_bb} can be used to show analogous statement on \dSPh[2]  profiles:

\begin{clm}\label{cor:5_c_forces_bb}
For any set of $5$ alternatives $\aaa$ and two linear orders $\ax, \ax'$ of $\aaa$, there must exist three distinct alternatives $a, b, c \in \aaa$ such that $a \in \BB(b,c,(\ax, \ax'))$.
\end{clm}

We can deduce from the definition of \dSPhness[2] (\cref{def:SP}) that for every triple $(a, b, c)$ of alternatives, if there is a voter satisfying $\{b, c\} \succ a$, then $a \notin \BB(b, c, (\ax_1, \ax_2))$. We observe that for every triple in $\{1,2,4,5,6\}$, the only triples $(a,b,c)$ that do not have a voter such that $\{b,c\} \succ a$ are $(2, 1,4)$ and by symmetry $(2,4,1)$. 
Also observe that the alternatives $1,4,5$, and $6$ all have a voter who places them last, so it is sufficient to consider triples where $a = 2$.

Moreover, by \cref{cor:5_c_forces_bb}, there must be three alternatives $a,b,c \in \{1,2,4,5,6\}$ such that $a \in \BB(b,c, (\ax_1, \ax_2))$.
Thus it must be that $2 \in \BB(1,4,(\ax_1, \ax_2))$. Through identical reasoning on the set $\{1,3,4,5,6\}$ we must have $3 \in \BB(1,4,(\ax_1, \ax_2))$.

Without loss of generality, assume that $1 \ax_i \{2,3\} \ax_i 4$ for $i \in [2]$.

Now let us consider the order of $2$ and $3$ on $\ax_1$ and $\ax_2$. If  $1 \ax_i 2 \ax_i 3 \ax_i 4$ for every $i \in [2]$, then $2 \in \BB(1,3,(\ax_1, \ax_2))$ and the preferences of the voter $v_3 \colon \{1, 3\} \succ 2$ lead to a contradiction. If  $1 \ax_i 3 \ax_i 2 \ax_i 4$ for every $i \in [2]$, then $3 \in \BB(1,2, (\ax_1, \ax_2))$ and the preferences of the voter $v_2 \colon \{1, 2\} \succ 3$ lead to a contradiction. Thus one linear order must have $2$ before $3$ and the other $3$ before $2$. Without loss of generality, assume $1 \ax_1 2 \ax_1 3 \ax_1 4$ and $1 \ax_2 3 \ax_2 2 \ax_2 4$. We depict this in \cref{fig:smallest_SC_non2manfig}.

Next, we wish to place alternative~$6$ on the linear orders. Let us proceed with case distinction.

\begin{figure}
\centering
\begin{tikzpicture}[black]
  \drawgridSP
  \foreach \x / \y / \n / \nn / \typ / \p / \dx / \dy in {2/2/a1/1/voter/{above left}/-1/-3, 3/4/a2/2/voter/above left/-1/-3, 4/3/a3/3/voter/above left/-1/-3, 5/5/a4/4/voter/above left/-1/-3, 5/1/a/\ax_1/{}/left/10/10, 1/5/b/\ax_2/{}/below/10/10} {
    \node[\typ] at (\x\y) (\n) {};
    \node[\p = \dx pt and \dy pt of \n] {$\nn$};
  }
  
  \begin{pgfonlayer}{background}
    \foreach \s / \t in {12/62,13/63,14/64,15/65,21/26,31/36,41/46,51/56} {
      \path[draw,lines] (\s) edge (\t);
    }
    \drawreg
  \end{pgfonlayer}
\end{tikzpicture}
\caption{An illustration for the proof of \cref{thm:smallest_SC_non2man}. We have $\ax_1$ on the horizontal axis and $\ax_2$ on the vertical axis.}\label{fig:smallest_SC_non2manfig}
\end{figure}

\begin{description}
  \item[Case 1:  $6 \ax_1 2$.]
  If $6 \ax_2 2$ as well, then $2 \in \BB(4,6,(\ax_1, \ax_2))$, a contradiction to voter $v_4$ satisfying $\{4, 6\} \succ 2$. If $2 \ax_2 6$, then $2 \in \BB(3,6,(\ax_1, \ax_2))$, a contradiction to $v_4$ satisfying $\{3, 6\} \succ 2$.

  \item[Case 2: $2 \ax_1 6 \ax_1 3$.]
  If $6 \ax_2 3$, then $3 \in \BB(4,6,(\ax_1, \ax_2))$, a contradiction to $v_4$ satisfying $\{4, 6\} \succ 3$. If $3 \ax_2 6 \ax_2 2$, we have $6 \in \BB(2,3,(\ax_1, \ax_2))$, a contradiction to $v_1$ satisfying $\{2, 3\} \succ 6$. If $2 \ax_2 6$, we have $2 \in \BB(1,6,(\ax_1, \ax_2))$, a contradiction to $v_3$ satisfying $\{1, 6\} \succ 2$.

  \item[Case 3: $3 \ax_1 6$.]
  If $6 \ax_2 3$, then $3 \in \BB(2,6,(\ax_1, \ax_2))$, a contradiction to $v_2$ satisfying $\{2, 6\} \succ 3$. If $3 \ax_2 6$, then $3 \in \BB(1,6,(\ax_1, \ax_2))$, a contradiction to~$v_2$ satisfying $\{1, 6\} \succ 3$.
\end{description}
As all cases lead to a contradiction, $\ppp_{\thevsevencounter}$ cannot be \dSPh[2].

\mypar{The remaining part of Statement~\eqref{thm:smallest_SC_non2man:5alt}.}
The profile $\ppp_{\theveightcounter}$ from \cref{ex:nfourmsixSC} is not \dManhattan[2] by Theorem 3 from Escoffier et al.~\cite{EST2021Euclidlp} as it has more than $4$ distinct last choices; intuitively, since circles in \dManhattan[2] space are squares parallel to the coordinate axes rotated by 45 degrees, every alternative who is least preferred by one of the voters must be extremal along one the rotated coordinate axes, and there only four possible extrema.
However, the profile is \SC wrt.\ the linear order $v_1 \ax v_2 \ax v_3 \ax v_4 \ax v_5$. 
\end{proof}

The next result compares single-peakedness with \dManhattan[2]{ness}.
\begin{proposition}
There is a profile that is \SP but not \dManhattan[2].
\end{proposition}

\begin{proof}
Consider a \SP profile with $19$ alternatives. Given a \SP order, there are $2^{19-1}=262144$ possible voters with pairwise disjoint preference orders~\cite{Kreweras1963Decisions}. However, the proof of Theorem 4 from Escoffier et al.~\cite{EST2021Euclidlp} implies that a \dManhattan[2] profile with $19$ alternatives has at most $6 \cdot (19 \cdot(19-1)/2)^2 = 175446$ pairwise disjoint preference orders. Thus the \SP profile with $19$ alternatives and all the possible disjoint voters is \SP but not \dManhattan[2].
\end{proof}

We do not know the smallest \SP profile that is not \dManhattan[2]. However, the following theorem implies that such a profile has at least $4$ voters.
\begin{proposition}\label{thm:2sp_d_d-1}
For all profiles with $d \geq 2$ voters, it holds that if two of the voters are \SP, then the profile is \dMax[(d-1)].
\end{proposition}

\begin{proof}

We prove this by induction.

\paragraph*{Base case: $d = 2$.} Chen et al.~\cite{ChenGrottke2021} show that any \SP profile with two voters is \dEuclid[1] and thus \dMax[1]. This proves the base case.

\paragraph*{Inductive step.} Assume that every profile with $n - 1$ voters, two of which are single-peaked, is \dMax[(n - 2)].
We show that any profile with $n$ voters is \dMax[(n - 1)].

Intuitively, we reuse the embedding from the inductive assumption and embed the remaining $n^{th}$ voter in the $(n-1)^{th}$ dimension in such a way that each alternative is equally far from her. Then we slightly tweak the positions of the alternatives in such a way that the previous voters' preferences are unchanged, but we embed the preferences of the~$n^{th}$ voter.

Let $\ppp = (\vvv = \{v_1, \dots, v_n\}, \aaa = [m], \rrr)$ be a profile with $n$ voters. Without loss of generality, assume~$v_1$ and $v_2$ are \SP wrt.\ some linear order. 

By inductive assumption $\ppp$ restricted to the first $n-1$ voters is \dMax[(n-2)]. Let $F$ be the embedding of $(\vvv \setminus \{v_n\}, [m], \{\succ_1, \dots, \succ_{n-1}\})$ to $\mathds{R}^{n-2}$. Without loss of generality, assume that $F(x)[k] \geq 0$ for every $x \in \aaa \cup \vvv \setminus \{v_n\}, k \in [n - 2]$. If not, we can shift the embedding so that this holds.

We choose a constant $\epsilon > 0$ to be smaller than the smallest difference in distances between a voter and two alternatives in $F$, divided by the number of alternatives.
Formally, let $\epsilon < \min_{\substack{i,j \in \aaa\\ v_k \in \vvv \setminus \{v_n\}}}\frac{\displaystyle|\Maxdis{F(i) - F(v_k)} - \Maxdis{F(j) - F(v_k)}|}{\displaystyle m}$. Because the preferences are strict, it holds that $|\Maxdis{F(i) - F(v_k)} - \Maxdis{F(j) - F(v_k)}| > 0$ for every $i,j \in \aaa, v_k \in \vvv \setminus \{v_n\}$.

Now let us define \dMax[(n-1)] embedding $E$ as follows:
For every $i \in \aaa$ let \[E(i) = (F(i)[1], \dots, F(i)[n - 2], \epsilon \cdot \rank_n( i)).\]

For every $k \in [n - 1]$, let \[E(v_k) = (F(i)[1], \dots, F(i)[n - 2], 0).\] Let $M = \max_{\substack{i \in \aaa\\ z \in [n - 2]}}\{F(i)[z]\}$. Let $E(v_n) = (0, \dots, 0, -M)$.

It remains to show that $E$ is an \dMax[(n - 1)]-embedding. Let  $i, j \in \aaa$ be two arbitrary alternatives, and $v_k \in \vvv$ a voter such that $v_k \colon i \succ j$.\\

\noindent\textbf{Case 1:} $v_k \in \{v_1, \dots, v_{n-1}\}$.

Note that for every $i' \in \aaa$ we have that
\begin{align*}
\Maxdis{E(i') - E(v_k)} &= \max_{\ell \in [n - 1]}\{|E(i')[\ell] - E(v_k)[\ell]|\}
= \max\left \{\max_{\ell \in [n - 2]}\{|F(i')[\ell] - F(v_k)[\ell]|\}, | \epsilon \cdot \rank(v_n, i)|\right\}\\
&= \max\{\Maxdis{F(i') - F(v_k)}, \epsilon \cdot \rank_n( i')\}.
\end{align*}

We know by our choice of $\epsilon$ that $\Maxdis{F(j) - F(v_k)} -\Maxdis{F(i) - F(v_k)}> \epsilon m$. 

We proceed in three cases:
\begin{description}
\item[Case 1.a: $\Maxdis{F(i) - F(v_k)} < \epsilon \cdot \rank(v_n, i) $.]
We must have that $\Maxdis{F(j) - F(v_k)} > \epsilon m$ and therefore \[\Maxdis{E(j)- E(v_k)} - \Maxdis{E(i), E(v_k)} = \Maxdis{F(j) - F(v_k)} - \epsilon m > 0,\] as required.
\item[Case 1.b: $\Maxdis{F(j) - F(v_k)} < \epsilon \cdot \rank(v_n, i) $.]
This case is impossible, because $\Maxdis{F(j) - F(v_k)} - \epsilon m> \Maxdis{F(i) - F(v_k)} \geq 0$.
\item[Case 1.c: Otherwise.]
We have
\begin{align*}
\Maxdis{E(j) - &E(v_k)} - \Maxdis{E(i) - E(v_k)} = \Maxdis{F(j) - F(v_k)} - \Maxdis{F(i) - F(v_k)} > 0,
\end{align*} as required,
because $F$ is a \dMax[(n - 2)] embedding.
\end{description}

\noindent\textbf{Case 2:} $v_k = v_n$.

Note that for each $i' \in \aaa$, we have that 
\begin{align*}
\Maxdis{E(i') -  E(v_n)} &= \max_{z \in [n - 1]}[|E(i')[z] - E(v_n)[z]|]
= \max\left\{\max_{z \in [n - 1]}\{|F(i')[z]]\}, M + \epsilon \cdot \rank_n( i')\right\}\\
&= M + \epsilon \cdot \rank_n( i'),
\end{align*}

where the last equality holds by our choice of $M$.
This is clearly linear in the ranks, thus proving the statement.\\

Because $\Maxdis{E(j) - E(v_k)} > \Maxdis{E(i) - E(v_n)}$ holds for every $v_k \in \vvv, i, j \in \aaa$ such that $v_k \colon i \succ j$, $E$ is a \dMax[(n-1)] embedding of $\ppp$.
\end{proof}

\begin{proposition} Every profile with $3$ voters is \dManhattan[2] if two of the voters are \SP. The reverse direction does not hold: there is a \dManhattan[2] profile on $3$ voters, where no pair of voters is \SP.
\end{proposition}

\begin{proof}
\cref{prop:max_man_eq,thm:2sp_d_d-1} directly imply that any profile on three voters, two of which are \SP, is \dManhattan[2].

To see that the reverse direction does not hold, consider the profile
\begin{align*}
v_1 \colon 1 \succ 2 \succ 3 \succ 4 \succ 5,\
v_2 \colon 5 \succ 4 \succ 1 \succ 3 \succ 2,\
v_3 \colon 5 \succ 4 \succ 2 \succ 3 \succ 1.
\end{align*}

By the result of~\citet{BH11} any profile that contains four alternatives $a, b, c, d$ and two voters $u, w$ such that
\begin{align*}
u \colon a \succ b \succ c, d \succ b \text{~ and ~} w \colon c \succ b \succ a, d \succ b
\end{align*}
is not \SP.

To show $v_1, v_2$ are not \SP, rename $(2, 3, 4, 1) \coloneqq (a, b, c, d)$ and $(v_1, v_2) \coloneqq (u,w)$. 
To show $v_1, v_3$ are not \SP, rename $(1, 3, 4, 2) \coloneqq (a, b, c, d)$ and $(v_1, v_3) \coloneqq (u,w)$. 
To show $v_2, v_3$ are not \SP, rename $(1, 3, 2, 5) \coloneqq (a, b, c, d)$ and $(v_2, v_3) \coloneqq (u,w)$.

As this profile has $3$ voters and $5$ alternatives, it is \dManhattan[2] by \cref{prop:n3-m5+n4-m4}.
\end{proof}

\section{Conclusion}\label{sec:conclude}

We initiated the systematic study of \dManhattan{} preferences, focusing on the smallest dimension~$d$ sufficient for a profile to be \dManhattan.
We proved that every profile with $m$ alternatives and $n$ voters is \dManhattan{} whenever $d \ge \min(n, m-1)$, and we determined tight bounds on the smallest non-\dManhattan[2] profiles.

The central technical contribution of this paper is the development of forbidden substructures for \dManhattan[2] preferences: the \bet-configuration, the \ext-configuration, and the \worstinconsistentconfig.
These are, to our knowledge, the first such structural characterizations for \dManhattan[2] preferences.
They describe how specific preference patterns among small sets of voters constrain the geometry of any \dManhattan[2] embedding, and they were the key tools enabling our non-embeddability proofs.

We believe these forbidden substructures have potential beyond the results of this paper.
In particular, characterizing \dManhattan{} profiles through finitely many forbidden subprofiles is an important open problem.
Such characterizations exist for single-peakedness~\citep{BH11} and single-crossingness~\citep{BCW12}, while for \dEuclid[1] preferences a finite forbidden subprofile characterization has been shown to be impossible~\citep{ChePruWoe2017}.
Additionally, the computational complexity of recognizing \dManhattan{} profiles remains open.
While recognizing \dEuclid\ profiles for $d\ge 2$ is $\ETR$-complete~\citep{Peters2017}, recognizing \dManhattan\ preferences is in~NP for fixed~$d$~\citep{Peters2017}.
Our forbidden substructures may be directly useful for constructing gadgets in potential NP-hardness reductions, as they provide concrete preference patterns that force or preclude specific geometric configurations.

Several further questions remain open.
For two-dimensional space ($d=2$), our bounds are tight: there are non-\dManhattan[2] profiles with~$3$ voters and with $4$ alternatives, but we have not established tight bounds for general~$d$.
It is known that for every $d \in \mathds{N}$, there is a non-\dManhattan{} profile with $2^d + 1$ voters and alternatives~\cite{EST2021Euclidlp}, but finding better bounds remains open.
It would also be interesting to extend our results to preferences with ties~\citep{BoLa2006}, to determine the smallest \SP profile that is not \dManhattan[2], and to settle whether for three voters \SC{ness} implies \dManhattan[2]{ness} and whether \dManhattan{ness} implies \dSPh{ness}.
Finally, it remains to be seen whether assuming \dManhattan{} preferences can lower the complexity of computationally hard social choice problems.

\paragraph{Acknowledgments.}
A conference version of this article appeared in the 15th Latin American Symposium 2022.
Jiehua Chen and Sofia Simola are supported by the Vienna Science and Technology Fund (WWTF)~[10.47379/ VRG18012].
Ana{\"i}s Villedieu is supported by the Austrian Science Fund (FWF) under grant P31119.
Markus Wallinger is supported by the Vienna Science and Technology Fund (WWTF) under grant ICT19-035.

\clearpage
\bibliographystyle{plainnat} 
\bibliography{bib}

\clearpage
\begin{table}[t!]
  \centering
  \Large \textbf{\appendixtitle}
\end{table}
\bigskip

\appendix

\appendixtext

\end{document}